\documentclass[11pt]{article}
\usepackage{fullpage}

 \usepackage[T1]{fontenc} %

\usepackage{amssymb}
\usepackage[compact]{titlesec}
\usepackage{amsmath}
\usepackage{amsthm} %
\usepackage{nicefrac}
\usepackage{mathdots,amsthm,url}
\usepackage{ifpdf,color}
\usepackage{graphicx}
\usepackage{subfigure}
\usepackage{enumerate}
\usepackage{xparse}

\usepackage{authblk}

\usepackage{xcolor}

\usepackage[vlined,ruled,linesnumbered]{algorithm2e}
\usepackage[colorlinks=true, allcolors=blue]{hyperref}
\usepackage{cleveref}
\usepackage{thmtools, thm-restate}

\newtheorem{theorem}{Theorem}[section]
\newtheorem*{theorem*}{Theorem}

\newtheorem{proposition}[theorem]{Proposition}
\newtheorem{lemma}[theorem]{Lemma}
\newtheorem*{lemma*}{Lemma}
\newtheorem{corollary}[theorem]{Corollary}

\theoremstyle{definition}

\newtheorem{definition}[theorem]{Definition}

\newcommand{\R}{\mathbb{R}}

\newcommand{\I}{\mathcal{I}}

\newcommand{\flowP}{\operatorname{FlowP}}

\newcommand{\DITC}{ITCO}
\newcommand{\ODITC}{OITCO}

\SetKwProg{myproc}{Procedure}{}{}

\newcommand{\approxFlow}{\codeStyle{ApproxFlow}}
\newcommand{\RemoveRoot}{\codeStyle{RemoveRoot}}
\newcommand{\PostProcess}{\codeStyle{PostProcess}}
\newcommand{\CreateCompact}{\codeStyle{Create\-Aux\-Instance}}
\newcommand{\Initialize}{\codeStyle{Initialize}}
\newcommand{\FastMaxFlow}{\codeStyle{FastMaxFlow}}
\newcommand{\ConvertBasic}{\codeStyle{ConvertBasic}}
\newcommand{\MaxCap}{\codeStyle{MaxCap}}
\newcommand{\Root}{\operatorname{root}}

\newcommand{\SendFlow}{\codeStyle{SendFlow}}
\newcommand{\Saturate}{\codeStyle{Saturate}}
\newcommand{\RerouteS}{\codeStyle{RerouteShortcut}}

\newcommand{\Fe}{F_{\mathrm{e}}}
\newcommand{\Fs}{F_{\mathrm{s}}}

\newcommand{\ab}{\Gamma}

\newcommand{\essentialthreshold}{\Gamma^{-5}\varepsilon}

\newcommand\bi[1]{^{(#1)}}

\newcommand{\Vess}{\roots_{\mathrm{ess}}}

\newcommand{\CP}{\mathcal{P}}
\newcommand{\ACP}{E_{\CP}}
\newcommand{\Eaux}{E_{\mathrm{ess}}}

\newcommand{\tr}{{\operatorname{tr}}}

\newcommand{\Tsolv}{\mathcal{T}_{\mathrm{solv}}}

\newcommand{\din}[2]{\partial_{#1}^{\mathrm{in}}(#2)}
\newcommand{\dout}[2]{\partial_{#1}^{\mathrm{out}}(#2)}
\newcommand{\dtot}[2]{\partial_{#1}(#2)}
\newcommand{\rev}[1]{\overleftarrow{#1}}

\newcommand{\poly}{\mathrm{poly}}
\newcommand{\eps}{\varepsilon} %

\newcommand{\roots}{\mathcal{R}}

\newcommand{\cu}{\bar{u}}
\newcommand{\cG}{\bar{G}}
\newcommand{\cV}{\bar{V}}
\newcommand{\cE}{\bar{E}}
\newcommand{\gV}{V}
\newcommand{\gE}{E}
\newcommand{\gEi}{E^\circ}
\newcommand{\gGi}{G^\circ}
\newcommand{\gVi}{V^\circ}
\newcommand{\ui}{u^{\circ}}
\newcommand{\gEStr}{E_{\mathrm{gap}}}
\newcommand{\gEA}{E_{\mathrm{abd}}}
\newcommand{\gEF}{E_{\mathrm{free}}}
\newcommand{\gEP}{\ACP}
\newcommand{\gET}{\gEA^\tr}
\newcommand{\nr}{r}

\newcommand{\res}[1]{u^{#1}}

\newcommand{\supp}{\mathrm{supp}}

\newcommand{\exc}[2]{\ensuremath{\mathrm{ex}_{#1}({#2})}}

\newcommand{\val}[1]{\operatorname{val}(#1)}

\newcommand{\gEbd}{E_{\mathrm{bound}}}

\newcommand{\defeq}{\stackrel{\scriptscriptstyle{\mathrm{def}}}{=}}

\newcommand{\runtime}{\mathcal{T}}

\newcommand{\Rnn}{\mathbb{R}_{\ge0}}
\newcommand{\Rp}{\mathbb{R}_{>0}}
\newcommand{\Rs}{\mathbb{R}_{+\infty}}

\newcommand{\0}{\mathbf{0}}

\newcommand{\dsStyle}[1]{\mathsf{#1}}
\newcommand{\codeStyle}[1]{\mathsf{#1}}

\newcommand{\ancestor}{\mathrm{anc}}
\newcommand{\descendant}{\mathrm{desc}}
\newcommand{\related}{\mathrm{rel}}

\newcommand{\out}{\mathrm{out}}

\newcommand{\Fone}{F_{\mathrm{one}}}
\newcommand{\Ftwo}{F_{\mathrm{two}}}

\newif\ifnotes\notestrue

\ifnotes
\usepackage{color}
\newcommand{\notename}[2]{{\textcolor{red}{\footnotesize{\bf (#1:} {#2}{\bf ) }}}}

\newcommand{\anote}[1]{{\notename{Aaron}{#1}}}
\newcommand{\sidford}[1]{}
\newcommand{\sidfordLongerTodo}[1]{}
\newcommand{\lnote}[1]{{\notename{Laci}{#1}}}
\newcommand{\dnote}[1]{{\notename{Daniel}{#1}}}
\newcommand{\jnote}[1]{{\notename{Jim}{#1}}}
\else

\newcommand{\notename}[2]{{}}

\newcommand{\anote}[1]{}
\newcommand{\sidford}[1]{}
\newcommand{\sidfordLongerTodo}[1]{}
\newcommand{\lnote}[1]{}
\newcommand{\dnote}[1]{}
\newcommand{\jnote}[1]{}

\fi

\newcommand{\otilde}{\tilde{O}}

\newcommand{\wInstance}{\mathcal{W}}

\newcommand{\treeComplexity}{\mathrm{size}}

\newcommand{\loc}{\mathrm{loc}}

\newcommand{\dsordRebuild}{\dsStyle{Update}}

\newcommand{\dsordNewTransitive}{\dsStyle{NewTransitive}}
\newcommand{\dstreeAddOut}{\dsStyle{AddOut}}
\newcommand{\dstreeAddIn}{\dsStyle{AddIn}}

\newcommand{\dG}{\mathcal{G}}
\newcommand{\dV}{\mathcal{V}}
\newcommand{\dE}{\mathcal{E}}
\newcommand{\dELcl}{\mathcal{E}_*}
\newcommand{\dEL}{\mathcal{E}}
\newcommand{\dELsub}[1]{\mathcal{E}^{\leq #1}}

\newcommand{\madd}{m_{\mathrm{add}}}
\newcommand{\Efinal}{\dE_{\mathrm{final}}}
\newcommand{\Etree}{\dE_{T}^{\mathrm{rel}}}

\newcommand{\canc}{\mathrm{canc}}
\newcommand{\cdesc}{\mathrm{cdesc}}
\newcommand{\Ecov}{\dE_{\mathrm{cov}}}

\newcommand{\Edesc}{\dE_{\descendant}}

\newcommand{\dsTcInit}{\dsStyle{TcInit}}
\newcommand{\dsTcAdd}{\dsStyle{TcAdd}}
\newcommand{\dsTcCover}{\dsStyle{TcCover}}
\newcommand{\dsTcReorder}{\dsStyle{TcReorder}}
\newcommand{\dsTcWit}{\dsStyle{TcWitList}}

\newcommand{\wit}{\mathrm{wit}}
\newcommand{\WitRoute}{\dsStyle{WitRoute}}
\newcommand{\wlistInst}{L}
\newcommand{\wlistSpace}{\mathcal{W}(\dV)}
\newcommand{\walk}{\mathrm{walk}}

\title{From Incremental Transitive Cover to\\Strongly Polynomial Maximum Flow}

\author[1]{Daniel Dadush}
\author[2]{James B. Orlin}
\author[3]{Aaron Sidford} 
\author[4]{L{\'{a}}szl{\'{o}} A. V{\'{e}}gh}
\affil[1]{Centrum Wiskunde \& Informatica and Utrecht University}
\affil[2]{Massachusetts Institute of Technology}
\affil[3]{Stanford University}
\affil[4]{University of Bonn} 

\begin{document}

\maketitle
\thispagestyle{empty}

\begin{abstract}
We provide faster strongly polynomial time algorithms solving maximum flow in structured $n$-node $m$-arc networks. Our results imply an $n^{\omega + o(1)}$-time strongly polynomial time algorithms for computing a maximum bipartite $b$-matching where $\omega$ is the matrix multiplication constant. Additionally, they imply an $m^{1 + o(1)} W$-time algorithm for solving the problem on graphs with a given tree decomposition of width $W$. 

We obtain these results by strengthening and efficiently implementing an approach in Orlin's (STOC 2013) state-of-the-art $O(mn)$ time maximum flow algorithm. We develop a general framework that reduces solving maximum flow with arbitrary capacities to (1) solving a sequence of maximum flow problems with polynomial bounded capacities and (2) dynamically maintaining a size-bounded supersets of the transitive closure under arc additions; we call this problem \emph{incremental transitive cover}. Our applications follow by leveraging recent weakly polynomial, almost linear time algorithms for maximum flow due to Chen, Kyng, Liu, Peng, Gutenberg, Sachdeva (FOCS 2022) and Brand, Chen, Kyng, Liu, Peng, Gutenberg, Sachdeva, Sidford (FOCS 2023), and by developing incremental transitive cover data structures.

\end{abstract}
\newpage
\thispagestyle{empty}
\enlargethispage{4\baselineskip}
\tableofcontents
\newpage
\pagenumbering{arabic}

\section{Introduction}

We consider the maximum flow problem on a directed graph or network $G = (V,E)$ with $n$ nodes $V$ and $m$ edges $E$ along with a source node $s \in V$, sink node $t$, and arc capacities $u \in [0,\infty]^E$. 
The goal in the maximum flow problem is to send as much flow as possible from $s$ to $t$ while satisfying capacity restrictions and maintaining conservation of flow at all nodes other than $s$ and $t$.

The maximum flow problem has a long and rich algorithmic history;  see~\cite{cruz2023survey} for a recent survey. Central to this history is the distinction between \emph{weakly polynomial} and \emph{strongly polynomial} maximum flow algorithms. 
A weakly polynomial algorithm corresponds to the standard notion of a polynomial algorithm in the Turing model: the running time is polynomially bounded in the description length of the instance, including the total bit length of the capacities. In contrast, a strongly polynomial algorithm may only use basic arithmetic operations ($+$, $-$, $\times$, $/$) and comparisons ($\leq$), whose total number is polynomial in $n$ and $m$; and further, it has to run in polynomial space. When the capacity vector is integral, weakly polynomial runtimes typically depend polynomially on $\log U$, where $U = \max_{e \in E, u_e < \infty} u_e$ is the maximum finite capacity. To simplify the discussion of runtimes throughout the paper, we assume that all arithmetic operations require only $O(1)$ time. 

Currently, the fastest weakly polynomial algorithms for maximum flow are either deterministic and run in $m^{1+o(1)} \log U$ time, due tovan den Brand, Chen, Peng, Kyng, Liu, Gutenberg, Sashdeva and Sidford~\cite{brand2023deterministic} or are randomized and run in $(m+ n^{1.5}) (\log n)^{O(1)}\log U$ time, due to van den Brand, Lee, Liu, Saranurak, Sidford, Song, and Wang~\cite{Brand2021} (which is faster on dense instances). The almost-linear deterministic algorithm builds upon the randomized algorithm of Chen, Kyng, Liu, Peng, Gutenberg and Sachdeva~\cite{Chen2022maximum}. Deterministic almost-linear time algorithms for maximum flow in both decremental and incremental settings have also been obtained~\cite{chen2024almost,van2024almost}, where we note that~\cite{chen2024almost,van2024almost} may improve
 upon the $m^{o(1)}$-factor compared to~\cite{Brand2021,brand2023deterministic}. 
These algorithms are all based on highly efficient implementations of interior point methods using dynamic data structures to amortize the cost of the iterations. With these recent breakthroughs, we now have almost optimal maximum flow solvers for networks with polynomially bounded integer capacities (i.e., $U = \poly(n)$).

In contrast, the state-of-the-art strongly polynomial algorithm is $O(nm)$. This is obtained by taking the better of two deterministic algorithms:  the $O(nm\log n / \log(\frac{m}{n\log n}))$ time algorithm of King, Rao and Tarjan \cite{King1994}, which used a sophisticated push-relabel scheme, and the  $O(nm + m^{31/16} \log^2 m)$ time algorithm of Orlin~\cite{Orlin2013}. Taking the best of both algorithm, \cite{King1994} for $m \geq n^{1.06}$ and~\cite{Orlin2013} for $m \leq n^{1.06}$, the running time is $O(nm)$ for all parameter ranges.

Orlin's algorithm can be viewed as a framework for efficiently reducing exact maximum flow to a sequence of approximate maximum flow problems generated using an incremental transitive closure data structure. The approximate maximum flow problem is that of computing a flow which routes at least $1/2$ the maximum flow value. Importantly, this can also be solved in almost-linear time even for general capacities using weakly polynomial algorithms (see \Cref{sec:approx_flow} for more details). A key bottleneck in Orlin's algorithm is Italiano's $O(nm)$ total update time algorithm from 1986 for the incremental transitive closure problem~\cite{Italiano1986}, which remains state-of-the-art. Due to an additional bottleneck in Orlin's framework (see \Cref{sec:max-flow-tc} for more details), an improved incremental transitive closure data structure would not by itself be sufficient to improve his algorithm. We note that $\tilde{O}(nm)$\footnote{We use $\tilde O(.)$ notation to hide $\mathrm{polylog}(n)$ factors.}-time algorithms were known since the early 1980's~\cite{GalilNaaman80, Sleator83}, highlighting the difficulty of speeding up strongly polynomial flow algorithms beyond polylogarithmic factors. A long-standing open problem is to find a $O((nm)^{1-\eps})$-time maximum flow algorithm, for any positive $\eps > 0$.

While progress on the general problem remains elusive, an interesting line of inquiry is to understand which structured classes of maximum flow problems can be solved in strongly-polynomial time faster than $O(nm)$. With the exception of the extensively studied problem of maximum flow on planar networks~\cite{borradaile2009, hassin1985n, henzinger1997faster, itai1979maximum,miller1995flow,weihe1997maximum}, for which $O(n \log n)$-time of Borradaile and Klein~\cite{borradaile2009} is nearly optimal, there has been surprisingly little work in this direction. The goal of this work is to make significant progress on broad and interesting classes of flow problems.

\paragraph{Networks with Few Capacitated Arcs and Nodes.} 
One broad class of flow problems we consider are those involving a limited number of \emph{capacitated arcs}, that is arcs with positive finite capacity. Such networks also easily also encapsulate networks where there are also a limited number of \emph{capacitated nodes}, where we define a capacitated node as any node in which there is a finite capacity on the total flow into or out of that node. A capacity $u_v$ on the flow into (or out of) a node $v$ can be encoded by an arc capacity by splitting $v$ into two nodes $v'$ and $v''$, and then placing a capacity of $u_v$ on arc $(v', v'')$.\footnote{It is also possible to transform an arc capacity on $(i, j)$ into a node capacity.  In this transformation, one first adds a node $v$ to the network and replaces arc $(i, j)$ by uncapacitated arcs $(v, i)$ and $(v, j)$.} For additional details on such transformations, see e.g., Chapter 2 of \cite{amo}.   Correspondingly, our algorithms assume that any node capacities have already been encoded as arc capacities and we let $m_c$ denote the number of nodes plus the number of capacitated arcs. 

The parameter $m_c$ provides a useful way to parameterize the complexity of maximum flow instances. In particular, an important class of problems we highlight are those for which $m_c = O(n)$ (and therefore, the remaining arcs have infinite capacity). This problem class includes a variety of canonical and previously studied problems including the following:

\begin{itemize}
    \item \emph{Maximum node-capacitated flow}: given a maximum flow instance in which there are capacities on the nodes and there are no capacities on arcs, find the maximum flow from $s$ to $t$. Using the aforementioned transformation from node capacities to arc capacities, a maximum node-capacitated flow problem on a $n$-node, $m$-arc network, can be transformed into a maximum flow problem on a $2n$-nodes and $(m+n)$-arc network, of which only $n$ arcs have capacities.
    \item \emph{Maximum $b$-matching}: given a bipartite network on $n$ nodes with node capacities, send as much flow from left to right as possible without exceeding specified node capacities. This is a canonical special case of maximum node-capacitated flow and the reduction to the maximum node-capacitated flow problem simply adds a super source $s$ and sink $t$ to the network with the appropriate uncapacitated arcs from $s$ and to $t$.
    \item \emph{Maximum closure}: given a node weighted directed graph $G=(V,E)$, $w\in\R^{V}$, find a set $S\subseteq V$ maximizing $\sum_{i\in S} w_i$, such that the out-degree of $S$ is zero. The maximum closure problem can be formulated as a minimum cut problem in a flow instance with $O(n)$ finite capacity arcs. An example of the maximum closure problem is the \emph{open pit mining problem}, where the weights correspond to the profit of excavating a region and the arcs represent precedence relations (i.e., $(i,j) \in E$ means $j$ must be excavated before $i$), see \cite{hochbaum2001new,Hochbaum2000,johnson1968}. 
\end{itemize}

The state-of-the-art strongly polynomial maximum flow algorithms for graphs with a small number of capacitated arcs were given by Orlin~\cite{Orlin2013}. Orlin showed that one can significantly beat $O(nm)$ when $m_c/m$ is sufficiently small. Specializing to instances with $m = \Omega(n^2)$ and $m_c = O(n)$ for simplicity, \cite{Orlin2013} gave an $O(n^{8/3})$-time algorithm for these instances which relies on fast matrix multiplication based transitive closure.

\paragraph{Networks with Sublinear Treewidth.}

Another broad class of maximum flow problems we consider are those whose underlying undirected network have sublinear \emph{treewidth}. Treewidth has played a fundamental role in the study of fixed parameter tractability~\cite{cygan2015parameterized} and many interesting classes of graphs have sublinear treewidth. In particular, graphs excluding a fixed minor, such as planar or bounded genus graphs, have treewidth $\tilde{O}(\sqrt{n})$~\cite{alon1990separator}. Furthermore, many maximum flow problems on time-expanded networks naturally have sublinear treewidth (see \Cref{sec:layered-graphs} for more details). For networks of bounded treewidth, Hagerup, Katajainen, Nishimura and Ragde~\cite{Hagerup1998} gave a $O(2^{{\rm tw}(G)^2} n)$-time dynamic programming based algorithm that computes the value of the maximum flow, where ${\rm tw}(G)$ is the treewidth of the network. As most treewidth based algorithms have an exponential dependence on treewidth, Fomin, Lokshtanov, Pilipczuk, Saurabh and Wrochna~\cite{fomin2018fully} initiated the study of algorithms with polynomial dependence on treewidth for the sake of achieving significant speedups for polynomial-time solvable problems. As one such example, they gave an $O(n \cdot {\rm tw}(G)^2 \log n)$ algorithm to compute a maximum number of internally node disjoint paths from $s$ to $t$, corresponding to maximum flow with unit node capacities (note that this is essentially subsumed by the current fastest weakly polynomial solvers). Importantly, tree decompositions of width $O({\rm tw}(G) (\log n)^{3})$ were recently shown to be computable in randomized $m^{1+o(1)}$-time by Bernstein, Gutenberg and Saranurak~\cite{bernstein2022deterministic}, and hence can be leveraged in designing fast algorithms.

\subsection{Our results}
\label{sec:results}

We give substantially improved maximum flow algorithms for networks with a linear number of capacitated arcs or having sublinear treewidth. We achieve our results by making significant improvements to Orlin's transitive closure based framework, designing faster data structures and applying the almost-linear time solver of~\cite{brand2023deterministic}. We detail these results and their applications in this section. In~\Cref{sec:max-flow-tc}, we explain the guarantees of the improved framework, which is based on a new \emph{incremental transitive cover} data structure problem. A technical overview of our algorithm is found given in \Cref{sec:overview}, and a detailed comparison with Orlin's original framework is given in \Cref{sec:orlin-compare}.

\paragraph{\bf Networks with Few Capacitated Arcs and Nodes.} For these networks, we apply fast matrix multiplication, and thus our running times will depend on the matrix multiplication constant $\omega < 2.372$~\cite{duan2023faster,williams2024new}. Our main result is as follows:

\begin{restatable}[Fast-Matrix Multiplication based Maximum Flow]{theorem}{orderedthm}\label{thm:ordered-overall}
There is a deterministic strongly polynomial algorithm that computes a maximum flow in $\tilde{O}(n^{\omega-1}m_c + n m_c^{1+o(1)})$ time for an input instance $(G,u)$ with $n$ nodes, $m$ arcs, and a bound $m_c$ on the number of nodes plus capacitated arcs. Furthermore, using randomization, the running time can be reduced to $\tilde{O}(n^{\omega-1}m_c)$-time.
\end{restatable}

Note that the running time of the deterministic algorithm is $\tilde{O}(n^{\omega-1}m_c)$ if $\omega > 2$. The randomized algorithm achieves this running time unconditionally using the approximate maximum flow solver of~\cite{Brand2021} (this proof leverages that the approximate maximum flow instances solved in our framework are sufficiently dense for these instances). 

Using that node-capacitated flow, $b$-matching, and maximum weight closure reduce to the case $m_c = O(n)$, we have the following direct corollary.

\begin{corollary}
\label{cor:bmatch}
There are deterministic strongly polynomial $n^{\omega+o(1)}$-time algorithms for the maximum node-capacitated flow problem, the maximum bipartite $b$-matching problem, and the maximum weight closure problem on $n$-node networks.
\end{corollary}

As in \Cref{thm:ordered-overall}, the running time of \Cref{cor:bmatch} improves to  $\tilde{O}(n^{\omega})$ if $\omega > 2$ or if we allow randomization. From \Cref{cor:bmatch}, we solve the node-capacitated max flow problem, the maximum $b$-matching problem, and the maximum weight closure problem in almost the current time for computing the transitive closure of a dense $n$-node digraph once. Comparing with~\cite{Orlin2013}, if Orlin had used the almost-linear time solvers, the running time of his algorithm for the the node-capacitated max flow problem, the $b$-matching problem, and the maximum weight closure problem would have improved to $O(n^{1+2\omega/3})$. Our algorithm is faster by a factor of $n^{1-\omega/3} > n^{1/5}$ for the current bound on $\omega$.

\paragraph{\bf Networks with Sublinear Treewidth.} For these networks, our algorithm more directly relies on the notion of tree-depth, which is defined as follows:

\begin{definition}[Tree-depth]\label{def:tree-depth}
The \emph{depth} of a rooted spanning tree is the maximum number of nodes on any path starting at the root.
A graph $G = (V, E)$ is said to have \emph{tree-depth} at most $D$ if there is a rooted spanning tree $T$ of depth $D$ on the same node set such that for every arc $(i, j) \in E$, either $i$ is an ancestor or a descendant of node $j$ in $T$. We call such a $T$ is a \emph{representing tree} for $G$ and we call the minimum $D$ for which $G$ has tree depth $D$ the tree-depth of $G$.
\end{definition}

For any graph $G$ on $n$-nodes, we have ${\rm tw}(G) \leq {\rm td}(G) \leq {\rm tw}(G) \log n$~\cite{bodlaender1995approximating}, where ${\rm tw}(G)$ is the treewidth and ${\rm td}(G)$ is the tree-depth of $G$. 
As before, when working with directed graphs, the quantities above are defined on the underlying undirected version. To compute such trees the randomized $m^{1+o(1)}$-time algorithm of \cite{bernstein2022deterministic} directly computes a tree representing $G$ of depth ${\rm tw}(G)\log^3 n$. The previously mentioned result for treewidth is a direct corollary. 

Our main result for these networks is as follows:

\begin{restatable}[Tree-depth based Maximum Flow]{theorem}{threedepththm}\label{thm:tree-depth-overall}
There is a deterministic algorithm that computes a maximum flow in 
$O\left(m^{1+o(1)} D\right)$ time for an input instance $(G,u)$ with $n$ nodes, $m$ arcs with a given depth-$D$ representing tree for $G$.
\end{restatable}

The runtime in \Cref{thm:tree-depth-overall} is substantially faster than $O(nm)$ as long as tree-depth is at most $n^{1-\varepsilon}$ for any $\varepsilon > 0$. As mentioned above, a suitable representing tree can in fact be computed at the expense of randomization and additional $m^{o(1)}$ factors in the running time. 

As an application of this result, we derive a faster maximum flow algorithm for the family of graphs excluding a $K_r$ minor, the complete graph on $r$ vertices.

\begin{corollary}\label{label:cor-minor}
There is a randomized algorithm that solves maximum flow on networks excluding a $K_r$ minor in $O(r^{3/2} \sqrt{n} m^{1+o(1)})$-time. 
\end{corollary}

\Cref{label:cor-minor} follows directly from well-known bounds in extremal graph theory, namely that the treedwith of a $K_r$ minor free graph is $O(r^{3/2}\sqrt{n})$~\cite{alon1990separator}. Since the number of arcs in an $n$-node  $K_r$ minor free graph is $O(r \sqrt{\log r} n)$~\cite{thomason2001extremal} this running time is also bounded by $O(r^{5/2}\sqrt{\log r}\cdot n^{3/2+o(1)})$. Interestingly, the linear dependence on tree-depth in \Cref{thm:tree-depth-overall} is important for achieving a non-trivial speedup over general networks. 

As a final application in \Cref{sec:tree-depth-proof}, we show that \Cref{thm:tree-depth-overall} also implies speedups for maximum flow problems on time-expanded networks where each arc spans only a short time interval. 

\subsection{Maximum flow via incremental transitive cover}
\label{sec:max-flow-tc}

We now outline our improved framework and provide the guarantees of our framework informally in \Cref{thm:fundamental}. The guarantees of the framework depend on how we solve a particular data structure problem and how we solve approximate flow problems. 

We start by explaining and explaining this data structure. The data structure is closely related to the \emph{incremental transitive closure problem}. 
The transitive closure of a graph $G = (V, E)$ is the graph $G^{\tr} = (V, E^{\tr})$, where $E^{\tr}$ contains all arcs $(i, j) \in \R^{V \times V}$ with $i \neq j$ such that there is a directed path from $i$ to $j$ in $G$. The incremental transitive closure problem \cite{Italiano1986} is to maintain a data structure that supports query access to the adjacency matrix (or adjacency lists) of $G^{\rm tr}$ under $m$ arc insertions. The running time is measured as a function of the number of arc insertions and adjacency queries. One may also require that the data structure support path queries on the underlying graph, which then also contribute to the running time.

In his maximum flow algorithm, Orlin maintains the transitive closure of the subgraph of so-called \emph{abundant arcs}. These arcs are those whose residual capacity (i.e., capacity minus flow 
exceeds the residual flow value (i.e., maximum flow value minus current flow value) by a suitable $\poly(n)$-factor and can be treated as essentially infinity capacity arcs (e.g., deleting the arc capacity does not change the flow value). %

Orlin uses the transitive closure to suitably compress the routing information of these abundant arcs, and thereby minimize the size of the networks on which an approximate maximum flow solver is run. 
We note that for the almost-linear time solvers, only the number of arcs controls the running time of the approximate maximum flow solver. The main operation used for this purpose is to restrict the transitive closure to a subset of nodes $S$ containing the endpoints of all relevant arcs needed to decrease the optimality gap by a $\poly(n)$-factor.

For the sake of speeding up the approximate maximum flow solver, one may hope to encode the reachability information on $S$ using fewer arcs than the direct restriction of the transitive closure to $S$. For this purpose, it may in fact be helpful to use nodes outside of $S$. 
This goal leads us to a concept we call a \emph{transitive cover}. 
For a graph $G = (V, E)$ and a subset of nodes $S \subseteq V$,  we say that $(V', E')$ is a \emph{transitive cover of $S$ with respect to $G$} if it satisfies $S \subseteq V' \subseteq V$ and $(E')^{\tr}[S] = E^{\tr}[S]$.

The \emph{sparsity} of the cover is the number of arcs $|E'|$. As mentioned above, the restriction  $(S,E^{\tr}[S])$ of the transitive closure to $S$ is always a valid transitive cover; however, it may have higher sparsity than the smallest cover. For example, if $V = [2n+1]$, and the nodes $\{1,\dots,n\}$ all have an outgoing arc to node $2n+1$, and if $2n+1$ has outgoing arcs to $\{n+1,\dots,2n\}$, then the restricted closure on $S=\{1,\dots,2n\}$ has $n^2$ arcs (it is complete bipartite) whereas the full graph itself is a cover of sparsity $2n$. More generally, note that any set $S$ of nodes in a $m$ arc network, by taking the smaller of the full graph and restricted closure, the transitive cover sparsity at most $\min \{|S|^2, m\}$. This simple bound is one of the key reasons why Orlin's algorithm spends less than $O(nm)$-time solving approximate maximum flow problems on sparse networks.  Importantly, this bound can be substantially improved on networks with bounded tree-depth. In a network with tree-depth $D$, the sparsity of the minimum transitive cover is at most $2D|S|$, which is derived by retaining $S$ and the nodes on the paths from $S$ the root and only keeping the transitive arcs that are incident to $S$ and that respect the representing tree (see \Cref{sec:data-structure-overview} for more details). Note that this improves on the quadratic bound when $|S| > 2D$.

The above discussion motivates the formal study of the \emph{incremental transitive cover problem}.  In this problem, one must maintain a data structure that supports arc insertions and transitive cover queries on subsets. The running time of the data structure is measured as a function of the number of arcs added and the sizes of the queried sets. Given a set of queries $S_1,\dots,S_k$, we define the \emph{total query size} to be the sum of the queried set sizes $\sum_{i \in [k]} |S_i|$. A key quality measure of the data structure is then the total sparsity of the transitive covers output, that is $\sum_{i \in [k]} |E_i'|$ where $E_1',\dots,E_k'$ are the arc sets outputted on the queries $S_1,\dots,S_k$ respectively. We use this measure to bound the time spent on solving approximate maximum flow problems.

For the remainder of the introduction, we assume that we have a data structure that solves the incremental transitive cover problem (which we call a \DITC{} more compactly). Let ${\rm TC}_{\rm time}(n,g,h)$ denote the time needed by a transitive cover data structure to process $g$ arc insertions and set queries of total size $h$ on a graph with $n$ nodes. Similarly, let ${\rm TC}_{\rm arcs}(n,g,h)$ denote the total sparsity of the outputted transitive covers. (The more precise definition used later in the paper in given in \Cref{sec:analyzing} includes an additional operation which we discuss shortly.) We will later specialize these time and sparsity bounds based on the class of maximum flow instances we work on, as well as the access and update pattern restrictions we can guarantee. We will assume in this work that the transitive cover data structure must output the arcs of the cover one by one. 

Implicit representations of covers will not be used. Under this assumption, the processing time ${\rm TC}_{\rm time}(n,g,h)$ is always greater than or equal to the total sparsity ${\rm TC}_{\rm arcs}(n,g,h)$.

Our main result regarding frameworks for solving maximum flow is that maximum flow can be solved in time that depends only on the cost of implementing a \DITC{} and of solving maximum flow instances approximately. Below in \Cref{thm:fundamental} we state an informal version of this meta theorem when the maximum flow instances are solved approximately using almost-linear-time maximum flow algorithms. 

\begin{theorem}[Meta Theorem (Informal)]
\label{thm:fundamental}
The maximum flow on an input $n$-node, $m$-arc maximum flow instance with at most $m_c$ nodes plus capacitated arcs can be computed in time
\[
{\rm TC}_{\rm time}(n,O(m),O(m_c)) + \left({\rm TC}_{\rm arcs}(n,O(m),O(m_c))+m\right)^{1+o(1)}\,.
\]
\end{theorem}
The formal version of \Cref{thm:fundamental} is provided by~\Cref{thm:maxflow-analysis-main} in \Cref{sec:analyzing}. There, ${\rm TC}_{\rm time}$ is defined to include additional time to handle operations involving a mathematical object that we call a \emph{witness list} and introduce shortly (see also \Cref{sec:trans_cover_first_introduce}). In words, we essentially reduce the maximum flow problem to a transitive cover problem with $O(m)$ insertions and $O(m_c)$ total query size plus approximate maximum flow problems with a total number of arcs equal to the total sparsity of the outputted covers plus one additive factor of $m$. (We always assume $n \leq m$; hence, $m_c = O(m)$.  Recall that $m_c$ is $n$ plus the number of capacitated arcs.)

The $\left({\rm TC}_{\rm arcs}(n,O(m),O(m_c))+m\right)^{1+o(1)}$ term in \Cref{thm:fundamental} is the total time needed to solve the approximate maximum flow problems using the deterministic almost-linear time solver of~\cite{brand2023deterministic}. The $+m$ term inside the parentheses pays for the time to solve a single approximate maximum flow on the full network, as well as the $\tilde{O}(m)$ overhead cost associated with the improved framework. 

A crucial aspect of the framework we have not yet discussed is how to handle flows routed on arcs output by the transitive cover calls, which we call \emph{transitive arcs}. As already mentioned, transitive arcs are used to build what we call \emph{auxiliary networks} for the approximate maximum flow solves. In Orlin's framework, the cost of routing flows from the original network onto these networks and back can be a bottleneck.  The rerouting requires $O(nm_c)$-time over the course of his algorithm. At a high level, the dominant cost is the time needed to route flow from $O(m_c)$ transitive arcs along $O(n)$-length paths in the original network. Note that this can already be higher than the $O(m^{1+o(1)}D)$ running time of our bounded tree-depth algorithm when $m_c = \Omega(m)$. 

To circumvent this bottleneck in our framework, we add these transitive arcs to the original network and only reroute the flow on them at the very end of the algorithm. This change allows us to more effectively amortize the cost of rerouting, and also aids reduce the cost of moving flow onto the auxiliary networks (see \Cref{sec:orlin-compare} for more details). Importantly, it requires significant care to ensure that the flows obtained by this rerouting do not exceed the original arc capacities. To manage the rerouting information, we require additional functionality from the transitive cost data structure than described thus far. Specifically, we rely on the additional maintenance of what we call a \emph{witness list} (see \Cref{def:wit_list}), which contains a witness for each transitive arc.  The witnesses are stored in order of appearance. The time to compute the final rerouting can be made linear in the size of the witness list which is in turn at most linear in the running time of our \DITC{} in each of our applications.

Our formal definition of ${\rm TC}_{\rm time}(n,O(m),O(m_c))$ in \Cref{sec:analyzing} also includes the time to maintain the witness list. In all our implementations, the time to maintain these lists induces at most polylogarithmic overhead (this occurs only in the data structure presented in \Cref{sec:dsord};
see the discussion of 
``Witness List'' there
). 

As mentioned previously, both \Cref{thm:ordered-overall} and \Cref{thm:tree-depth-overall} are derived from suitable \DITC{} implementations combined with \Cref{thm:fundamental}. In these implementations, we leverage the specific restrictions on the updates and queries imposed by the maximum flow framework and the network class itself. In particular, in \Cref{thm:ordered-overall}, we leverage an intrinsic node ordering maintained by the maximum flow algorithm. At a high level, nodes are ordered by the largest residual capacity of a non-abundant arc to which they are incident. Updates and queries are guaranteed to occur in prefixes of this order. For \Cref{thm:tree-depth-overall}, we exploit the fact that the updates will always be compatible with the network's representing tree, which enables a simple and efficient tree-adapted variant of Italiano's algorithm. For additional details on the data structure implementations, please see the overview in \Cref{sec:data-structure-overview}. 

The above two applications exemplify the framework's flexibility, and motivates the search for broader classes of maximum flow instances for which \DITC{} can be implemented in faster than $O(nm)$-time. %

\medskip

\paragraph{New conceptual ingredients}  Our algorithm is inspired by and builds upon ideas of  Orlin’s \cite{Orlin2013} algorithm. 
We follow the high-level approach of his work, namely, reducing maximum flow to a sequence of approximate maximum flow problems on `compact' networks that are constructed with the help of a transitive closure data structure, and charging the running time of these solvers to the number of abundant arcs obtained.

However, there are significant differences and a number of new conceptual ideas that enable us to 
remove a number of bottlenecks inherent in Orlin's algorithm, and to obtain significant running time improvements for special settings; we now highlight three such ideas, and give a more detailed discussion in \Cref{sec:orlin-compare}, where we also provide  running time comparisons.
\begin{enumerate}[(i)]

\item {Extending the instance by transitive closure arcs, and postponing flow rerouting to the final step of the algorithm.}

\item {Maintaining highly structured flow iterates that enable fast auxiliary network constructions.}

\item {Using the more versatile transitive cover data structure instead of transitive closure, and developing a novel variant based on node orderings that relies on fast matrix multiplication.}
\end{enumerate}

\paragraph{Pathway to Faster Maximum Flow?} We conclude this section by explaining the implications of our framework to general maximum flow instances and speculate on possible improvements as well as algorithmic barriers. We first give the maximum flow running time achieved by using an \DITC{} based on Italiano's algorithm \cite{Italiano1986}. (See \Cref{sec:dsclo} for details of the specific \DITC{}.)%

\begin{restatable}[{Italiano Style \DITC{}}]{theorem}{genflowthm}\label{thm:general-flow}
There exists an  \DITC{} that for all $n,m > 1$ has
\[
{\rm TC}_{\rm time}(n,O(m),O(m)) =  O(mn)
\text{ and }
{\rm TC}_{\rm arcs}(n,O(m),O(m)) = O(m^{1.5})\, .
\]
Using that \DITC{}, there is a strongly polynomial $O(mn+m^{1.5+o(1)})$ time maximum flow algorithm.
\end{restatable}

Given \Cref{thm:general-flow}, \DITC{} remains the bottleneck for general maximum flow as long as $m$ is $O(n^{2-\epsilon})$, for any positive $\epsilon > 0$. Comparing to Orlin~\cite{Orlin2013}, a careful analysis of his framework enhanced with the almost-linear approximate maximum flow solver reveals that the time spent on approximate maximum flow computations would be $m^{5/3+o(1)}$ (see \Cref{sec:orlin-compare}) compared to our $m^{3/2+o(1)}$. As mentioned previously, a more severe bottleneck in Orlin's framework is the $O(n m)$-time spent on flow rerouting from the auxilliary network to full network, which we reduce to $O(n^2)$.

An immediate corollary of \Cref{thm:fundamental} is that an \DITC{} with faster runtime than what Italiano's algorithm suggest would directly imply a faster maximum flow algorithm as formalized in \Cref{cor:conditional}. Given the aforementioned bottlenecks, a data structure improvement as in~\Cref{cor:conditional} would not be sufficient to improve maximum flow running times using Orlin's original framework.

\begin{corollary}[Conditional \DITC{} improvement]
\label{cor:conditional}
If for $c \in (0,1]$ and all $n,m > 1$, there exists an  \DITC{} satisfying
${\rm TC}_{\rm time}(n,O(m),O(m)) = O(mn^{1-c})$,
then there exists an $m^{1+o(1)} n^{1-c}$ strongly-polynomial time maximum flow algorithm.
\end{corollary}

Importantly, the online Boolean matrix multiplication conjecture (OMv conjecture) by Henzinger, Krinninger, Nanongkai and Saranurak~\cite{henzinger2015unifying} (see \Cref{sec:omv}) gives conditional lower bounds that restricts the range of possible improvements for \DITC.  In the context of incremental transitive closure, they prove (see \cite[Corollary 4.8]{henzinger2015unifying}) that any algorithm that can process $O(m)$ updates and $O(n^2)$ queries starting from any empty $n$-node graph in $O(n^{1-\varepsilon} m)$-time, for any $\varepsilon > 0$, would break the OMv conjecture. Additionally, a generalization of the OMV conjecture to a broader class of not necessarily Boolean matrices was proven in \cite{CliffordGL15}.%

We note that transitive cover can simulate transitive closure queries by querying on sets of size $2$ and performing a BFS. Since the processing time dominates the total sparsity of the outputted covers, and hence also the BFS time, the same lower bound holds for transitive cover when replacing the number of queries by the total size of the queried sets.
However, the total query size in \Cref{cor:conditional} is only $O(m)$, which can be much smaller than $O(n^2)$ for sparse graphs. By optimizing their lower construction for $O(m)$ insertions and queries, we only obtain a lower bound of $\Omega(m^{3/2})$ (see \Cref{sec:omv} for details). Surprisingly, this is precisely the time our current framework spends on approximate maximum flow calls.

Given the above, it remains unclear whether one can rule out the conditional improvement in~\Cref{cor:conditional} for sparse graphs based on OMv alone. As already mentioned, the queries and updates made by our framework have non-trivial additional structure that can be exploited in the cases explored above. We leave it as an interesting research direction to explore the possible landscape of lower bounds for \DITC{} and their implications for maximum flow.

\subsection{Overview of the algorithm}
\label{sec:overview}

Here we give a high-level overview of our algorithms and  algorithmic contributions of our paper. For formal definitions of concepts mentioned in this overview, see  \Cref{sec:prelim}, \Cref{sec:trans_cover_first_introduce}, and \Cref{sec:maxflow_framework}.

For a flow $f$ in a maximum flow instance $(G,u)$, let $\res{f}_e$ denote the residual capacity of arc $e$, and let $\nu(G,u)$ denote the maximum flow and minimum cut value.
We assume the graph is symmetric, i.e., the reverse $\rev{e}=(j,i)$ of every arc $e=(i,j)$ is also present (possibly with 0 capacity).
Given the input instance $(\gGi,\ui)$, at a given iteration, the algorithm may create an extended graph $(G,u)$ by adding new arcs to the graph with $\infty$ capacity.  The newly added arcs are guaranteed not to be in any minimum cut, ensuring $\nu(G,u)=\nu(\gGi,\ui)$; i.e., the maximum flow value is preserved.  It also ensures that any minimum cut in the extended graph is also a minimum cut in the input instance.
The algorithm finds a maximum flow and a minimum  $s$--$t$ cut $S$ in the extended graph $(G,u)$. By rerouting flow from added arcs to original ones, the final maximum flow $f$ in the extended graph $(G,u)$ is converted to a maximum flow $f^\star$ in $(\gGi,\ui)$.
\footnote{Note that while given a maximum flow, a minimum cut can be easily found using a breadth first search in the residual graph, there is no known nearly-linear-time reduction in the other direction.}.

Our algorithm belongs to the broad family of scaling algorithms. At each iteration, we maintain a flow $f$ in the current instance $(G,u)$, and a bound $\varepsilon$ such that $f$ is an $\varepsilon$-optimal flow; that is, the flow value is $\val{f}\ge \nu(G,u)-\varepsilon$. The accuracy $\varepsilon$ decreases by at least a factor $\ab^{3}$ at each iteration for the parameter $\ab=4n^2$.
At a high-level, the main progress in our algorithm is then gradually learning  redundant arc capacities by calling approximate flow solvers on suitable auxiliary instances in every iteration. 
In an $\varepsilon$-optimal flow, we have $\nu(G,\res{f})\le \varepsilon$. Therefore,  an arc with $\res{f}_e>\varepsilon$ cannot cross any minimum cut. Similarly to Orlin's algorithm \cite{Orlin2013}, we call arcs with high residual flow value \emph{abundant}. For technical reasons, we use a higher threshold, and  call arcs $e$ with $\res{f}_e\ge\ab \varepsilon $ abundant for $\ab$ as above. %
Once an arc becomes abundant, it remains so for the rest of the algorithm. Arcs $e=(i,j)$ such that the reverse arc $\rev{e}=(j,i)$ is also abundant are called \emph{free}.

In an $\varepsilon$-optimal flow $f$, we define the set of \emph{free components} $\CP=\CP_{f,\varepsilon}$ as {\em (a)} the set of nodes reachable from $s$ on abundant paths, denoted as $P_s$; {\em (b)} the set of nodes that can reach $t$ on abundant paths, denoted as $P_t$; and {\em (c)} the connected components of free arcs (see \Cref{def:free-comps}).  This partition is gradually coarsened throughout the algorithm, and every minimum cut is the union of certain components of $\CP$. At termination, the component $P_s$ containing $s$ will have $\res{f}_e=0$ on all outgoing arcs and no incoming flow, at which point $P_s$ will be a minimum cut.

\subsubsection{Essential arcs and the role of transitive cover}
A central concept in the design and analysis for our framework is \emph{essential arcs} $\Eaux$, that is, arcs between different free components with $\ab^{-5}\varepsilon\le \res{f}_e< \ab\varepsilon$.
Essential arcs are not abundant but within a $\mathrm{poly}(n)$ factor from the threshold. %
They play two key roles in the algorithm and analysis:
\begin{enumerate}[(1)]
    \item In each iteration, it suffices to call an approximate maximum flow solver on an auxilliary instance constructed from $\Eaux$ together with transitive arcs outputted by the transitive cover on the nodes incident $\Eaux$.  The value of this instance will be sufficiently close to the value of the residual instance $(G,\res{f})$. 
    
    \item The total number of essential arcs throughout is at most $O(m_c)$.
\end{enumerate}
The two points above imply that the total size of the nodes sets in the approximate maximum flow calls is $O(m_c)$. This is reflected in the ${\rm TC}_{\rm time}(n,O(m),O(m_c))$ and ${\rm TC}_{\rm arcs}(n,O(m),O(m_c))$ terms in 
 \Cref{thm:fundamental}. We now discuss these two points.

\paragraph{(1) The auxiliary graph.} 
The auxiliary instance depends on the transitive cover on the set of abundant arcs.
 Let $R=V(\Eaux)\cup\{s,t\}$ denote the set of endpoints of the essential arcs together with the source and the sink.   We obtain a transitive cover $(V',H)$ of $R$ from an  \DITC{}. 
The auxiliary instance $(\cG,\cu)$ is defined as $\cG=(V',\Eaux\cup H)$, with
 residual capacities $\cu_e=\res{f}_e$ for $e\in\Eaux$ and $\cu_e=\infty$ for $e\in H$.\footnote{The actual  construction in \Cref{sec:max-flow} will be slightly different, by handling nodes in free components together, and 
 by using different capacities.} We call the approximate maximum flow solver in this instance, and map back the flow to the original instance; this may involve adding new arcs to the instance.

The total capacity of arcs with $\res{f}_e<\ab^{-5}\varepsilon$ is at most $n^2\ab^{-5}\varepsilon<\ab^{-4}\varepsilon$. We now argue that $\nu(\cG,\cu)\ge \nu(G,\res{f})-\ab^{-4}\varepsilon$ (see \Cref{lem:compact-flow-value}).
For any $s$--$t$ min cut $S'\subseteq V'$ in $(\cG,\cu)$, we can consider the set $S\subseteq V$ reachable from $S'$ on abundant arcs. Using the transitive cover property, it 
 can be shown that $S\cap V'=S'$; in particular, $S$ is still an $s$--$t$ cut. Now, every arc leaving $S$ that did not leave $S'$ must have $\res{f}_e<\ab^{-5}\varepsilon$. Consequently, the cut value of $S$ in $(G,\res{f})$ may be at most $\ab^{-4}\varepsilon$ higher than the value of $S'$ in $(\cG,\cu)$, implying the bound. 

\paragraph{(2) Bounding the number of essential arcs.} 
The invariants of  `valid iterates'  maintained throughout the algorithm (see \Cref{def:proper}) require that for any arc not incident to $s$ or $t$ between two different free components, either  $\res{f}_e=0$ or $\res{f}_{\rev{e}}=0$, or if both  $\res{f}_e>0$ and $\res{f}_{\rev{e}}>0$, then both residual capacities will be large enough to guarantee that $e$ becomes free in the next iteration.  In the latter case, the arc is referred to as a \emph{gap arc}.  If there are $k$ gap arcs at an iterate, then the number of free components decreases by at least $k/2$ at the next iteration. Thus, the total number of gap arcs over all iterations is bounded by $O(n)$. Further, the  invariant on residual arc capacities  also implies each arc that is not incident to $s$ or $t$ with 
 $u_e+u_{\rev{e}}<\infty$ may be essential for at most two iterations. This implies that if total number of essential arcs not incident to  $s$ and $t$  over all iterations is $O(m_c)$.

We do not enforce the invariant on residual capacities for arcs incident to $s$ or $t$, as they are used to re-balance flow. Instead, we establish the following with respect to arcs incident to $s$  (and a similar result is true for arcs incident to $t$). If $(s,i)$ is the only essential  arc incident to a component $P\neq P_s$ of $\CP_{f,\varepsilon}$, then within $O(1)$ iterations, $P$ will either be incident to a new essential arc that is not incident to $s$ or $P$ gets merged into another free component. The former may happen $O(m_c)$ times over all iterations, and the latter $O(n)$ times overall.

\paragraph{Extending the instance and routing back flow.}
Note that the arcs $H$ used in the auxiliary instance are not necessarily part of the original arc set $E$. Still, the approximate solver may send large values on
such arcs. We extend the instance by adding these new arcs with $\infty$ capacity; adding such arcs does not change the instance value $\nu(G,u)$. 

The transitive cover data structure includes  $\dsTcWit(\cdot)$, which creates a `witness list' at the end of the algorithm. 
This witness list is used in the final subroutine $\WitRoute(\cdot)$, which transfers flow
  from arcs in the \emph{extension arc set} $F$, which had been added to the original instance $(\gGi,\ui)$ throughout the algorithm. Such a mapping is possible because when any arc $(i, j)$ is added to $F$, there is an abundant  path of original arcs and previously added extension arcs from $i$ to $j$. 
If the extension arc $e=(i,j)$ was added in the  $\varepsilon$-scaling phase, then 
every arc in the abundant path of original arcs  
has residual capacity $\ge \ab\varepsilon$.
Further, following the  $\varepsilon$-scaling phase, all subsequent changes in the flow will be bounded as $O(\varepsilon)$; i.e., we only need to reroute $O(\varepsilon)$ final flow units from $e$, and the underlying path still has nearly $\ab\varepsilon$ residual capacity. 

 The rerouting is carried out recursively, since flow from extension arcs may be moved to extension arcs added earlier that in turn are routed on even older extension arcs. Still, the abundant threshold $\ge \ab\varepsilon$ is sufficient so that we can reroute the entire flow from all extension arcs simultaneously and obtain a feasible flow in the original graph.

\medskip

In some applications such as maximum closure, applied in particular to open pit mining  \cite{hochbaum2001new,Hochbaum2000,johnson1968}, the goal is to find a minimum cut, and the maximum flow may not be necessary. In this case, our algorithm can be significantly simplified. There is no need to map the maximum flow back to the original instance, and consequently, the data structure does not need to maintain witnesses for rerouting or create a witness list.

\subsubsection{Incremental transitive cover algorithms}

\label{sec:data-structure-overview}

We now give an overview of the \emph{incremental transitive cover data structures} (\DITC{}s) developed in the paper. Below we explain how our data structures compute and output transitive covers. Though we defer the details to \Cref{sec:trans_cover_first_introduce,sec:data-structures}, we note here that these data structures support an additional operation to allow flow on transitive arcs to be efficiently routed on the original graph. Whenever an arc $(i,j)$ is added to the transitive closure in these data structures, we maintain a \emph{witness} $v$ such that $(i,v)$ and $(v,j)$ were already in the transitive closure. By ensuring that the list arcs, transitive arcs, and witnesses are suitable stored in what we call a \emph{rooted witness list}, we are able to efficient move flow off arcs not in the original instance.
As noted above, the witness operations can be omitted if we only need to find a minimum cut.

\paragraph{Italiano's transitive closure algorithm.} In \Cref{sec:dsclo}, we show how to adapt Italiano's \cite{Italiano1986} approach to transitive closure to design an \DITC{}. This data structure essentially maintains for each node a tree corresponding to what nodes it can reach. Maintaining this tree for a single node over $m$ additions can be done in $O(m)$ time; therefore, maintaining it for all nodes can be done in $O(mn)$ time. To compute a cover of $S \subseteq V$, our \DITC{} simply outputs either the closure with respect to $S$ or the entire graph, depending on which is presumed to be sparser.

\paragraph{Bounded tree-depth.} In \Cref{sec:tree_depth}, we build upon Italiano's approach to develop an  \DITC{} suitable for efficiently solving the maximum flow problem on graphs with bounded tree-depth. This \DITC{} assumes that a rooted tree of depth $D$ is given and that all arc insertions \emph{respect the tree}, meaning that one endpoint is an ancestor or a descendant of the other. 

Similarly to Italiano's algorithm, this algorithm essentially maintains reachability trees from nodes. It departs from Italiano's algorithm in terms of what these trees are and how they are used. This algorithm maintains for each node $a$, a reachability tree to its descendants using only arcs where both endpoints are descendants of $a$.  It also maintains a reachability tree from $a$'s descendants to $a$ using only arcs where both endpoints are descendants of $a$. 
Formally, the \DITC{} maintains the following transitive arcs 
\[
E_* = \left\{ (a,b) \in \R^{V \times V} ~ | ~ (a, b) \in E[\descendant_T(a)]^\tr \text{ or } (a, b) \in  E[\descendant_T(b)]^\tr \right\},
\]
where $\descendant_T(a)$ is the set of descendants of $a$, and $E$ is the set of currently added arcs. Leveraging insights from Italiano's approach regarding dynamically maintaining reachability from a single source, we show that $E_*$ can be maintained over $m$ insertions in $O(mD)$ time, improving upon the $O(mn)$ time of Italiano's algorithm.

To compute the cover in this \DITC{}, we use the tree structure to output more efficient covers. To output the cover of a set $S$, we simply output for each $a \in S$ the arcs in $E_*$ where $a$ is one endpoint and the other endpoint is an ancestor of $a$. This graph contains $O(|S|D)$ arcs and can be output in $O(|S|D)$ time. Interestingly, this set of arcs is a transitive cover of $S$ that may be sparser than the transitive closure of $S$.  

\emph{Additional related work:} 
In our maximum flow framework, we generally know beforehand a superset of the arcs that will be inserted into the \DITC{}: namely, the arcs and reversals of arcs in the starting network (ignoring the transitive arcs potentially outputted by the cover). Our \DITC{} for this case of bounded tree-depth leverages the structure of these arcs, even though the the precise order in which these arcs will be added is unknown. Such a model was also studied by Italiano, Karczmarz, {\L}acki and Sankowski~\cite{italiano2017decremental} in the context of incremental transitive closure for planar networks, where it was dubbed the ``switch on'' model. The fact that having small balanced vertex separators is useful for transitive closure maintenance, which we heavily rely on in this bounded tree-depth setting, was also already noticed by Karczmarz~\cite{karczmarz2018decremental}. Karczmarz studied the potentially more challenging decremental transitive problem and correspondingly, our tree-depth based variant of Italiano's algorithm is arguably simpler than Karczmarz's data structure.

\paragraph{Node orderings.}
Finally, in \Cref{sec:fmm}, we give an \DITC{} suitable for efficiently solving the maximum flow problem on graphs with bounded $m_c$; this \DITC{} supports our $O(n^{\omega + o(1)})$ time maximum flow algorithms when $m_c = o(n)$. In addition to receiving arc additions, this data structure maintains an ordering of the nodes and supports an operation where the position of nodes in the ordering can be changed. 
 We call such a data structure an \emph{ordered incremental transitive cover data structure (\ODITC{})}. We give an implementation of an \ODITC{} where the implementation cost depends on how aligned arc additions, vertex re-orderings, and cover queries are with the ordering.

This data structure periodically computes the transitive closure on subsets of the current arcs. The data structure essentially maintains the transitive covers of the first $\Theta(2^k)$ nodes with respect to the current node ordering for each $k \in [K]$, where $K=\log_2 n$. It then periodically updates them, leveraging that the transitive closure (and more broadly reachability trees) on $k$ node graphs can be computed using fast matrix multiplication (FMM) in $O(k^\omega)$ time.

To explain the \ODITC{}'s running time, suppose that there are a series of additions, re-orderings, and queries and the $i$-th operation involves the first $X^{(i)}$ nodes in the ordering. Then if $X = \sum_{i} X^{(i)}$, by carefully nesting and updating the covers, the cover corresponding to the first $2^{k}$ nodes gets updated at most $\tilde{O}(X 2^{-k})$ times. Consequently, the data structure can be implemented in time:
\[
\otilde\left(\sum_{k \in [K]} \frac{X}{2^k} \cdot (2^k)^\omega\right) = \otilde\left(X \sum_{k \in [K]} 2^{k(\omega-1)}\right) = \otilde\left(X 2^{K (\omega - 1)}\right)=\otilde\left(X n^{\omega - 1}\right)\,.
\]
When applying this data structure for maximum flow, $X$ is  bounded by the total number of essential arcs and thus $X = O(m_c)$. Consequently, for the maximum flow problems in which $m_c = O(n)$, the total running time required by the data structure is $\tilde{O}(n^\omega)$, which is almost the same as the cost of computing a single transitive closure of the entire graph.

\emph{Additional related work:} Our \ODITC{} is closely related to a broader literature on dynamic data structures that use algebraic techniques to maintain transitive closures, strongly connected components, and reachability information \cite{KingS99, Sankowski04, DemetrescuI05, SankowskiM10, Kavitha14, BrandNS19, Karczmarz0S22, KarczmarzS23a,  KarczmarzS23b, BrandK23, BrandFNP24, HenzingerSSY24, LiuS24, GoranciKMP25}. Particularly relevant to our approach are results on dynamic graph algorithms which use additional ordering information, e.g., 
\cite{KhannaMW98, SankowskiM10,Kavitha14,BrandNS19,BrandFNP24,HenzingerSSY24,LiuS24}. For example, \cite{KhannaMW98, SankowskiM10,Kavitha14,BrandNS19} considers related problems where the order of operations is approximately known and \cite{BrandFNP24,HenzingerSSY24,LiuS24} considers variants of these problems where there is a prediction of the ordering of operations and runtimes depend on the quality of the prediction. Our data structure perhaps differs slightly from individual prior works in terms of how ordering information is handled (e.g., we allow the ordering to be changed for the queried elements), the particular queries and updates supported (e.g., adding edges within a vertex subset rather than updating a column of the matrix), the particular running time trade-offs we obtain between operations, and in how we require supporting witness lists. However, the techniques we apply, e.g., carefully applying FMM on geometrically growing subsets, is a standard approach in the literature (see e.g.,  \cite{SankowskiM10, BrandNS19}).

\subsubsection{Comparison to Orlin's algorithm}
\label{sec:orlin-compare}

Our framework differs from Orlin's algorithm in significant ways. We now elaborate on the three aspects mentioned in \Cref{sec:max-flow-tc}, and give running time comparisons.
 
\paragraph{Structured flow iterates} We maintain valid iterates (as in \Cref{def:proper}) imposing much stronger structural properties than in Orlin's algorithm. A key feature is that in our algorithm, most arcs between free components throughout will be at the lower or upper capacity bound.  That is, after contracting the free components, our flow will be close to a basic feasible flow. Specifically, any arc between different free components that is not incident to $s$ or $t$ with positive residual capacities on both sides will become free (i.e., abundant in both directions) in the next iteration, that is, it will merge two free components. The nearly basic flow structure ensures that we learn roughly 1 abundant arc per free component incident to an essential arc in each iteration (though proof is rather technical). As we compute the transitive cover in each iteration on a set of representatives nodes for each free component incident to an essential arc, this guarantee is what ensures that the total size of the queries to the transitive cover data structure is bounded by $O(m_c)$.

To maintain the requisite sparsity of valid iterates, a post-processing in each iteration sacrifices a small amount of flow amount by sending it towards the source or away from the sink. For this purpose, we add infinite capacity arcs from all nodes to the source, and from the sink to all nodes. This does not change the maximum flow value and allows us to easily deal with small with small violations of flow conservation in a local manner. Another key technical idea in this context is to use \emph{safe residual capacities}: to ensure that some of the remaining arcs become abundant in the next iteration, we artificially restrain the flow these arcs to a smaller interval while maintaining the maximum flow value.

Orlin's framework however, using the present language, does not guarantee that free components incident to essential arcs contribute one new abundant arc per iteration.
In the network where free components are contracted, there may be problematic nodes which are only used to route flows on abundant paths (i.e., every arc on the path has infinite residual capacity) in the current flow. As these abundant paths could potentially be rerouted to completely avoid any such node in the next iteration, one cannot ensure that such a node will be incident to a new abundant arc. To remove these ``pass through'' nodes when creating the auxilliary network, Orlin uses dynamic trees to replace abundant flow paths by single edges (in essence, this replaces the flow on abundant arcs by a corresponding $b$-matching), which results in $O(m \log n)$ overhead whenever Orlin constructs an auxilliary network. Unfortunately, one may not be able to effectively amortize the cost of this $O(m \log n)$ overhead to the abundant arcs learned in the next iteration, since the resulting auxilliary network may have very few nodes. To overcome this difficulty in this case, Orlin created an alternative auxiliary network with $O(m^{1/3})$ nodes such that finding an optimal flow in this network would result in learning a sufficient number of abundant arcs. To compute the optimal flow in the alternative network, Orlin used the exact $O(n^3)$-time strongly polynomial solver of~\cite{malhotra1978v}.

Our algorithm is considerably simpler in this respect: we only need one type of auxiliary network that is always fast to construct (essentially linear in the size of the network modulo the transitive cover call), and we only require a fast weakly polynomial subsolver. 
Using the $m^{1+o(1)}$-time solver in Orlin's algorithm (with suitable optimization of parameters), this improvement allow us to reduce the time spent on approximate maximum flow computations on general networks from $O(m_c m^{2/3+o(1)})$ in Orlin's framework to $O(m_c \sqrt{m})$ in our framework. 

\paragraph{Extending the instance}
In both algorithms, the auxiliary instances include transitive closure arcs supplied by the transitive closure/cover data structures. 
Whenever an approximate maximum flow contains transitive arc with positive positive flow,  Orlin’s algorithm  transfers these flows to flows on original arcs during an $O(m \log n + Kn)$ post-processing step of the auxiliary network, where $K$ is the number of nodes in the auxiliary network. Thus the transferring of flows adds $O(m_c n)$ time to the running time of Orlin's algorithm, another $O(nm)$ bottleneck when most arcs are capacitated. Our algorithm avoids this running time by delaying these transfers to a final stage of the algorithm. Instead, it permits flows on transitive arcs at intermediate stages.  Thus the first maximum flow $f^*$ obtained by our algorithm may include flows on transitive arcs. Our transitive closure includes an operation $\dsTcWit()$ that is used to transfer flow from the transitive arc flows of $f^*$ to flows on original arcs, resulting in a maximum flow on the original network. This transfer takes only $O(n^2)$ time in the general case and $O(nD)$-time in the bounded tree-depth setting.

\paragraph{New dynamic data structures} We introduce the versatile ITCO data structure. Using transitive covers rather than closures enables more flexibility in the applications, in particular, it is the key tool in the almost-linear algorithm for bounded tree-depth. In the context of data structures, our main new contribution is developing the node-ordered variant using fast matrix multiplication (FMM). While FMM already appeared in Orlin's algorithm for $b$-matching, it only involved periodic recomputation of the entire transitive closure from scratch. Our amortized recomputation scheme is the key to the improved running time in this context.

\subsection{Overview of the paper}
In \Cref{sec:prelim}, we describe preliminaries of the algorithm. The  \DITC{} is introduced in \Cref{sec:trans_cover_first_introduce}.
 In \Cref{sec:maxflow_framework},  we develop the conceptual framework of the algorithm, including the concepts of abundant and free arcs and valid iterates.
 In \Cref{sec:max-flow}, we present the meta-algorithm
and its sub-routines.
\Cref{sec:analyzing} presents and proves the main correctness and running time bound \Cref{thm:maxflow-analysis-main}, a formal version of \Cref{thm:fundamental}.
In \Cref{sec:data-structures}, we describe the data structures for maintaining the \DITC{} in the different cases, as well as prove correctness and running time bounds.
In \Cref{sec:algorithms-overall}, we combine the meta-algorithm with the particular implementations to prove the bounds in the different settings, that is, \Cref{thm:general-flow}, \Cref{thm:ordered-overall}, and \Cref{thm:tree-depth-overall}.
\Cref{sec:approx_flow} shows how a fast weakly polynomial time algorithm for maximum flows can be used within a strongly polynomial time algorithm to obtain an approximate maximum flow.

\section{Preliminaries}\label{sec:prelim}

$\Rs \defeq \Rnn\cup\{\infty\}$, and $[k]\defeq\{1,2,\ldots,k\}$ for $k\in \mathbb{N}$. For $v\in\R$, we let $v^+\defeq\max\{v,0\}$ and $v^-\defeq \max\{-v,0\}$; note that $v=v^+-v^-$. For a vector $x\in \R^X$, we let $\supp(x)\defeq\{i\in X:\, x_i\neq 0\}$ denote its support. For a subset $S\subseteq X$, we let $x(S) \defeq \sum_{i\in S} x_i$.

\paragraph{Graphs.}
Let $G=(\gV,\gE)$ be a directed graph; in this paper we allow directed graphs to have parallel arcs but no self-loops. We let $V(G)=\gV$ and $E(G)=\gE$ refer to the nodes and arcs of $G$, respectively. For $F\subseteq E$, we let $V(F)\subseteq V$ denote the set of endpoints of $F$.
We say that $j$ is the head and $i$ the tail of the arc $e=(i,j)$.

For  $S,T\subseteq\gV$, we let $\dtot{\gE}{S,T}\defeq\{e=(i,j)\in \gE:\, i\in S\setminus T,\, j\in T\setminus S\}$ denote the set of arcs from $S$ to $T$. 
Additionally, we let 
$E[S] \defeq \{e=(i,j)\in E:\, i,j\in S\}$, and let $G[S]\defeq(S,E[S])$
denote the graph induced on $S$.
We let $F\uplus H$ denote multiset union for arc sets.

For $x\in\Rs^E$, we let $x(S,T)\defeq x(\dtot{\gE}{S,T})$.
For $S\subseteq \gV$, we let  $\din{\gE}{S}\defeq\dtot{\gE}{V\setminus S,S}$ and $\dout{\gE}{S}\defeq\dtot{\gE}{S,V\setminus S}$ denote the set of arcs entering and leaving $S$, respectively, and let $\dtot{\gE}{S} \defeq \din{\gE}{S}\cup \dout{\gE}{S}$. When $S$ is a singleton, i.e.,\ $S = \{i\}$, we  write $i$ for brevity, e.g., $\din{\gE}{i} \defeq\din{\gE}{\{i\}}$,  $\dout{\gE}{i} \defeq \dout{\gE}{\{i\}}$ and $\dtot{\gE}{i} \defeq \dtot{\gE}{\{i\}}$.

\paragraph{Maximum flow instances.} 
A \emph{maximum flow instance} (or \emph{instance}, for short) is $(G,u)$, where $G=(\gV,\gE)$ is a directed graph with a source node  $s\in \gV$, sink node $t\in \gV$, and arc capacities $u\in \Rs^{E}$. $s$ and $t$ are  fixed throughout and thus, we omit them from the instance notation. 

 We assume that the arc set $E$ in a maximum flow instance is \emph{symmetric}: 
for every $(i,j)$-arc $e\in \gE$, there is a $(j,i)$-arc $\rev{e}\in \gE$ called the \emph{reverse arc}. The reverse mapping is an involution: if $h=\rev{e}$ then $\rev{h}=e$.
For an arc subset $F\subseteq E$, we let $\rev{F}\defeq\{\rev{e}:\, e\in F\}$. Thus, we have  $E=\rev{E}$. 
For any $f,g\in\Rnn^{\gE}$, we define $f\oplus g\in\Rnn^{\gE}$ for all $e\in{\gE}$ by $(f\oplus g)_e=(f_e-f_{\rev e}+g_e-g_{\rev e})^+$.

In the input instance, we assume that for every $e\in \gE$, at least one of the capacities $u_e$ and $u_{\rev{e}}$ is finite; 
 we also assume there is no $(s,t)$ arc.\footnote{Note that if both $e=(i,j)$ and $\rev{e}=(j,i)$ are infinite, then we can reduce to an equivalent smaller instance by contracting $i$ and $j$. All $(s,t)$ arcs can be removed without loss of generality, as every maximum flow must saturate them.}
 
Furthermore, we assume that 
\begin{equation}\label{eq:i-s-t}
\textrm{for each $i\in \gV\setminus\{s,t\}$, there exists arcs $(i,s), (t,i)\in\gE$ with infinite capacities.} 
\end{equation}
This assumption is without loss of generality;  every instance can be easily transformed to an equivalent one satisfying this assumption by adding $O(|V|)$ arcs in $O(|V|)$ time. In the bounded tree-depth case, this may increase the tree-depth by at most 2. We discuss this further in \Cref{sec:algorithms-overall}.

 For a vector $f\in\R^\gE$ and a set $S\subseteq \gV$, we let 
\[
\exc{f}{S}\defeq f(\din{\gE}{S})-f(\dout{\gE}{S})
\]
denote the \emph{excess} of $S$; we use $\exc{f}{i}$ for the singleton $S=\{i\}$. Note that $\exc{f}{\cdot}$ is a modular function, i.e., $\exc{f}{S}=\sum_{i\in S}\exc{f}{i}$.

The \emph{flow polyhedron} associated with $(G,u)$ is
\begin{equation}
\flowP(G,u)\defeq \left\{f \in \R^\gE \,:\,0\leq f \leq u \text{ and } \exc{f}{i} =0\, \text{ for all }  i \in \gV\setminus\{s,t\}\, \right\}\, .
\end{equation}
A vector $f\in\flowP(\I)$ is called a \emph{feasible flow}, or simply a \emph{flow}.
When clear from the context, we let $x_{ij}$ and $u_{ij}$ denote the flow $x_e$ and capacity $u_e$ on the arc $e=(i,j)$. (Note that there may be parallel arcs, so this may not always be uniquely defined.)
The \emph{value} of a  flow $f$ is denoted as 
\[
\val{f}\defeq-\exc{f}{s}\, .%
\]
Note that $\val{f}=\exc{f}{t}$.
We let
\[
\nu(G,u)\defeq\max\{\val{f}:\, f\in\flowP(G,u)\}
\]
 denote the maximum flow value in $(G,u)$. We say that $f\in\R^\gE$ is a \emph{maximum flow} if $f\in\flowP(G,u)$ and $\val{f}=\nu(G,u)$. 
We say a instance $(G,u)$ is \emph{bounded} if $\nu(G,u)<\infty$; note that $(G,u)$ is bounded if and only if there is no $s$--$t$ path in $\gE$ of arcs with infinite capacity. Throughout, we assume that the instance is bounded as this can easily be checked in linear time.

\paragraph{Minimum cuts.}
A set $S\subseteq \gV\setminus\{t\}$ with $s\in S$ is called an \emph{$s$--$t$ cut}. We call  $u(S,V\setminus S)$ the \emph{cut capacity} of $S$. An $s$--$t$ cut is a \emph{minimum cut} if it minimizes $u(S,V\setminus S)$ over all
 $s$--$t$ cuts. We recall the maximum flow minimum cut theorem:
\begin{theorem}[Maximum flow minimum cut]\label{thm:mfmc}
In any instance  $(G,u)$ with $G=(V,E)$, $\nu(G,u)$ equals
the minimum capacity of an $s$--$t$ cut. \end{theorem}

Given an instance $(G,u)$ and  $\varepsilon\ge 0$, we say that a flow $f$ is  \emph{$\varepsilon$-optimal} if $f$ is feasible  and $\val{f}\ge\nu(G,u)-\varepsilon$. 
An $s$--$t$ cut $S$ is \emph{$\varepsilon$-optimal} if $u(S,V\setminus S)\le  \nu(G,u)+\varepsilon$.

\paragraph{Acylic basic flows.}A flow $f$ in $(G,u)$ is 
 \emph{acyclic} if $\supp(f)$ does not contain any directed cycles.  Further, $f$ is a 
 \emph{basic flow} if it cannot be written as a nontrivial convex combination of other flows. Note that a basic flow is not necessarily acyclic as the support may include directed cycles saturating an arc capacity; conversely, an acyclic flow is not necessarily basic. By an \emph{acylic basic flow}, we mean a basic flow that is also acyclic.
 The following lemma asserts that any flow can be efficiently converted to an acyclic basic flow, and lists some useful properties.
\begin{restatable}{lemma}{basicflow}\label{lem:basic-flow}
Given a flow $f$ in the instance $(G,u)$ on $n$ nodes and $m$ arcs, in $O(m\log n)$ time    one can find a basic acylic flow $\bar f$ such that $\supp(\bar f)\subseteq \supp(f)$, $\val{\bar f}\ge \val{f}$. Every acyclic basic flow $\bar f$ satisfies the following properties.
\begin{enumerate}[(i)]
\item\label{it:basic-size} 
$\|f\|_\infty\le\val{\bar f} $. 
\item\label{it:nobackflow} $\bar f_e=0$ for all $e\in \din{E}{s}\cup\dout{E}{t}$. 
\item\label{it:acyclic} The set of arcs $\{e\in E\, :\, 0<\bar f_e<u_e\}$ form a forest, with $s$ and $t$ belonging to different components.
\item\label{it:flow-pm}  For every $e\in\gE$, the flow value $\bar f_e$ can be written in the form $\bar f_e=\sum_{e'\in E(G)} \chi_{ee'} u_{e'}$, where $\chi_{ee'}\in \{0,\pm 1\}$. 
\end{enumerate}
\end{restatable}
 The proof is given in \Cref{sec:approx_flow}; we give a brief outline here. One can obtain a basic acylic flow by first computing an acyclic flow and then turning it into basic; both steps can be performed in $O(m\log n)$ using link-cut trees \cite{Sleator83}. Among the properties, \eqref{it:basic-size} and 
\eqref{it:nobackflow} hold for all acyclic flows and follow using standard  flow decomposition (\Cref{thm:flow-decomp}); \eqref{it:acyclic} is easy to show for all basic flows, and implies  \eqref{it:flow-pm}.

\paragraph{Instance extensions.}
Given $(G,u)$ and two (possibly empty) arc sets $\Fone,\Ftwo\subseteq V\times V$, we define the instance $(G,u)\oplus(\Fone,\Ftwo)=(G',u')$ by adding all arcs in $\Fone$ with infinite capacity, and all arcs in $\Ftwo$ with infinite capacity in both directions. Formally, we let $G'\defeq(V(G),E')$, where $E'=E(G)\uplus \Fone\uplus \rev{\Fone}\uplus \Ftwo\uplus \rev{\Ftwo}$ denotes the multigraph obtained by adding $\Fone$, $\Ftwo$, $\rev{\Fone}$, and $\rev{\Ftwo}$ to $E(G)$ as a multiset union.
Further, we let $u'_e\defeq u_e$ for $e\in E$, $u'_e\defeq \infty$ for $e\in \Fone$ and $u'_e\defeq 0$ for $e\in \rev{\Fone}$, and $u'_{e}=u'_{\rev{e}}\defeq\infty$ for $e\in \Ftwo$.\footnote{We need to add the reverse arcs $\rev{\Fone}$ to maintain symmetry of the instance.}

\begin{definition}[Admissible extension]\label{def:admissible}
The  instance $(G',u')=(G,u)\oplus(\Fone,\Ftwo)$ is an \emph{admissible extension} of $(G,u)$ if the arcs in $\Fone\cup \Ftwo\cup \rev{\Ftwo}$ do not leave any  minimum $s$--$t$ cut in $(G,u)$.
\end{definition}
\Cref{def:admissible} implies that if $(G',u')$ is an admissible extension of $(G,u)$, then  $\nu(G',u') = \nu(G,u)$.

\subsection{Residual capacities and safe capacities}
Given an instance $(G,u)$ and a flow 
$f\in\R^\gE$, we define the \emph{residual capacity} vector $\res{f}\in \Rs^{\gE}$ by
 $\res{f}_e\defeq u_e-f_e+f_{\rev{e}}$. The instance $(G,\res{f})$ is called the \emph{residual instance of $f$}.

\begin{lemma}\label{lem:res-value}
Given an instance $(G,u)$ and a feasible flow $f$, we have
\[
\nu(G,u)=\val{f}+\nu(G,\res{f})\, .\]
\end{lemma}
\begin{proof}
We first show $\nu(G,\res{f})\ge \nu(G,u)-\val{f}$.
Let $f^\star$ be a maximum flow of $(G,u)$, i.e., $\val{f^\star}=\nu(G,u)$ and $f^\star \in \flowP(G,u)$. Let $h\defeq f^\star\oplus (-f)$. %
 Then, $0\le h\le \res{f}$ and $\exc{f}{i}=0$ for $i\in V\setminus \{s,t\}$; thus, $h$ is a feasible flow in the residual instance $(G,\res{f})$, giving  $\val{h}\le \nu(G,\res{f})$. Further, $\exc{f^\star}{s}=\exc{f}{s}+\exc{h}{s}$, giving $\val{h}= \val{f^\star}-\val{f}=\nu(G,u)-\val{f}$. %
In the other direction, let  $g$ denote a maximum flow in $(G,\res{f})$. Then, $g\oplus f$ is feasible in $(G,u)$, and $\val{g\oplus f}=\val{g} + \val{f} \le \nu(G, u)$.%
\end{proof}

In the algorithm, it will be convenient to restrict the residual capacities, as long as this change does not decrease the residual instance value. This is captured in the following definition.

\begin{definition}[Safe Residual Capacities]\label{def:safe-cap}
Given an instance $(G,u)$ and a flow $f$, we say that $\nr\in \Rs^{\gE}$ is a \emph{safe residual capacity vector with respect to $f$}, or that \emph{$\nr$ is safe with respect to $f$}, if $0\le \nr\le \res{f}$, and $\nu(G,\res{f})=\nu(G,\nr)$.
\end{definition}
We will also use the shorter expression `safe capacity vector.'
Note that in particular $\nr=\res{f}$ is always a safe (residual) capacity vector.

\begin{lemma}\label{lem:safe-cap}
Given an instance $(G,u)$ with a 
 feasible flow $f$, the capacity vector $0\le\nr\le \res{f}$ is safe w.r.t.\ $f$ if and only if there exists a maximum flow $f^\star$ in $(G,u)$ with $f^\star_e\le f_e-f_{\rev{e}}+\nr_e$ for every $e\in\gE$. Furthermore, if $\nr$ is  safe  w.r.t.~$f$ then $\nr_e= \res{f}_e$ for all $e \in \dout{\gE}{S}$ where $S$ is a min-cut.
\end{lemma}
\begin{proof}
The first part is immediate by \Cref{lem:res-value}. The second part follows since every maximum flow must have $f^\star_e=u_e$ for every arc $e$ in a minimum cut in $(G,u)$, which are the same as the minimum cuts in $(G,\res{f})$ by \Cref{lem:res-value}.
\end{proof}

\subsection{Approximate flow solver}

\begin{definition}[Approximate Flow Solver]\label{def:approx-solver}
By a \emph{$\runtime(m,M)$-time approximate-flow solver}, we mean a subroutine  
$\approxFlow(G,u,M)$ that, for an instance $(G,u)$ with $|E(G)|=m$ and $M\in\Rp$,
in time $\runtime(m,M)$ outputs an acyclic basic flow $f$ in $(G,u)$ along with an arc $e^\star\in E(G)$ such that 
\[
\nu(G,\res{f})\le m\res{f}_{e^\star} \le m\nu(G,\res{f})\le \frac{\nu(G,u)}{M}\, .
\]
\end{definition}
The arc $e^\star$ will serve as a certificate of the value
of the residual flow, and will be particularly useful to keep the bit-complexity of the algorithm bounded.

Recent breakthrough results led to almost-linear time solvers for maximum flow (and more generally, for min-cost flow). Following the randomized algorithm by 
Chen, Kyng, Liu, Peng, Gutenberg, and Sachdeva~\cite{Chen2022maximum},  a deterministic almost-linear time algorithm was recently obtained by van den Brand, Chen, Kyng, Liu, Peng, Gutenberg, Sachdeva, and Sidford~\cite{brand2023deterministic}.
\begin{theorem}[{\cite{brand2023deterministic}}]\label{thm:approx-flow-near-linear}
There is a deterministic 
 polynomial time algorithm that given an $m$-arc, $n$-node instance $(G,u)$ with integral capacities bounded by $U$, computes an integer valued exact maximum flow in time 
$m^{1+o(1)}\log U$.%
\end{theorem}

We note that for  dense instances,  van den Brand, Lee,  Liu,  Saranurak, Sidford, Song, and Wang \cite{Brand2021} gave a  a near-linear $\tilde O((m+n^{1.5})\log U)$ algorithm.
In  \Cref{sec:approx_flow}, we show a reduction from \Cref{thm:approx-flow-near-linear} to a near-linear time approximate maximum flow solver, as stated next. The reduction uses standard techniques, including a subroutine $\MaxCap(G,u)$ that computes the maximum capacity of an $s$--$t$ path. This will also be instrumental in obtaining the arc $e^\star$ such that $m\res{f}_{e^\star}$ is an upper bound on the residual flow value. 
 We use \Cref{thm:approx-flow-near-linear} for instances with integer capacities $O(m^2M)$.
\begin{restatable}{theorem}{apxflowsolver}\label{thm:apx-flow-solver-run}
There exists an  $m^{1+o(1)}\log M$-time deterministic approximate maximum flow solver.
\end{restatable}

\section{Incremental Transitive Cover Data Structures}
\label{sec:trans_cover_first_introduce}

We now introduce the incremental transitive cover data structure. Throughout this section (and later when referring to general transitive closure data structures)  the graph, $\dG=(\dV,\dE)$,  denotes an arbitrary directed graph and not a maximum flow instance; we do not require $\dG$ to be symmetric, and it does not have arc capacities. We write $\dG=(\dV,\dE)$ rather than $G = (V,E)$ to clarify the distinction. In the context of our maximum flow algorithm, $\dG=(\dV,\dE)$ will model the graph of \emph{abundant arcs}, which are introduced in Section~\ref{sec:abundant}. In this setting $\dV = V$, however we use the notation $\dG=(\dV,\dE)$ as the development of our \DITC{}s is self-contained nature and potentially more broadly applicable.

Given a directed graph $\dG=(\dV,\dE)$,
the \emph{transitive closure} graph $\dG^\tr=(\dV,\dE^\tr)$ contains all arcs $(i,j)\in \dV \times \dV$ for which there is a directed $i$--$j$ path in $G$.%

\begin{definition}[Transitive cover]\label{def:trans-cover}
For a graph $\dG=(\dV,\dE)$ and subset $S\subseteq \dV$, we say that the graph $(\dV',F)$ is a \emph{transitive cover of $S$} (in $\dG$),
if $S\subseteq \dV'\subseteq \dV$, $F\subseteq E^\tr[\dV']$, and $E^\tr[S] = F^\tr[S]$. 
\end{definition}
Thus, a transitive cover of $S$ is a subgraph $(\dV',F)$ of $\dG^\tr$ such that whenever $\dG$ has an $i$--$j$ path between two nodes $i$ and $j$ in $S$, then there is an $i$--$j$ path in $F$.

A key ingredient of our maximum flow algorithms is a transitive cover data structure. This data structure assumes an initial empty graph on $n$ nodes $V$ and processes a sequence of requests to add arcs ($\dsTcAdd$) and compute a transitive cover of a set of nodes ($\dsTcCover$). The data structure also supports an additional operation, $\dsTcWit()$, which outputs what we call a \emph{witness list}.  The witness list certifies that all arcs in the output of $\dsTcCover(\cdot)$ are indeed in the transitive closure at that time. Later, we provide a procedure, $\WitRoute(\cdot)$ (\Cref{alg:witRoute}), that uses the witness list to efficiently transfer flows from transitive arcs to flows on original arcs. Our maximum flow algorithm proceeds through a sequence of admissible extensions of the input graph, where new artificial, transitive arcs are added using the transitive cover; $\WitRoute(\cdot)$ is used in the final step to reroute flow from the artificial arcs and obtain a maximum flow in the input graph.

We call such a data structure an \emph{(incremental) transitive cover data structure} where we use the term \emph{incremental} when we wish to emphasize that arcs are only added. In the remainder of this section we introduce and formally define witness lists, and use them to define \DITC{}'s.
In \Cref{sec:data-structures}, we design \DITC{}'s in different settings, some of which use specific properties of the input arcs, e.g., that they respect the hierarchy of a given bounded-depth tree as in the bounded tree-depth case.

\paragraph{Witness list.}
We start by formally defining a witness list as follows.

\begin{definition}[Witness list]
\label{def:wit_list}
We call $L = (a_i,b_i,w_i)_{i \in [K]} \in \dV \times \dV \times \dV$ a \emph{witness list on $\dV$}, denoted $L \in \wlistSpace$, if the $K$ $(a_i,b_i)$ pairs are distinct with $a_i \neq b_i$, and for each $(a_i,b_i,w_i)$ with $w_i \neq b_i$ there are $j_i,k_i \in [i - 1]$ with $(a_i,w_i) = (a_{j_i},b_{j_i})$ and $(w_i,b_i) = (a_{k_i},b_{k_i})$. 
We define: 
\begin{itemize}
    \item \emph{The transitive arcs of $L$}: $\dELcl(L) \defeq \{(a_i,b_i) ~ | ~ i \in [K]\}$.
    \item \emph{The arcs of $L$}: $\dEL(L) \defeq \{(a_i,b_i) ~ | ~ i \in [K] , w_i = b_i\}$.
    \item \emph{The first $j$ arcs of $L$}: $\dELsub{j}(L) \defeq \{(a_i,b_i) \in \dEL(L) ~ | ~ i \in [j]\}$.
    \item \emph{$L$'s witness for $(a_i,b_i) \in \dELcl(L)$}: $\wit_L(a_i,b_i) \defeq w_i$.
    \item \emph{$L$'s walk for $(a_i,b_i)$}:
    $
    \walk_L(a_i,b_i) \defeq 
    \begin{cases}
        (a_i, b_i) & \text{if }w_i = b_i\\
        \walk_L(a_{i},w_i) , \walk_L(w_i,b_i)
        & \text{otherwise}
    \end{cases}
    $
\end{itemize}
Additionally, we say that such witness list is \emph{path-structured} if $\walk_L(a_i,b_i)$ is a path for all $i \in [K]$.
\end{definition}

A witness list, $L$, is simply a list of triples of nodes, $(a_i,b_i,w_i) \in \dV \times \dV \times \dV$ for each $i \in [K]$ that satisfies certain properties. For such a list, we refer  to the set of all $(a_i,b_i)$ pairs, denoted $\dELcl(L)$, as \emph{transitive arcs} and we refer to the subset of these pairs with $w_i = b_i$ as arcs, denoted $\dEL(L)$. In our \DITC{}s $\dEL(L)$ will correspond to arcs added to the data structure and $\dELcl(L)$ will correspond to both these arcs as well as additional arcs in the transitive cover of these arcs that it computes. 

To be a valid witness list, the $(a_i,b_i)$ need to be distinct and $L$ needs to provide a way to map each non-arc of $L$, i.e., $(a_i,b_i) \in \dELcl(L) \setminus \dE(L)$ to an $a_i$--$b_i$ path on the arcs that appear earlier than $(a_i,b_i)$ in the list, i.e., $\dELsub{i - 1}(L)$. A witness list ensures this simply by requiring that for each $(a_i,b_i) \in \dELcl(L) \setminus \dE(L)$, it is the case that $(a_i,w_i)$ and $(w_i,b_i)$ appear earlier in the list, i.e., there is $j_i,k_i \in [i - 1]$ with $(a_i,w_i) = (a_{j_i},b_{j_i})$ and $(w_i,b_i) = (a_{k_i},b_{k_i})$. We call $w_i$ the \emph{witness for $(a_i,b_i)$} and the requirement is equivalent to there being an $a_i$--$b_i$ path of length $2$ through $w_i$ on the transitive arcs before $i$ in $L$. We let $\walk_L(a_i,b_i)$ denote the result of recursively applying such a substitution and the following lemma shows that indeed, $\walk_L(a_i,b_i)$ is a $a_i$--$b_i$ walk on $\dELsub{i - 1}(L)$. Finally, we call a witness list, \emph{path-structured}, if the walks it outputs are always paths; in all  our applications, witness lists have this property.

\begin{lemma}
\label{lem:witlist_walk}
$\walk_L(a_i,b_i)$ is an $a_i$--$b_i$ walk on $\dELsub{i}(L)$ for any $L = (a_i,b_i,w_i)_{i \in [K]} \in \wlistSpace$
and $(a_i,b_i) \in \dELcl(L)$.
\end{lemma}
\begin{proof}
We prove by strong induction on $i$. For $i = 1$, \Cref{def:wit_list} implies that $w_1 = b_1$ (since $[1 - 1] = \emptyset$) and consequently, $\walk(a_1,b_1) = (a_1,b_1) \in \dELsub{1}(L)$. Now suppose the claim holds for all $i < i_* \leq K$ we prove it holds for $i_*$. If $w_{i_*} = b_{i_*}$ then,  $\walk(a_{i_*},b_{i_*}) = (a_{i_*},b_{i_*}) \in \dELsub{i_*}(L)$, by definition, and the claim follows. Otherwise, by definition,
$\walk(a_{i_*},b_{i_*}) = \walk_L(a_{j_{i_*}},b_{j_{i_*}}) , \walk_L(a_{k_{i_*}},b_{k_{i_*}})$ for $j_{i_*},k_{i_*} \in [i_* - 1]$ where $(a_{i_*},w_{i_*}) = (a_{j_{i_*}},b_{j_{i_*}})$ and $(w_{i_*},b_{i_*}) = (a_{k_{i_*}},b_{k_{i_*}})$. Consequently, by the inductive hypothesis, $\walk_L(a_{j_{i_*}},b_{j_{i_*}})$ is an $a_{i_*}$--$w_{i_*}$ walk on $\dELsub{j_{i_*}}(L) \subseteq \dELsub{i_*}(L)$ and $\walk_L(a_{k_{i_*}},b_{k_{i_*}})$ is an $w_{i_*}$--$b_{i_*}$ walk on $\dELsub{k_{i_*}} \subseteq \dELsub{i_*}(L)$. The result follows as the concatenation of a $a_{i_*}$--$w_{i_*}$ walk and a $w_{i_*}$--$b_{i_*}$ walk, both on $\dELsub{i_*}(L)$, is an $a_{i_*}$--$b_{i_*}$ walk on $\dELsub{i_*}(L)$.
\end{proof}

\paragraph{Incremental transitive cover.}
With witness lists established, we now formally define an \DITC{}.

\begin{restatable}[Incremental Transitive Cover]{definition}{defGeneralTransitiveCover}
\label{def:general_transitive_cover}
A data structure is an \emph{incremental transitive cover (\DITC{})} data structure if it supports the following operations.
    \begin{itemize}
        \item $\dsTcInit(\dV)$: initializes with finite set of $n$ nodes $\dV$ and sets $\dE \gets \emptyset$.

        \item $\dsTcAdd(F\subseteq \dV\times \dV)$: 
        for $F$ with $F \cap \dE = \emptyset$\footnote{The assumption that $F \cap \dE = \emptyset$ was made to simplify theorem statements. The assumption can be removed by increasing the time of implementing each $\dsTcAdd(\cdot)$ by an additive $O(|F|)$ (by simply maintaining $\dE$ and in $O(|F|)$ time remove the elements that are already in $\dE$). 
        } where no arc in $F$ is a self-loop,
        sets $\dE \gets \dE\cup F$.%
        
        \item $\dsTcCover(S\subseteq \dV)$: outputs graph $G_{\mathrm{out}} = (\dV_{\mathrm{out}},\dE_{\mathrm{out}})$ that is a transitive cover of $S$ in $(\dV,\dE)$.

        \item $\dsTcWit()$: returns a path-structured witness list $L = \{(a_i,b_i,w_i)\}_{i \in [K]} \in \wlistSpace$ where $\dEL(L) \subseteq \dE$ and if $(a,b)$ was output in $\dsTcCover(\cdot)$, then $(a,b) = (a_i,b_i)$ for some $i \in [K]$ and $\dELsub{i}(L)$ was added to $\dE$ before $(a,b)$ was output.
    \end{itemize}
\end{restatable}

An \DITC{} supports $\dsTcAdd(\cdot)$ for adding arcs and $\dsTcCover(\cdot)$ for computing transitive covers, as described earlier. In addition, it supports the $\dsTcWit()$ operation which returns a path-structured witness list corresponding to the arcs added in  $\dsTcAdd(\cdot)$ and additional transitive arcs computed in handling $\dsTcCover(\cdot)$ operations. In particular, the arcs in the witness list $L$ must returned by  $\dsTcWit()$ must be arcs in the associated  \DITC{} instance and any arc $(a,b)$ output in $\dsTcCover()$ must appear in the output witness list with the arcs that $\walk_L(a,b)$ uses having been added in $\dsTcAdd(\cdot)$ before $(a,b)$ was output (this encoded through the condition that ``$\dELsub{i}(L)$ was added to $\dE$ before $(a,b)$ was output'' in \Cref{def:general_transitive_cover}). These conditions allow flow on any $(a,b)$ arc output by $\dsTcCover()$ to be mapped to a flow on $a$--$b$ path on arcs that were added to $\dsTcAdd(\cdot)$ earlier. We perform this transfer through the procedure $\WitRoute(\cdot)$ which is described and analyzed in \Cref{sec:witroute}; this self-contained procedure efficiently transfer flows on the transitive arcs of a witness list to a flow using only the arcs of the witness list.%

\section{Arc Types and Valid Iterates}
\label{sec:maxflow_framework}

Throughout, we use $n\defeq|\gVi|$ and $m\defeq|\gEi|$ for a fixed input instance $(\gGi,\ui)$ with $\gGi=(V,\gEi)$. 
 During the algorithm, new arcs may be added to $\gGi$. In particular, our `working instance' will be an admissible extension $(G,u)=(\gGi,\ui)\oplus(\Fone,\Ftwo)$.

\subsection{Abundant and free arcs}
\label{sec:abundant}

Given an  instance $(G,u)$ and
 $\varepsilon\ge 0$, let $f$ be  an $\varepsilon$-optimal flow $f$ in $(G,u)$. Our next definitions are motivated by the following observation: if $\res{f}_e>\varepsilon$ for some $e\in\gE$, then \Cref{lem:res-value} implies that $e$ is not an element of any $s$--$t$ min-cut. If both $\res{f}_e,\res{f}_{\rev{e}}>\varepsilon$, then neither $e$ nor $\rev{e}$ may be an element of  any $s$--$t$ min-cut. Equivalently, even if we increase the capacities of these arcs to $\infty$, it does not affect the maximum flow value. %

 Thus, arcs with large residual capacities are virtually uncapacitated, and they also impose restrictions on the possible sets of $s$--$t$ min-cuts. 
 We call arcs with residual capacity at least $4n^2\varepsilon$ \emph{abundant};  our algorithm maintains  the property that once an arc becomes abundant, it remains so for the rest of the algorithm.  
 The reason for choosing the higher threshold  $4n^2\varepsilon$ instead of $\varepsilon$ is that whereas the existence of a maximum flow $f^\star$ with $\|f^\star-f\|_\infty\le \varepsilon$ is guaranteed for every $\varepsilon$-optimal $f$ (\Cref{lem:safe-cap}), our algorithm 
 will only guarantee  $\|f^\star-f\|_\infty\le 4n^2\varepsilon$  for the current $\varepsilon$-optimal flow $f$ and the final optimal flow $f^\star$ returned.\footnote{The factor $4n^2$ was chosen to leave enough slack for the flow rerouting arguments. It could be tightened; however, this would not affect the asymptotic running times in our analysis.}

\begin{definition}[Arc Types] \label{def:abundant}
For an $\varepsilon$-optimal flow $f$ in instance $(G,u)$
and $\varepsilon\ge 0$ we define the 
set of \emph{$(f,\varepsilon)$-abundant arcs} $\gEA(f,\varepsilon)\subseteq \gE$ and the set of \emph{$(f,\varepsilon)$-free arcs} $\gEF(f,\varepsilon)\subseteq \gE$ as
\[
\gEA(f,\varepsilon)\defeq \left\{e\in \gE:\, \res{f}_e\ge \ab \varepsilon\right\}
\text{ and }
\gEF(f,\varepsilon)\defeq \left\{e\in \gE:\, e, \rev{e}\in\gEA(f,\varepsilon) \right\}
\text{ where }
\ab\defeq 4n^2
\, .
\]
\end{definition}
When $f$ and $\varepsilon$ are clear from the context, for brevity, we simply refer to $\gEA(f,\varepsilon)$ and $\gEF(f,\varepsilon)$ as the abundant and free arcs respectively. 
As already noted, \Cref{lem:res-value} implies that no $(f,\varepsilon)$-abundant arc is in a min-cut in $(G,u)$. The next proposition asserts a stronger claim. 

\begin{proposition}\label{prop:abundant-min-cut}
Let $(G,u)$ be an instance with an $\varepsilon$-optimal flow $f$, where 
 $\varepsilon\ge 0$
  and let $S$ be an $s$--$t$ cut. 
 If $\res{f}(S,\gV\setminus S)< \ab\varepsilon$, then  $\dout{\gE}{S}\cap \gEA(f,\varepsilon)=\emptyset$. In particular, 
 there is no $s$--$t$ path in $\gEA(f,\varepsilon)$.
 Furthermore,
if $S$ is an $\varepsilon'$-optimal  $s$--$t$ cut for $\varepsilon'\ge 0$, %
 then $\res{f}(S,\gV\setminus S)\le\varepsilon'+\varepsilon$. 
\end{proposition}
\begin{proof}
The first part is immediate by the definition of abundant arcs, and the second part by the first. For the last part, note that by the definition of residual capacities, 
\[
\begin{aligned}
\nu(G,u)+\varepsilon'&\ge
u(S,\gV\setminus S)=\res{f}(S,\gV\setminus S)+f(\dout{E}{S})-f(\din{E}{S})\\
&= \res{f}(S,\gV\setminus S)+\val{f}\ge \res{f}(S,\gV\setminus S)+\nu(G,u)-\varepsilon\,.\quad\quad\quad\quad\qedhere
\end{aligned}
\]
\end{proof}

Throughout our algorithm, if $f$ is an $\varepsilon$-optimal flow at a current iteration in $(G,u)$, then for all subsequent iterations $f'$, we will have $|f_e-f_{\rev{e}}-f'_e+f'_{\rev{e}}|\le 2n\varepsilon$ for all $e\in E(G)$. Thus, the following simple lemma
guarantees that once an arc becomes abundant in our algorithm, it will remain abundant in all subsequent iterations; the proof is immediate.
\begin{lemma}\label{lem:stays-abundant}
Let $(G',u')$ be any admissible extension of the instance   $(G,u)$, and $0\le \varepsilon'\le \varepsilon/2$. Let $f$ be an $\varepsilon$-optimal flow in  $(G,u)$ and let $f'$ be a $\varepsilon'$-optimal flow in  $(G',u')$ such that  $|f_e-f_{\rev{e}}-f'_e+f'_{\rev{e}}|\le 2n\varepsilon$ for all $e\in E(G)$. Then, $\gEA(f,\varepsilon)\subseteq \gEA(f',\varepsilon')$. 
\end{lemma}

For an instance $(G,u)$ and an $\varepsilon$-optimal flow $f$, we let 
$(V,\gEA^\tr(f,\varepsilon))$ denote the transitive closure of $(V,\gEA(f,\varepsilon))$.
 That is, $(i,j)\in\gET(f,\varepsilon)$ if and only if $\gEA(f,\varepsilon)$ contains an $i$--$j$ path. 
 \begin{lemma}\label{lem:add-abundant}
Let $(G,u)$ be instance with an $\varepsilon$-optimal flow $f$ for $\varepsilon\ge 0$. Then, $(G,u)\oplus(\Fone,\Ftwo)$ is an admissible extension of  $(G,u)$ for any  $\Fone\subseteq \gET(f,\varepsilon)$ and $\Ftwo\subseteq\{e\,:\,e,\rev{e}\in \gET(f,\varepsilon)\} $.
\end{lemma}
\begin{proof}
The claim follows by showing that for 
any $s$--$t$ min-cut $S$ in $(G,u)$, $\dout{\Fone}{S}=\emptyset$ and $\dtot{\Ftwo}{S}=\emptyset$. 
 For a contradiction, assume $e\in \Fone$ leaves $S$; by the definition of the transitive closure, there must be an abundant arc $e'\in \gEA(f,\varepsilon)$ also leaving $S$. This contradicts \Cref{prop:abundant-min-cut}. The same argument holds for arcs in $\Ftwo$.
\end{proof}

\paragraph{Free components.} 
We now define a key concept, \emph{$(f,\varepsilon)$-free components}.
These will correspond to a coarsening sequence of partitions 
over the iterations in the algorithm, 
with the guarantee that every $s$--$t$ min-cut is the union of partition components.
\begin{definition}[Free components]\label{def:free-comps}
Let $f$ be an $\varepsilon$-optimal flow in
the instance
 $(G,u)$ for $\varepsilon\ge0$. We say that $P\subseteq V$ is a 
\emph{$(f,\varepsilon)$-free component} if
\begin{itemize}
\item $P$ is the set of nodes reachable on an abundant path in $\gEA(f,\varepsilon)$ from $s$; or
\item  $P$ is the set of nodes that can reach $t$ on an abundant path in $\gEA(f,\varepsilon)$; or
\item $P$ is a connected component of $\gEF(f,\varepsilon)$ disjoint from the above two components.
\end{itemize}
We denote this partition of $V$ as 
\[
\CP_{f,\varepsilon}\defeq\{P_i\,:\, i\in \roots\}\, ,
\]
where the index set $\roots\subseteq V$ will be a set of \emph{roots}, picking one representative node from each component such that $i\in P_i$. We always let $s$ and $t$ be the root of the first and second components listed above, respectively.
\end{definition}
Note that \Cref{prop:abundant-min-cut} guarantees that  $P_s\neq P_t$.
 For a node $j\in V$, we let $\Root(j)$ denote the root of the component containing $j$. Thus, $j\in P_{\Root(j)}$.
 Note that the indexing is not consecutive. When two components $P_i$ and $P_j$ are merged during the algorithm into a new component  $P_i\cup P_j$, the new component will have root node $i$ or $j$ and will  be denoted as   $P_i$ or $P_j$.
When clear from the context, we simply refer to the $P_i$ sets as free components.
We let $\gEP(f,\varepsilon)\subseteq \gE$ denote the set of arcs between different components. In particular, $\gEP(f,\varepsilon)\cap \gEF(f,\varepsilon)=\emptyset$.

\subsection{Valid iterates and essential arcs}
We now formalize properties of tuples $(f,\nr,
\varepsilon,\roots)$ of a flow, safe capacity vector, accuracy bound, and root set appearing in the algorithm. 
 We define the set of \emph{gap arcs} as the set of arcs between different components
  where the residual capacity is greater than the safe capacity by at least $2\ab^{-2}\varepsilon$ in both directions: 
\[
\begin{aligned}
\gEStr(f,\nr,\varepsilon)\defeq&\left\{e\in \gEP(f,\varepsilon)\, :\, 
 \res{f}_e\ge \nr_e+2\ab^{-2}\varepsilon\, \, \mbox{and}\, \,  \res{f}_{\rev{e}}\ge \nr_{\rev{e}}+2\ab^{-2}\varepsilon\right\}\, .
\end{aligned}
\]
In our algorithm, all such arcs will become free arcs in the subsequent iteration. 

For roots $i\in\roots$, the arcs $(s,i)$ and $(i,t)$ will play a special role  in the algorithm: we use them to restore flow balance of  component $P_i$ after arc adjustments. (Recall assumption~\eqref{eq:i-s-t} that all such arcs are assumed to be present in the graph.) We call these \emph{boundary arcs}:
\[
\gEbd(\roots)\defeq \left\{(s,i),(i,s),(t,i),(i,t)\, :\, i\in\roots\setminus\{s,t\}\right\}\, .\]
 For a subset $\roots'\subseteq\roots$, we will use $\gEbd(\roots')$ to denote the corresponding subset.

The invariants in the following definition will be satisfied at the beginning of each iteration. Throughout,  $(G, u)$ will be a  valid extension of the input instance $(\gGi, \ui)$ and  may include extension arcs.

\begin{definition}[Valid Iterate] \label{def:proper}
We say that $(f,\nr,
\varepsilon,\roots)$ is a \emph{valid iterate} in instance $(G,u)$ if
\begin{enumerate}[(A)]
\item\label{it:eps-opt} $f$ is an $\varepsilon$-optimal feasible flow in $(G,u)$;
\item\label{it:safe}  $\nr$ are  safe capacities for $f$. %
\item\label{it:bounds} For every arc $e\in \gEP(f,\varepsilon)$, at least one of the following hold:  $\res{f}_e=0$;  $\res{f}_{\rev{e}}=0$; $e\in\gEStr(f,\nr,\varepsilon)$;  or $e\in\gEbd(\roots)$.
\item\label{it:oneway}  $\res{f}_{si}=0$ or $\res{f}_{it}=0$ for $i\in\roots$.
\item\label{it:roots} $\roots$ forms a set of roots of the free components, i.e., every component of $\CP_{f,\varepsilon}$ contains exactly one node from $\roots$, and $s,t\in \roots$.
\end{enumerate}
\end{definition}
Property \eqref{it:bounds} is the following strong requirement: for any pair of non-boundary arcs $e,\rev{e}$ between different free components, the flow should be either at a capacity bound, or it should be a gap arc, that is, the safe capacity should be less than the residual capacity on both sides by at least $2\ab^{-2}\varepsilon$. As shown below in \Cref{lem:strict-to-free}, all gap arcs will become free in the subsequent iteration.
\begin{definition}[Essential arcs and roots]\label{def:pot-function}
Let $f$-be an $\varepsilon$-optimal flow in $(G,u)$ for $\varepsilon\ge 0$. 
We say that $e\in E$ is an \emph{essential arc} if $e\in \gEP(f,\varepsilon)$ and 
\[
\essentialthreshold\le \res{f}_e< \ab\varepsilon\quad \mbox{or}\quad \essentialthreshold\le \res{f}_{\rev{e}}< \ab\varepsilon\, .
 \]
  Given the root set $\roots$, $i\in\roots$ is  an \emph{essential root} if there is an essential arc incident to the component $P_i$, i.e., there exists an essential arc $e=(j,k)$ with 
   $i=\Root(j)$.
\end{definition}
The following lemma is immediate  by the definition of essential arcs, part~\eqref{it:bounds} in \Cref{def:proper}, and by observing that $u_e+u_{\rev{e}}=\res{f}_e+\res{f}_{\rev{e}}$ holds for every arc.
\begin{lemma}\label{lem:improved-opt-1}
 Let
$(f,\nr,\varepsilon,\roots)$ be a valid iterate in $(G,u)$. 
 For every essential arc  $e$, either $e\in \gEStr(f,\varepsilon,r)$, or $e\in \gEbd(\roots)$, or 
$\essentialthreshold\le u_e+u_{\rev{e}}< \ab\varepsilon$.
\end{lemma}

\subsection{Series of valid iterates}
Our next definition formalizes the proximity property subsequent valid iterates will satisfy.

\begin{definition}[Valid successors] \label{def:successor}
Let $(G',u')=(G,u)\oplus (\Fone,\Ftwo)$ be an admissible extension of $(G,u)$.
Let $(f,\nr,\varepsilon,\roots)$ be a valid iterate in $(G,u)$, and let $(f',\nr',\varepsilon',\roots')$ be a valid iterate in $(G',u')$. We say that $(f',\nr',\varepsilon',\roots')$ is a \emph{valid successor} of $(f,\nr,\varepsilon,\roots)$ if the following hold:
\begin{enumerate}[(A)]
\item\label{it:eps-decrease} 
$\varepsilon'\le \ab^{-3}\varepsilon$.
\item\label{it:arc-prox} For every $e\in E(G)$,
$f'_e-f_e+f_{\rev{e}}\le \min\{\nr_e,\varepsilon\}+\ab^{-2}\varepsilon<2\varepsilon$,
and
for every $e\in E(G')\setminus E(G)$, $f'_e\le 3\varepsilon$.
\end{enumerate}
\end{definition}
The above definition guarantees that all gap arcs will become free in the next iterate, as shown next. %

\begin{lemma}\label{lem:strict-to-free}
Let  $(f,\nr,\varepsilon,\roots)$ be a valid iterate in $(G,u)$.
Let $(G',u')$ be an admissible extension of $(G,u)$ with a valid iterate $(f',\nr',\varepsilon',\roots')$ that is a valid successor of $(f,\nr,\varepsilon,\roots)$.
\begin{enumerate}[(i)]
\item\label{it:free-free} $\gEA(f,\varepsilon)\subseteq \gEA(f',\varepsilon')$ and $\gEF(f,\varepsilon)\subseteq \gEF(f',\varepsilon')$. 
\item\label{it:gap-free} Every arc $e\in \gEStr(f,\nr,\varepsilon)$ is $(f',\varepsilon')$-free.
\end{enumerate} 
\end{lemma}

\begin{proof}
Part {\em(i)} is immediate by \Cref{lem:stays-abundant}.
For part {\em(ii)}, consider an arc $e\in \gEStr(f,\nr,\varepsilon)$. If $u_e=\infty$, then we immediately have $e\in \gEA(f',\varepsilon')$. If $u_e<\infty$, then by \Cref{def:successor}\eqref{it:arc-prox} and by the definition of gap arcs, we see that
\[
\res{f'}_e\ge 
\res{f}_e-\nr_e-\ab^{-2}\varepsilon\ge \ab^{-2}\varepsilon\ge\ab\varepsilon'\, .\]
Thus, $e\in \gEA(f',\varepsilon')$. A similar argument shows that $\rev{e}\in \gEA(f',\varepsilon')$, and consequently, $e\in \gEF(f',\varepsilon')$, as required.
\end{proof}

\section{The Maximum Flow Meta-Algorithm}\label{sec:max-flow}

\SetCommentSty{algCommentFont}

\begin{algorithm}[htb!]
    \caption{\textsc{Max-Flow}}\label{alg:max-flow}
    \KwData{An instance $(\gGi,\ui)$ with $\gGi=(\gV,\gEi)$ and $\nu(\gGi,u)<\infty$.}
    \KwResult{A maximum flow $f^\star$ in  $(\gGi,\ui)$.}
    $(G,u)\gets (\gGi,\ui)$ ; $E\gets \gEi$ ; 
    $\Fe\gets\emptyset$ ; $\Fs\gets\emptyset$ \; %
    $(f,\nr,\varepsilon,\roots)\gets\Initialize$ \tcp*{Obtain initial valid iterate}  
    ${\cal T}_i\gets(\{i\},\emptyset)$ for all $i\in V$ \;
   $\dsTcInit(V)$ ; $\dsTcAdd(\gEA)$    \tcp*{Initialize transitive cover data structure} 
    \While{$\varepsilon>0$}{
       $\Eaux\gets\left\{e\in \gEP\,:\,  \essentialthreshold\le \res{f}_e< \ab\varepsilon\right\}$\tcp*{Essential arcs}
   $\Vess\gets$ roots of components with essential arcs incident\; 
   \lIf{$\{e\in \gEP\,:\,\res{f}_e<\essentialthreshold\}\neq\emptyset$}{$\delta_1\gets n^2\max\left\{\res{f}_e\,:\, e\in \gEP, \res{f}_e<\essentialthreshold\right\}$}\lElse{$\delta_1\gets 0$ }
      $(\cG,\cu,H)\gets\CreateCompact(\Vess,\Eaux)$ \;
    \If{$s,t\in\Vess$}{
   $(y,e^\star)\gets \approxFlow(\cG,\cu,\Gamma^{5})$\tcp*{see Definition~\ref{def:approx-solver}}
   $\delta_2\gets |E(\cG)| \cu^y_{e^\star}$ \;
        \lFor{$e\in \Eaux$}{$\SendFlow(e,y_e)$ \tcp*[f]{Add $y$ to flow} \label{line:sendflow}}
    \lFor{$e=(i,j)\in H$}{
        $E\gets E\uplus\{e,\rev{e}\}$ ;
        $u_e\gets\infty$ ; $u_{\rev{e}}\gets 0$ ; $f_e\gets y_e$ ; $f_{\rev{e}}\gets 0$  \label{line:sendflow2}}
   }%
   \lElse{$\delta_2\gets 0$}
       $\delta\gets\delta_1+\delta_2$ \tcp*{current flow is $\delta$-optimal}
    $\PostProcess(\delta,(\Eaux\setminus\gEbd)\uplus H)$ \;
    $E\gets E\setminus\{e,\rev{e}\,:\,e\in H, f_e=0\}$ \tcp*{Keep only new arcs with positive flow} 
    $\Fe\gets \Fe\cup (H\cap\supp(f))$ \;  

      $\varepsilon\gets 2n^2\delta$ \;
        Update $\gEA$, $\gEF$, $\CP$, $\gEP$ and $\roots$ \; 
    $J\gets$ set of newly abundant arcs ;
    $\dsTcAdd{(J\setminus H)}$   \tcp*{Update transitive cover} 
  \lFor{$k$ removed from $\roots$}{$\RemoveRoot(k)$}
        }
    $\RerouteS(\Fs)$ \;
    $L\gets\dsTcWit(\cdot)$ \;
    $h\gets \WitRoute(L,f|_{\gEA},\rev{H})$ \tcp*{Map flow back to original arcs} 
     \lFor{$e\in\gEi$}{
        $f^\star_e\gets h_e$ if $e \in \gEA$ and $f^\star_e\gets f_e$ if $e \notin \gEA$}
  \Return{$f^\star$}
\end{algorithm}

Algorithm~\ref{alg:max-flow} takes a bounded input instance $(\gGi,\ui)$ with $\gGi=(\gV,\gEi)$.  In every iteration it maintains an admissible extension $(G,u)=(\gGi,\ui)\oplus (\Fe,\Fs)$; we let $E=E(G)$. 
Here, $\Fe$ is a set of newly added arcs with $u_e=\infty$, $u_{\rev{e}}=0$; these will be called  \emph{extension arcs}. The arcs $\Fs$ have $u_e=u_{\rev{e}}=\infty$ and will be called \emph{shortcut arcs}; their role will be explained below. The node set remains $V(G)=V$ throughout.

 At the beginning of every iteration, we have a valid iterate $(f,\nr,\varepsilon,\roots)$. The iterate at the beginning of every iteration will be a valid successor of the iterates from all previous iterations.
A crucial subroutine is the approximate flow solver $\approxFlow(\cdot)$ as in Definition~\ref{def:approx-solver}.

Throughout, we maintain the partition $\CP=\CP(f,\varepsilon)$, 
the arc sets $\gEP=\gEP(f,\varepsilon)$ and $\gEA=\gEA(f,\varepsilon)$.
We also maintain  pointers  $\Root(i)\in\roots$ such that $i\in P_{\Root(i)}$.
In accordance with \Cref{lem:strict-to-free},
 the set of abundant and free arcs may only be extended, and consequently,  $\CP$ may only become coarser. 

 In each of the components, we maintain a spanning subgraph  ${\cal T}_i$ formed by abundant arcs. If $i\notin\{s,t\}$, then ${\cal T}_i$ will be a bidirected spanning tree formed by abundant arcs. That is,  $(j,k)\in {\cal T}_i$  if and only if  $(k,j)\in {\cal T}_i$, and the pairs of reverse arcs form a spanning tree. For $i=s$, ${\cal T}_s$  contains an out-arboresence of abundant arcs from $s$, and the reverse abundant arcs of some of these arcs.
 Similarly, for  $i=t$, ${\cal T}_t$  contains an in-arboresence of abundant arcs to $t$, and the reverse abundant arcs of some of the arcs.

In Section~\ref{sec:update-abundant}, we show how the components and spanning subgraphs can be efficiently maintained throughout the algorithm in total time $\tilde O(m)$.

 We also maintain a transitive cover data structure on $(V,\gEA)$ as in Definition~\ref{def:general_transitive_cover}  with subroutines $\dsTcInit(\cdot)$, $\dsTcAdd(\cdot)$, $\dsTcCover(\cdot)$, and $
\dsTcWit(\cdot)$. These subroutines will be given in Section~\ref{sec:data-structures}. We will use different implementations for different subclasses of graphs. The subroutine $\WitRoute{}$ is described in Section~\ref{sec:witroute}.

We start by explaining some additional concepts and invariants used in the algorithm, and then give a high level overview. Section~\ref{sec:main-alg-subroutines} presents the main subroutines.

\paragraph{Shortcut arcs.}
Throughout, for each node $i\notin\roots$, we keep arcs $(i,\Root(i))$ in the shortcut arc set $\Fs$ with $u_e=u_{\rev{e}}=\infty$. These arcs enable sending flow directly between a node and the root node of its component. Adding them does not change the flow value of the instance, since they are inside free components (see Lemma~\ref{lem:add-abundant}). Moreover, these arcs will not be added to the transitive cover data structure, and hence they do not need to be consistent with the specific graph structure such as bounded tree-depth.

These new arcs will be added when updating the components $\CP_{f,\varepsilon}$. When two components $P_i$ and $P_j$ are merged, and $|P_i|\ge |P_j|$, then we pick $i$ to be the root of $P_i\cup P_j$, unless $j\in\{s,t\}$. 
 We add new infinite capacity $(k,i)$ and $(i,k)$ arcs for all $k\in P_j$. The size of the component containing a node $k$ doubles each time $\Root(k)$ changes until $\Root(k)\in\{s,t\}$, after which point it will not change anymore. Consequently,  $|\Fs|=O(n\log n)$ throughout the algorithm.
\paragraph{Algorithm overview.}
The algorithm starts with the initialization subroutine 
$\Initialize$ that returns an initial valid iterate
$(f,r,\varepsilon,\roots)$ with $\roots=V$. Note that simply setting $f\equiv \0$  may violate Definition~\ref{def:proper}\eqref{it:bounds}. We define the initial flow such that $\res{f}_e=0$ or $\res{f}_{\rev{e}}=0$ for all arcs not incident to $s$ or $t$. This, as well as the other subroutines, are described in \Cref{sec:main-alg-subroutines}.
Throughout, the instance $(G,u)$ and valid iterate $(f,r,\varepsilon,\roots)$ are global variables that can be accessed and changed in the subroutines.%

The main while cycle of the algorithm runs until $\varepsilon=0$. At this point, $f$ is a maximum flow in the current instance $(G,u)$. We then transform it to a maximum flow in the input instance $(\gGi,\ui)$ as follows. First, 
 $\RerouteS(\Fs)$ reroutes the flow from the shortcut  arcs to the spanning subgraphs ${\cal T}_i$ inside each component. We call $\dsTcWit(\cdot)$ to create a path-structured witness list $L$ with $\dEL(L) \subseteq \gEA$ (see \Cref{def:wit_list}). The subroutine $\WitRoute(.)$ uses this to reroute the flow from the extension arcs to original arcs. We thus obtain a maximum flow $f^\star$ in the input instance $(\gGi,\ui)$. 

We now describe the main while cycle. We
 let $\Eaux$ and $\Vess$ denote the sets of essential arcs and essential roots as in Definition~\ref{def:pot-function}. We define $\delta_1$ as $n^2$ times the largest residual  arc capacity below the essential threshold $\essentialthreshold$.

The subroutine $\CreateCompact(\Vess,\Eaux)$ returns
an auxiliary instance $(\cG,\cu)$ along with   a special 
 arc set $H\subseteq E(\cG)$. 
We 
 call the transitive cover data structure to obtain a transitive cover  $(V',H)$ of $\Vess$. The node set will be $\cV=V'\cup V(\Eaux)$, i.e., also including the endpoints of all essential arcs. 
 The arc set of $\cG$ comprises of $H$, $\Eaux$, as well as the shortcut arcs between non-root nodes $i\in \cV$ and the corresponding $\Root(i)$.

We call the approximate solver $\approxFlow(\cG,\cu,\ab^5)$ to obtain an approximate flow $y$ and $e^\star$ such that $y$ is $\delta_2$-optimal in $(\cG,\cu)$ and $\delta_2\le \ab^{-5}\nu(\cG,\cu)$ for $\delta_2=|E(\cG)|{\cu}^{f}_{e^\star}$.
We (temporarily) add the arc set $H$ to our instance, and add $y$
to the original flow $f$. In case $s$ or $t$ was not in $\Vess$, we skip these steps,  keep $f$ unchanged, and set $\delta_2=0$. In both cases, we set $\delta=\delta_1+\delta_2$; the current flow $f$ will be $\delta$-optimal.

We now call $\PostProcess(\delta,(\Eaux\setminus\gEbd)\uplus H)$. This is needed in order to restore the invariants in Definition~\ref{def:proper}. In particular, for  every arc $e\in (\Eaux\setminus\gEbd)\uplus H$, we change the flow and set safety capacities to satisfy either $\res{f}_e=0$, or $\res{f}_{\rev{e}}=0$, or $\res{f}_e>\nr_e+\delta$ and $\res{f}_{\rev{e}}>\nr_e+\delta$.
We  set the new value $\varepsilon=2n^2\delta$.

We then remove all new arcs $e\in H$ from $E$ that do not carry flow, i.e., $f_e=0$. We update the extension arc set $\Fe$ by adding $H\cap\supp(f)$ as the set of new arcs kept in the instance.
We  update the set of abundant and free arcs, and free components. In the transitive cover data structure, we add all new abundant arcs outside $H$ (note that every $e\in H$ was itself returned by the data structure). 
The details of how these can be efficiently maintained are deferred to  Section~\ref{sec:update-abundant}.
For all nodes $k$ that are removed from $\roots$, we call the subroutine $\RemoveRoot(k)$. Let $j$ be the new root of $k$. We add  new shortcut arcs $(i,j)$  to $\Fs$ for every node $i$ whose root changed from $k$ to $j$. If $0<\res{f}_{sk}<2\varepsilon$ then  $(s,k)$ gets saturated, and if $0<\res{f}_{kt}<2\varepsilon$ then $(k,t)$ gets saturated. Note that arcs with $\res{f}_{sk}>2\varepsilon$ or $\res{f}_{kt}>2\varepsilon$ will become gap arcs in the next iteration.

\subsection{Main subroutines}
\label{sec:main-alg-subroutines}

\begin{algorithm}[p!]
    \caption{Subroutines in Algorithm~\ref{alg:max-flow}}\label{alg:subroutines}
     \myproc{$\SendFlow(e,\alpha)$}{
        $\alpha'\gets\min\{f_{\rev{e}},\alpha\}$ \;
        $f_{\rev{e}}\gets f_{\rev{e}}-\alpha'$ ; $f_e\gets f_{\rev{e}}+\alpha-\alpha'$
        }
         \myproc{$\Saturate(e=(i,j))$}{
        $\delta\gets \res{f}_e$ ; $f_e\gets u_e$ ; $f_{\rev{e}}\gets 0$ \;
        $p\gets\Root(i)$ ; $q\gets\Root(j)$ \;
        \lIf{$p\neq i$}{$\SendFlow((p,i),\delta)$}
        \lIf{$q\neq j$}{$\SendFlow((j,q),\delta)$}
        \lIf{$p\notin \{s,t\}$}{
            $\SendFlow((s,p),\min\{\delta,\res{f}_{sp}\})$ and $\SendFlow((t,p),\max\{0,\delta-\res{f}_{sp}\})$            }
        \lIf{$q\notin \{s,t\}$}{
            $\SendFlow((q,t),\min\{\delta,\res{f}_{qt}\})$ and $\SendFlow((q,s),\max\{0,\delta-\res{f}_{qt}\})$    
            }
        }

    \myproc{$\CreateCompact(\Vess,\Eaux)$}{
        $(V',H)\gets\dsTcCover(\Vess)$ \;
        $\cV\gets V(\Eaux)\cup V'$ \;
       $\cE\gets\Eaux$ \;%
       \lFor{$e\in\cE$}{
        $\cu_{e}\gets \nr_e$ }
       \lFor{$e\in H$}{ 
       $\cE\gets \cE\uplus \{e,\rev{e}\}$ ;
       $\cu_e\gets\infty$ ; $\cu_{\rev{e}}\gets 0$ }
       \For{$i\in V(\Eaux)\setminus V'$}
       {$\cE\gets\cE\uplus\left\{(i,\Root(i)),(\Root(i),i)\right\}$ ; $\cu_{i\Root(i)}\gets\infty$ ; $\cu_{\Root(i)i}\gets\infty$ }
        \Return{$(\cG,\cu,H)$ where $\cG = (\cV,\cE)$.}
        }

   \myproc{$\PostProcess(\delta,E')$}{
        \For{$e\in E'$}{
         \If{$0<\res{f}_{\rev{e}}<2\delta$ and $\res{f}_{{e}}>0$}
         {$\Saturate(\rev{e})$ \;
            $\nr_e\gets \min\{\res{f}_e,3\delta\}$ ; $\nr_{\rev{e}}\gets 0$
         }
         \lElse{$\nr_e\gets\min\{\res{f}_e,\delta\}$ }
        }
 }   

\end{algorithm}

\begin{algorithm}[htb!]
    \caption{Additional subroutines in Algorithm~\ref{alg:max-flow}}\label{alg:subroutines2}

  \myproc{$\RemoveRoot(k)$}{ %
  $j\gets \Root(k)$ \tcp*{new root}
  \For{$i$ whose root changed from $k$ to $j$}{
    $E\gets E\cup\{(i,j),(j,i)\}$ ; $\Fs\gets \Fs\cup\{(i,j)\}$ \;
       $u_{ij},u_{ji},r_{ij},r_{ji}\gets\infty$ ; $f_{ij},f_{ji}\gets0$  \;
}
     \lIf{$0<\res{f}_{sk}<2\varepsilon$}{$\Saturate((s,k))$}
\lIf{$0<\res{f}_{kt}<2\varepsilon$}{$\Saturate((k,t))$}
\lIf{$j\notin\{s,t\}$}{
  Join the tree ${\cal T}_k$ to ${\cal T}_j$ by a pair of abundant arcs between them}
  \lIf{$j=s$}{Merge  ${\cal T}_s$ and ${\cal T}_k$ by adding an abundant arc from ${\cal T}_s$ to ${\cal T}_k$ }
    \lIf{$j=t$}{Merge  ${\cal T}_t$ and ${\cal T}_k$ by adding an abundant arc from ${\cal T}_k$ to ${\cal T}_t$ }
 }

  \myproc{$\RerouteS(F)$}{ %
\lFor{$i\in V$}{$d(i)\gets 0$}
\For{$e=(i,j)\in F$}{$d(i)\gets d(i)+f_e-f_{\rev{e}}$ \; $d(j)\gets d(j)-f_e+f_{\rev{e}}$ \;}
\For{$k\in\roots$}{
  Select a leaf $i$ of the underlying undirected tree of ${\cal T}_k$ \;
 \If{$d(i)>0$ and $k\neq s$}{Let $e=(i,j)$ be an outgoing arc in  ${\cal T}_k$ \; $f_e\gets f_e+d(i)$ ; $d(j)\gets d(j)+d(i)$ ;}
\lIf{$d(i)>0$ and $k=s$}{$f_{is}\gets f_{is}+d(i)$  }
 \If{$d(i)<0$ and $k\neq t$}{Let $e=(j,i)$ be an incoming arc in  ${\cal T}_k$ \; $f_e\gets f_e-d(i)$ ; $d(j)\gets d(j)-d(i)$ ;}
\lIf{$d(i)<0$ and $k=t$}{$f_{ti}\gets f_{ti}-d(i)$  }
  Remove the leaf $i$ and the incident arcs in ${\cal T}_k$ and iterate \;
}
 }
 \end{algorithm}

\paragraph{Initialization.} 
The subroutine $\Initialize$
 outputs  a valid iterate $(f,\nr,\varepsilon,\roots)$.
 In order to satisfy Definition~\ref{def:proper}\eqref{it:bounds}, we initialize in a way that either $\res{f}_e=0$ or $\res{f}_{\rev{e}}=0$ for all arcs not incident to $s$ or $t$.

Namely, for every arc $e\in E(G)$ not incident to $s$ or $t$ with $u_e=\infty$, we set $f_e=0$ and $f_{\rev{e}}=u_{\rev{e}}$. Recall that we require that at least one of $u_e$ and $u_{\rev{e}}$ is finite for every arc in an instance; thus, this applies for at most one of $e$ and $\rev{e}$ for each pair of arcs.
For all pairs of arcs with $u_e,u_{\rev{e}}<\infty$, we pick arbitrarily one arc in the pair, say $e$, and set $f_e=0$, $f_{\rev{e}}=u_{\rev{e}}$.

Finally, for each $i\in V(G)\setminus\{s,t\}$, we set $f_{si}=u_{si}$, $f_{it}=u_{it}$, and set at most one of the flow values $f_{ti}$ and $f_{is}$ to nonzero so that flow conservation is satisfied at $i$.
 We set $\varepsilon=\sum_{e\in E(G): u_e<\infty} u_e$ as the sum of all finite capacities and  $\nr=\res{f}$.

  By this choice of $\varepsilon$, the only abundant arcs those with $u_e=\infty$. By the assumption that at most one of $e$ and $\rev{e}$ can be infinite, there are no free arcs after initialization. and also no abundant arcs leaving $s$ or entering $t$. Thus, all free components are singletons and we accordingly set $\roots=V$.

\paragraph{The $\SendFlow$ subroutine.}  Recall that the graph is symmetric, and the residual capacities are defined by taking both $f_e$ and $f_{\rev{e}}$ into account. For $e\in E(G)$ and $\alpha\ge 0$ and a flow $f$ in an instance $(G,u)$, the subroutine $\SendFlow(e,\alpha)$  changes these flows so that $f_e-f_{\rev{e}}$ increases by $\alpha$. Namely, it decreases $f_{\rev{e}}$ to $\alpha'=\min\{f_{\rev{e}},\alpha\}$, and increases $f_e$ to $f_e+\alpha-\alpha'$.

This and subsequent subroutines are shown in Algorithm~\ref{alg:subroutines}.
We note that in this and other subroutines, the instance $(G,u)$ is implicitly taken as the instance where the flow $f$ is defined, which should always be clear from the context.

\paragraph{The $\Saturate(\cdot)$ subroutine.}
The subroutine $\Saturate(e)$   takes an arc $e=(i,j)$  in the  instance $(G,u)$ with the flow $f$.
 It  saturates the arc $e$ in $f$ to attain $\res{f}_e=0$, by sending flow from $s$/$t$ to $t$/$s$ on  paths via $e$. More precisely, let $p=\Root(i)$ and $q=\Root(j)$. We let $\delta\defeq\res{f}_e$, and set $f_e=u_e$, $f_{\rev{e}}=0$. If $p\notin\{s,t\}$, then we  send $\delta'\defeq\min\{\delta,\res{f}_{sp}\}$  units of flow on $(s,p)$, and $\delta-\delta'$ units on $(t,p)$. If $p\neq i$, we send $\delta$ units on the shortcut arc $(p,i)$.
Similarly, if $q\neq j$, then we send $\delta$ units on the shortcut arc $(j,q)$, and if $q\notin\{s,t\}$, then we send $\delta''\defeq\min\{\delta,\res{f}_{qt}\}$ units of flow on $(q,t)$, and $\delta-\delta''$ units on $(q,s)$.

\paragraph{Creating the auxiliary instance.} The subroutine
$\CreateCompact(\Vess,\Eaux)$  returns an instance $(\cG,\cu)$ with $\cG=(\cV,\cE)$ and an arc set $H$.
We call $\dsTcCover(\Vess)$ from the transitive closure data structure, to obtain  a transitive cover  $(V',H)$ of $\Vess$, and let  
$\cV=\Vess\cup V(\Eaux)$.

We then obtain the arc set $\cE$ and capacities $\cu$ as follows. Each arc $e\in \Eaux$ is included with $\cu_e=\nr_e$.
For every $e\in H$, we include $e,\rev{e}\in \cE$ with $\cu_e=\infty$, $\cu_{\rev{e}}=0$.
Finally, for all nodes $i\in \cV\setminus \roots$, we include the shortcut arcs $(i,\Root(i))$ and $(\Root(i),i)$ with infinite capacity. Note that by the definition of $\Vess$, for every $e=(i,j)\in\Eaux$, we have $\Root(i),\Root(j)\in\Vess\subseteq\cV$.

\paragraph{Postprocessing.} The subroutine $\PostProcess(\delta,E')$ adjusts
the flow and safe capacities for $\delta>0$  and an arc set $E'$
as follows.  We 
 process the arcs in $E'$ one-by-one. If $e\in E'$ has $0<\res{f}_{\rev{e}}<2\delta$ and  $\res{f}_{{e}}>0$, then we call $\Saturate(\rev{e})$ to set $\res{f}_{\rev{e}}=0$. For such arcs, we set the safe capacity $\nr_e$ as the minimum of $\res{f}_e=u_e+u_{\rev{e}}$ and $3\delta$, and the safe capacity of the reverse arc as $\nr_{\rev{e}}=0$. For arcs not meeting this criterion, we set $\nr_e=\delta$.

Note that the order in which a pair of arcs $e$ and $\rev{e}$ are processed may lead to a different outcome. If initially $0<\res{f}_e,\res{f}_{\rev{e}}<2\delta$, then we set $\res{f}_{\rev{e}}=0$ if $e$ is processed first, and $\res{f}_{{e}}=0$ if $\rev{e}$ is processed first. However, note that the flow on a pair of arcs will be changed at most once, either for $e$ or for $\rev{e}$.

\paragraph{Removing a Root.} 
The subroutine $\RemoveRoot(k)$ (see Algorithm~\ref{alg:subroutines2}) is performed after the component of a root $k$ is merged into another component, and hence $k$ leaves $\roots$. Let $j=\Root(k)$ be its new root. For every node $i$ whose root has changed from $k$ to $j$, we add a new shortcut arc $(i,j)$ along with the reverse $(j,i)$, both with capacity $\infty$.  Finally, if $\res{f}_{sk}$ or $\res{f}_{kt}$ is positive but less than $2\varepsilon$, we saturate these arcs. Note that by Definition~\ref{def:proper}\eqref{it:oneway}, at most one of these flows could have been positive. We merge ${\cal T}_k$  into ${\cal T}_j$, connecting them by a pair of newly free arcs if $j\notin \{s,t\}$. If $j=s$, then we merge ${\cal T}_k$ to ${\cal T}_s$  by an abundant arc from ${\cal T}_s$ to 
${\cal T}_k$. Similarly, if  $j=t$, then we merge ${\cal T}_k$ to ${\cal T}_t$  by an abundant arc from ${\cal T}_k$ to 
${\cal T}_t$. 
\paragraph{Rerouting flow from shortcut arcs.}
After the end of the while cycle, $\RerouteS(F)$ is called to reroute the flow from the shortcut arc set $F=\Fs$ to the abundant trees ${\cal T}_j$. First, we create a vector $d(i)$ aggregating the total outgoing minus incomig flow on shortcut arcs at $i$. This is obtained by making a pass through $F$, and for each $e=(i,j)\in F$,  increasing $d(i)$ by $f_e-f_{\rev{e}}$, and decreasing $d(j)$ by the same amount.

We then consider the spanning subgraphs ${\cal T}_k$ one-by-one. In the underlying tree, we start from a leaf $i$.
Assume $d(i)>0$. If $i\neq s$, then there is an outgoing arc $e=(i,j)$ in ${\cal T}_k$.  We increase $f_e$ by $d(i)$, set $d(j)$ to $d(i)+d(j)$. If $i=s$, we send $d(i)$ units on $e=(i,s)$.
We update analogously if $d(i)<0$. We then remove $i$ and the incident arcs and iterate, until all nodes in $P_k=V({\cal T}_k)$ are processed.

\subsection{Rerouting using witness lists}\label{sec:witroute}

We now provide and analyze $\WitRoute(\cdot)$, a procedure that is used in the final flow rerouting step.  
 The algorithm is described 
in \Cref{alg:witRoute} and its main guarantees are provided later in \Cref{lem:routing} in \Cref{sec:witroute-analyze}. $\WitRoute$ takes as input a witness list, i.e., $L = (a_i,b_i,w_i)_{i \in [K]} \in \wlistSpace$ (see \Cref{def:wit_list}), assignments of real values to the transitive arcs of $L$, i.e., $f \in \R^{\dE_*(L)}$, and a subset of the arcs of $L$ which we call \emph{reverse arcs}, i.e., $R \subseteq \dEL(L)$. It is required that every reverse arc is \emph{valid} in the sense that each $(a_i,b_i) \in R$, is preceded by its reverse in the witness list (hence the term reverse arc). (However, every arc with this property need not be in $R$).

\begin{definition}[Valid Reverse Arcs]
For finite set $\dV$ and $L = (a_i,b_i,w_i)_{i \in [K]} \in \wlistSpace$, $R \subseteq \dEL(L)$ are \emph{valid reverse arcs (of $L$)} if for all $(a_i,b_i) \in R$ there exists $j_i \in [i - 1]$ with $(b_{j_i},a_{j_i}) = (b_i,a_i)$.
\end{definition}

In the context of \Cref{alg:max-flow}, $L$ will be a witness list formed by abundant arcs, and $R$ is the set $\rev{H}$, i.e., the reverses of extension arcs. As it is applied $\WitRoute(\cdot)$ then modifies the flow on abundant arcs only; note that the non-abundant arcs are not included in the \DITC{}.

More generally, $\WitRoute(\cdot)$  iterates over the transitive arcs of $L$ from $(a_L,b_L)$ to $(a_1,b_1)$ and modifies $f$ in two cases for the current iterate $(a_i,b_i,w_i)$. If  $(a_i,b_i) \in R$, then the $f_{(a_i,b_i)}$ is removed from $(a_i,b_i)$ and added to $(b_i,a_i)$. If this is not the case and $(a_i,b_i) \notin \dE(L)$, then $f_{(a_i,b_i)}$ is removed from $(a_i,b_i)$ and added to each edge in $\walk_L(a_i,b_i)$ instead (counting multiplicity). In other words, the procedure routes $f$ on reverse arcs backwards and reroutes $f$ on all other transitive arcs over a path on the arcs of the graph. Since the reverse arcs are valid, whenever $f$ is routed from reverse arc $(a_i,b_i)$ to $(b_i,a_i)$ if $(b_i,a_i)$ is transitive, it is later routed over $\walk_L(b_i,a_i)$. Consequently, the result of $\WitRoute(\cdot)$ only has non-zero values of $f$ on arcs of $ \dE(L)$ and can essentially be viewed as rerouting $f$ from reverse and transitive arcs onto $\dE(L)$.

Na\"\i{}vely, implementing $\WitRoute(\cdot)$ could be computationally expensive. A single $\walk_L(\cdot)$ can be computed $O(K)$ time and computing this for every transitive arc of $L$ can be done in $O(K^2)$ time. Instead, every iteration the implementation of $\WitRoute(\cdot)$ in \Cref{alg:witRoute} just removes $f_{(a_i,b_i)}$ from arc $(a_i,b_i)$ as needed and either adds it to $(b_i,a_i)$ or  $(a_i,w_i)$ and $(w_i,b_i)$ as appropriate. This procedure is easily implementable in $O(K)$ and we prove in \Cref{lem:routing} that it has the same effect as what is described in the preceding paragraph.

\begin{algorithm}[t!]
    \caption{Routing: $\WitRoute(L,f,R)$}
    \label{alg:witRoute}
    \SetCommentSty{algCommentFont}
    \SetKwProg{function}{function}{:}{}
    \SetKw{Return}{return}
    
    Finite set $\dV$ \tcp*{global variables}
    \BlankLine
    \function{$\WitRoute(L = (a_i,b_i,w_i)_{i \in [K]} \in \wlistSpace, f \in \R^{\dE_*(L)}, R \subseteq \dE(L))$}{
        $g \gets f$\;
        \For{$i = K$ to $1$\label{line:witroute:loop_start}}{
            \lIf{$(a_i,b_i) \in R$}{\label{line:witroute:1}
                $g_{(b_i,a_i)} \gets g_{(b_i,a_i)} - g_{(a_i,b_i)}$ and then 
                $g_{(a_i,b_i)} \gets 0$
            }
            \If{$(a_i,b_i) \notin \dE(L)$}{
                $g_{(a_i,w_i)} \gets g_{(a_i,w_i)} + g_{(a_i,b_i)}$ and  $g_{(w_i,b_i)} \gets g_{(w_i,b_i)} + g_{(a_i,b_i)}$ \label{line:witroute:2}\;
                $g_{(a_i,b_i)} \gets 0$ \label{line:witroute:3}\;
            }
        }
        \Return $g$\;
    }
\end{algorithm}

\section{Analyzing the Maximum Flow Framework}
\label{sec:analyzing}
Recall that for the  input instance $(\gGi,\ui)$, we have $n=|V(\gGi)|$ and $m=|E(\gGi)|$.  We assume  $n\ge 4$. We further let   $m_c$ denote the number of nodes and arcs with positive, finite capacity. denote the number of finite capacity pairs of arcs plus the number of nodes.
We introduce the following notation:
\begin{itemize}
\item Let $\Tsolv(\tilde m)$ denote the total running time of a sequence of calls to $\approxFlow$ on graphs with a total number of arcs at most $\tilde m$ and $M=\tilde m^{O(1)}$. We also assume $\Tsolv\left(O(\tilde m)\right)=\Omega(\tilde m)$.
 According to Theorem~\ref{thm:apx-flow-solver-run}, there is a solver with $\Tsolv(\tilde m)=O(\tilde m^{1+o(1)})$. %
 \item Let ${\rm TC}_{\mathrm{time}}(n,g,h)$ be runtime of the transitive cover data structure and let ${\rm TC}_{\mathrm{arcs}}(n,g,h)$ be the total number of arcs output by transitive cover queries,
 for $g$ additions and transitive cover queries on node sets of total size $h$, and final calls to 
 $\dsTcWit(\cdot)$ 
 and $\WitRoute(\cdot)$. %
\end{itemize}
We are ready to state the main theorem on the correctness and running time bounds on Algorithm~\ref{alg:max-flow}.

\begin{theorem}\label{thm:maxflow-analysis-main}
Algorithm~\ref{alg:max-flow} correctly terminates with a maximum flow in the input instance $(\gGi,\ui)$ on $n$  nodes, $m$ arcs, and a total of $m_c$ nodes and arcs with positive, finite capacity. The following bounds hold:
\begin{enumerate}[(i)]
\item\label{part:tot-new} The total number of extension arcs $\Fe$ added to the instance throughout is at most $2n$, and the total number of shortcut arcs $\Fs$ is at most $O(n\log n)$.
\item\label{part:tot-essential} The total number of essential arcs  throughout all iterations is $O(m_c)$.
\item \label{part:tot-data}  
Throughout Algorithm~\ref{alg:max-flow}, the total number of arcs added in $\dsTcAdd$ is $O(m)$.
 The total number of nodes queried in $\dsTcCover$ is $O(m_c)$.
\item \label{part:tot-running}
Algorithm~\ref{alg:max-flow} can be implemented in time
\[
{\rm TC}_{\rm time}(n,O(m),O(m_c)) + 
O\left(\Tsolv\left(O(m_c)\right)+{\rm TC}_{\rm arcs}(n,O(m),O(m_c))\right)+\tilde O(m)\,.
\]
\item\label{part:bit-complexity} Algorithm~\ref{alg:max-flow} is strongly polynomial. For rational input, all values of $f$, $\varepsilon$, and $r$ remain polynomially bounded in the input encoding length.
\end{enumerate}
\end{theorem}

\paragraph{Notation} 
Throughout the analysis, we use the following notation.
Let $T$ denote the total number of calls to the while cycle. 
For $\tau=1,2,\ldots, T$, we
let $(G^{(\tau)},u^{(\tau)})$ denote the instance $(G,u)$ at the beginning of the $\tau$-th call of the main while cycle with $E^{(\tau)}=E(G^{(\tau)})$, and let $(f^{(\tau)},\nr^{(\tau)},\varepsilon^{(\tau)},\roots^{(\tau)})$ denote the corresponding valid iterate. For $\tau=T+1$, let $(f^{(T+1)},\nr^{(T+1)},\varepsilon^{(T+1)},\roots^{(T+1)})$ be the iterate at the end of the final while cycle.
 We let 
$\Eaux^{(\tau)}$ denote the significant arc set, $\Vess^{(\tau)}$ the set of essential roots, $(\cG^{(\tau)},\cu^{(\tau)})$ the auxiliary instance.
In particular, $(G\bi{1},u\bi{1})=(\gGi,\ui)$.
Further, let  $H^{(\tau)}$ denote the set $H$ returned by $\dsTcCover(\Vess^{(\tau)})$ in iteration $\tau$. 

 We say that an arc is a boundary, gap, or essential arc, or that a node is essential in a certain iteration if it belongs to the corresponding set at the start of this iteration.

\subsection{Maintaining valid iterates}\label{sec:approx-compact}

We show that all $(f^{(\tau)},\nr^{(\tau)},\varepsilon^{(\tau)},\roots^{(\tau)})$ are valid iterates and moreover valid successors to each other. The key step is showing that this property is maintained in every iteration.

\begin{lemma}\label{lem:valid-maintain}
Assume $(f,\nr,\varepsilon,\roots)$ is a valid iterate in $(G,u)$ at the beginning of the while cycle of Algorithm~\ref{alg:max-flow}, and 
let $(f',\nr',\varepsilon',\roots')$ denote the corresponding quantities at the beginning of the next iteration in the instance $(G',u')$. Then, $(f',\nr',\varepsilon',\roots')$ is a valid iterate that is a valid successor of 
$(f,\nr,\varepsilon,\roots)$. Moreover, the following stronger version of Definition~\ref{def:successor}\eqref{it:arc-prox} holds:
\begin{enumerate}[(A')]
\setcounter{enumi}{1}
\item\label{it:arc-prox-mod} For every $e\in E(G)$,
$f'_e-f_e+f_{\rev{e}}\le \min\{\nr_e,\varepsilon\}+\ab^{-3}\varepsilon$, and 
for every $e\in E(G')\setminus E(G)$, $f'_e\le 2\varepsilon$.
\end{enumerate}
\end{lemma} 

The proof of Lemma~\ref{lem:valid-maintain} follows the statement and proof of our next lemma, which bounds the value of the auxiliary instance.

\begin{lemma}\label{lem:compact-flow-value} Given an instance $(G,u)$ with a valid iterate $(f,\nr,\varepsilon,\roots)$, let $\Vess$ be the set of essential nodes; let $\Eaux$ be the set of essential arcs and let $\delta_1= \max\left\{0,n^2\max\left\{\res{f}_e\,:\, \res{f}_e<\essentialthreshold\right\}\right\}$.
For the instance $(\cG,\cu)$ returned by $\CreateCompact(G,u,\nr,\Vess,\Eaux)$, 
$\val{\cG,\cu}\ge \val{G,\res{f}}-\delta_1$ holds.
If $s\notin\Vess$ or $t\notin \Vess$, then $\val{\cG,\cu}=0$.
\end{lemma}
\begin{proof}
Recall that by the definition of safe capacities, we have $\val{G,\res{f}}=\val{G,\nr}$.
First, let $(G',u')$ denote the graph obtained from $(G,\nr)$ by the following steps. First, delete all arcs $e\in E$ such that $u_e+u_{\rev{e}}< \essentialthreshold$. Next, for all arcs $e\in E$ such that $\rev{e}\in \gEA(f,\varepsilon)$ and $\res{f}_e<\essentialthreshold$, set $u'_e\defeq 0$. Further, for every abundant arc $e$, let us set $u'_e\defeq\infty$. Finally, for every other arc, we let $u'_e\defeq\nr_e$. 

We claim that $\val{G,\res{f}}\ge \val{G',u'}\ge \val{G,\res{f}}-\delta_1$.
The lower bound follows since $\delta_1$ is an upper bound on the total amount of capacity decrease. The upper bound holds since we only increased the capacity of abundant arcs, that could not have crossed any minimum cut in $G$ according to Proposition~\ref{prop:abundant-min-cut}. Note that at this point, every arc $e$ incident to a node $V\setminus V(\Eaux)$ will have capacity $u_e=\infty$ or $u_e=0$.

The auxiliary instance $(\cG,\cu)$ can be obtained from $(G',u')$ by restricting the node set from $V$ to $\bar V$ that contains $\Vess$, removing all  arcs with $\cu_e=\infty$, and adding a transitive cover $H$ for $\Vess$. Further, we keep the shortcut arcs $(i,\Root(i))$ and $(\Root(i),i)$ for every $i\in\bar V\setminus \roots$. The proof is complete by showing that $\val{\cG,\cu}=\val{G',u'}$.

First, we show $\val{\cG,\cu}\le\val{G',u'}$. Deleting nodes and arcs may only decrease the flow value. All new arcs added will be from the transitive closure of abundant arcs. By Proposition~\ref{prop:abundant-min-cut}, these new arcs cannot cross any minimum $s$--$t$-cut in $G'$.

Next, let us show $\val{\cG,\cu}\ge\val{G',u'}$. Let $\bar S\subseteq \bar V$ be a minimum $s$--$t$-cut in $G'$ such that $|\bar S|$ is as small as possible, and let $S$ be the set of nodes reachable from $\bar S$ in $V$ using $\infty$-capacity arcs in $G'$.

Let us first define  $S'\subseteq \bar S$ as the set reachable on $\infty$-capacity arcs from $\bar S\cap \Vess$. The set $S'$ clearly includes $V(\Eaux)\cap \bar S$. We claim that $S'=\bar S$. For a contradiction, assume that  $\bar S\setminus S'\neq\emptyset$.
Since $(\bar S\setminus S')\cap V(\Eaux)=\emptyset$, 
all arcs $e$ incident to $\bar S\setminus S'$ must have $\cu_e=\infty$ or $\cu_{e}=0$. As no arc with $\infty$ capacity leaves $S'$,  it follows that $\cu(\bar S,\bar V\setminus \bar S)=\cu(S',\bar V\setminus S')$, a contradiction to the minimal choice of $\bar S$.

Thus, $S'=\bar S$. 
Next, we claim that $t\notin S$, i.e., $S$ is an $s$--$t$ cut in $V$.
Using that $\bar S=S'$, every path from $\bar S$ to $t$ can be extended to a path from a node $i\in \bar S\cap \Vess$ to $t$. By the transitive cover property, $H$ must also contain an $i$--$t$ path of arcs with $\infty$ capacity, a contradiction to $\val{\cG,\cu}<\infty$. 

Consequently, $S$ is an $s$--$t$ cut in $(G',u')$. Using again that all arcs incident to $V\setminus V(\Eaux)$ have capacity $0$ or $\infty$, it follows that $u'(S,V\setminus \bar S)=\cu(\bar S,V\setminus \bar S)$, proving that 
$\val{\cG,\cu}\ge\val{G',u'}$.%

Finally, if $s\notin \Vess$, then $\val{\cG,\cu}=0$ follows because $\cu(\dtot{\cE}{P_s\cap \cV,\cV\setminus P_s}=0$. This follows because according to Definition~\ref{def:free-comps}, no abundant arc is leaving $P_s$, and therefore we deleted all outgoing arcs from $P_s$ when creating $(G',u')$. The analogous argument holds for $t\notin\Vess$.
\end{proof}

\begin{proof}[Proof of Lemma~\ref{lem:valid-maintain}]
First, note that the new instance $(G',u')$ is a valid extension of $(G,u)$. This follows by  Lemma~\ref{lem:add-abundant}, noting that all new arcs added either come from the transitive cover output $H$, or are shortcut arcs added inside free components. In particular, $\nu(G',u')=\nu(G,u)$.

Next, we  show that  
$(f',r',\varepsilon',\roots')$ is a valid iterate in $(G',u')$. If $s\notin\Vess$ or $t\notin \Vess$, then $\delta_2=0$ is set, and $f'$ is $\delta=\delta_1$-optimal by the last part of Lemma~\ref{lem:compact-flow-value}. Otherwise, let
$(y,e^\star)$ be the output of $\approxFlow{(\cG,\cu,\ab^{5})}$, $\delta=|E(\cG)|\cu^{y}_{e^\star}$, and let $\hat f$ be a flow obtained from $f$ after adding $y$ (lines~\ref{line:sendflow}--\ref{line:sendflow2}). First, we show that $\hat f$ is $\delta$-optimal for $\delta=\delta_1+\delta_2$.

By Lemma~\ref{lem:compact-flow-value}, $\val{\cG,\cu}\ge \val{G,\res{f}}-\delta_1$. By the definition of 
$\approxFlow{(\cG,\cu,\ab^{5})}$, the output $(y,e^\star)$ is such that $y$ is a $\delta_2$-optimal flow in $(\cG,\cu)$ for $\delta_2=|E(\cG)|u^y_{e^\star}$, that is, $\val{y}\ge \nu(\cG,\cu)-\delta_2$. Thus, after adding $y$ to $f$, the current flow $\hat f$ will  $\delta$-optimal for $\delta=\delta_1+\delta_2$.

 In $\PostProcess(\cdot)$, we consider the arc set 
$(\Eaux\setminus\gEbd)\uplus H$, and change the flow by at most $2\delta$ on each arc, which may decrease the flow value by at most $2\delta$ each time. This happens to at most one among each pair $e,\rev{e}$ of arcs, and the instance may contain at most two parallel set of pairs between any two nodes. Thus, the total change is at most  $2\delta n(n-1)$, showing that the flow after $\PostProcess(\cdot)$ is $(2n^2\delta)$-optimal. In $\RemoveRoot(i)$, we may only saturate $(s,i)$ or $(i,t)$; this  may only increase the flow value. Consequently, $f'$ is $\varepsilon'$-optimal for $\varepsilon'=2n^2\delta$, verifying Definition~\ref{def:proper}\eqref{it:eps-opt}.

Next we show Definition~\ref{def:proper}\eqref{it:safe}, namely that 
 $r'$ is a safe capacity vector.
Before $\PostProcess(\cdot)$,  we had a $\delta$-optimal flow $\hat f$, meaning that there is a maximum flow $f^\star$ with $f^\star_e\le \hat f_e-\hat f_{\rev{e}}+\delta$ for each $e\in E$ (Lemma~\ref{lem:res-value}). The safe capacity vector $r'$ is defined in $\PostProcess(\cdot)$ for all arcs in $E\cup H$; for the shortcut arcs added to $\Fs$ in $\RemoveRoot(\cdot)$, the safe capacity is set to $\infty$. In  $\PostProcess(\cdot)$, we change $\hat f$ to $f'$ such that $|\hat f_e-\hat f_{\rev{e}}-(f'_e-f'_{\rev{e}})|\le 2\delta$ for each arc $e$, and then set the safe capacity as $\nr'_e=\min\{3\delta,\res{f'}_e\}$. If $f'_e=\hat f_e$, then  the safe capacity is set as $\nr'_e=\min\{\delta,\res{f'}_e\}$. In both cases we see that 
 $f^*_e\le f'_e-f'_{\rev{e}}+\nr'_e$ on every arc.

Definition~\ref{def:proper}\eqref{it:bounds} is immediate by the execution of $\PostProcess(\cdot)$, noting that $\delta< 2\ab^{-2}\varepsilon$. In $\RemoveRoot(i)$, flow will only change on boundary and shortcut arcs, an after the execution of this subroutine, $\res{f}_e>0$ for $e=(s,i)$ or $e=(i,t)$ only if $\res{f}_e\ge 2\varepsilon$.

Definition~\ref{def:proper}\eqref{it:oneway} holds for $f$, that is, $\res{f}_{si}=0$ or $\res{f}_{it}=0$ for every $i\in V$. 
According to Lemma~\ref{lem:basic-flow}\eqref{it:nobackflow}, $y$ obtained in $\approxFlow$ does not send flow to $s$ or from $t$, and therefore adding $y$ to $f$ maintains this property.  It is easy to see that this property is also maintained throughout the subroutines $\PostProcess(\cdot)$ and $\RemoveRoot(\cdot)$.
Finally, Definition~\ref{def:proper}\eqref{it:roots} clearly holds for $\roots'$.

\medskip
It remains to show that $(f',r',\varepsilon',\roots')$ is a valid successor of $(f,r,\varepsilon,\roots)$, with the stronger property (\ref{it:arc-prox-mod}').
The bound $\varepsilon'\le \ab^{-3}\varepsilon$ follows since $\delta=\delta_1+\delta_2\le n^2\ab^{-5}\varepsilon+\ab^{-5}\varepsilon\le \ab^{-4}\varepsilon$, and
$\varepsilon'=2n^2\delta\le \ab^{-3}\varepsilon$.

To verify (\ref{it:arc-prox-mod}'),
first, note that that $y_e\le \min\{r_e,\varepsilon\}$ on every arc by the definition of $(\cG,\cu)$ and Lemma~\ref{lem:basic-flow}\eqref{it:basic-size}.
Subsequently, $\PostProcess(\cdot)$ saturates arcs of residual capacity less than $2\delta$. The operations in $\Saturate(\cdot)$ may thus send in total at most $2n^2\delta\le \ab^{-3}\varepsilon$ units of flow on shortcut and boundary arcs. For new arcs  $e\in H$ added in $E(G')\setminus E(G)$, we have $f'_e\le y_e\le \varepsilon$, noting also that $\PostProcess(\cdot)$ may only decrease the flow on such arcs. Finally, for the shortcut arcs added in $\RemoveRoot(\cdot)$ may only carry at most $2\varepsilon$ units of flow.
The above observartions together imply property (\ref{it:arc-prox-mod}').
 \end{proof}

We can now easily derive the following:
\begin{lemma}\label{lem:valid-main}
For $t=1,2,\ldots,T+1$, $(f\bi{\tau},r\bi{\tau},\varepsilon\bi{\tau},\roots\bi{\tau})$ is a valid iterate in the instance $(G\bi{\tau},u\bi{\tau})$, and for each $\tau<\tau'\le T+1$, $(G\bi{\tau'},u\bi{\tau'})$ is an admissible extension of $(G\bi{\tau},u\bi{\tau})$, and $(f\bi{\tau'},r\bi{\tau'},\varepsilon\bi{\tau'},\roots\bi{\tau'})$ is a valid successor of $(f\bi{\tau},r\bi{\tau},\varepsilon\bi{\tau},\roots\bi{\tau})$.  
\end{lemma}
\begin{proof}
Lemma~\ref{lem:valid-maintain} already shows that all these are valid iterates, and the valid successor property for consecutive iterates $\tau'=\tau+1$, with the stronger property (\ref{it:arc-prox-mod}'). The only remaining property to show is Definition~\ref{def:successor}\eqref{it:arc-prox} for $\tau'>\tau+1$.
This follows by using (\ref{it:arc-prox-mod}') for all consecutive iterates between $\tau$ and $\tau'$, and that $\varepsilon$ decreases at least by a factor $\ab^{-3}$ between any two consecutive iterations.
\end{proof}

The next lemma states that input arcs with infinite capacity not incident to $s$ or $t$ remain at the other capacity boud throughout the while cycle of the algorithm. We note that their flow may be changed in the very final step, when $\RerouteS{}$ and
 $\WitRoute$ reroute the flow to original arcs from the shortcut arcs and the transitive closure arcs, respectively.
\begin{lemma}\label{lem:nailed-arcs}
For every arc $e=(i,j)\in \gEi$ with $u_e=\infty$, $\{i,j\}\cap\{s,t\}=\emptyset$, we have $\res{f\bi{\tau}}_{\rev{e}}=0$ for all $\tau=1,2,\ldots,T$. 
\end{lemma}
\begin{proof}
This property holds in $\Initialize$. In all subsequent iterations, note that such an arc $e$ will never be essential, since $u_e+u_{\rev{e}}=\infty$. The auxiliary graphs $(\cG,\cu)$ will only contain essential arcs and arcs from the set $H$. All arcs in $H$ that carry flow will be added to  $\Fe$ as new extension arcs.\footnote{This may create parallel copies if the arc was already present.} 
Since $e$ is not a boundary arc nor a shortcut arc, the subroutines $\PostProcess(\cdot)$ or $\RemoveRoot(\cdot)$ will not use it to route flow in calls to $\Saturate(\cdot)$. Thus, $\res{f}_{\rev{e}}=0$ is maintained through each while cycle.
\end{proof}

\begin{lemma}\label{lem:shortcut-flow-bound}
Let $e$ be a shortcut arc added in $\RemoveRoot(\cdot)$ at the end of iteration $\tau$. Then, $f^{(\tau')}_e\le 4\varepsilon^{(\tau+1)}$ in all iterations $\tau'>\tau$.
\end{lemma}
\begin{proof}
At the point where $e$ is added as a shortcut arc, we already set $\varepsilon=\varepsilon^{(\tau+1)}$ as the new value. $\RemoveRoot(i)$ may send at most $2\varepsilon^{(\tau+1)}$ units of flow on the newly added shortcut arcs $(i,\Root(i))$ or $(\Root(i),i)$. Thus, $f^{(\tau+1)}_e\le 2\varepsilon^{(t+1)}$. This  may subsequently change by at most $2\varepsilon^{(\tau+1)}$ according to Lemma~\ref{lem:valid-main} and Definition~\ref{def:successor}\eqref{it:arc-prox}.
\end{proof}

\subsection{Bounding the number extension arcs}

We are ready to  show Theorem~\ref{thm:maxflow-analysis-main}\eqref{part:tot-new}.
\begin{lemma}\label{lem:extension-tot} 
The total number of extension arcs $\Fe$ added to the instance throughout is $2n$, and the total number of shortcut arcs $\Fs$ is at most $n\log_2 n$.
\end{lemma}
\begin{proof}
Since the flow $y$  returned by $\approxFlow $is a basic acylic flow, 
by Lemma~\ref{lem:basic-flow}\eqref{it:acyclic}, in every iteration the set of arcs  $e\in H$ with $y_e>0$ is a forest. This is a superset of the new arcs added to $\Fe$ in this iteration. Moreover, all new arcs added will be gap arcs or free arcs in the next iteration, according to Lemma~\ref{lem:valid-main} and Lemma~\ref{lem:strict-to-free}. Consequently, if $\gamma^{(\tau)}$ new arcs are added to $\Fe$ in iteration $\tau$, then the number of components decreases by at least $\gamma^{(\tau)}$ between iterations  $\tau$ and $\tau+2$.  Hence, we can bound $|\Fe|= \sum_{\tau=1}^T \gamma^{(\tau)}\le 2n$.

The bound on the number of shortcut arcs follows because every time a non-root node $i$ is connected to a new root $j$ by a new shortcut arc $(i,j)$, we have that either $j\in \{s,t\}$, in which case $\Root(i)$ will not change anymore, or $|P_{\Root(i)}|$ has doubled by the choice that we always keep the root of the larger component when merging components. Hence, at most $\log_2 n$ shortcut arcs will be created for every $i\in V$.
\end{proof}

\subsection{Correctness and running time of rerouting}\label{sec:witroute-analyze}

In \Cref{sec:witroute}, we described the subroutine $\WitRoute(\cdot)$ that is used at the end of the algorithm for the final flow rerouting.  
\Cref{alg:witRoute} describes an efficient rerouting procedure given a witness list. We now show  the correctness of this subroutine and bound its running time.
\begin{lemma}
\label{lem:routing}
For arbitrary witness list $L = (a_i,b_i,w_i)_{i \in [K]} \in \wlistSpace$, flow $f \in \R^{\dE_*(L)}$ and valid reverse arcs $R \subseteq \dEL(L)$, let $g^{(K)} = f$ and iteratively define for all $i = K$ to $1$:
\begin{equation}
\label{eq:route-cases}
g^{(i - 1)} \gets 
\begin{cases}
    g^{(i)} - g_{(a_i,b_i)}^{(i)} (\vec{1}_{(a_i,b_i)} - \vec{1}_{(b_i,a_i)})
        & \text{ if $(a_i,b_i) \in R$} \\
    g^{(i)} 
        & \text{ if $(a_i,b_i) \in \dE(L) \setminus R$}\\
    g^{(i)} - g_{(a_i,b_i)}^{(i)} (\vec{1}_{(a_i,b_i)} - \vec{1}_{\walk_L(a_i,b_i)})
        & \text{ if $(a_i,b_i) \in \dE_*(L) \setminus \dE(L)$}
\end{cases}
\end{equation}
where $\vec{1}_{W}$ denotes the vector where each $[\vec{1}_{W}]_e$ is the number of times arc $e$ is traversed in walk $W$. Then $g^{(0)} =  \WitRoute(L,f,R)$ (\Cref{alg:witRoute}) and can be computed in $O(K)$ time.\footnote{This assumes standard graphical access to $\dE_*(L)$ in that given an $(a_i,b_i,w_i)$ and  $x \in \R^{\dE_*(L)}$ it is possible to lookup and change the value associated $x_e$ for any $e \in \{(a_i,b_i),(a_i, w_i),(w_i,b_i)\} \in \dE_*(L)$ in $O(1)$ time. This assumption is met in our applications of \Cref{lem:routing}.}
\end{lemma}

\newcommand{\alg}{\mathrm{alg}}

\begin{proof}
Since $K = |\dE_*(L)|$ it is straightforward to implement \Cref{alg:witRoute} in $O(K)$ time by spending $O(K)$ time to mark which $(a_i,b_i)$ are in $R$ and $\dE(L)$ and then directly implementing the operations.  Therefore, in the remainder of this proof we let $h^{(0)} \defeq \WitRoute(L,f,R)$ and prove that $h^{(0)} = g^{(0)}$ by induction on $K$. When $K = 1$, $(a_1,b_1) \in \dE(L) \setminus R$ (since both $(a,b) \in \dELcl(L) \setminus \dE$ and $(a,b) \in R$ require $\dELsub{1-1}(L) \neq \emptyset$) and therefore, $\WitRoute(L,f,R)$ outputs the desired $h^{(0)} = f = g^{(0)}$.

Next, suppose that the claim holds for $K = K_*$. We prove that it holds for $K = K_{*} + 1$. Our proof strategy is to reason about the the first iteration of the for loop (\Cref{line:witroute:loop_start}) in $\WitRoute(\cdot)$ directly and then use the inductive hypothesis to reason about remaining iterations. Let $h^{(K_*)}$ be $g$ at the start of the $\WitRoute(\cdot)$ for loop (\Cref{line:witroute:loop_start}) for $i = K_*$ (i.e., $g$ after one iteration of the for loop). Let $\hat{L} \defeq (a_i,b_i,w_i)_{i \in [K_*]}$ denote the first $K_*$ elements of the witness list, let $\hat{f} \defeq h^{(K_*)} |_{\dE_*(\hat{L})}$ denote $h^{(K_*)}$ ignoring the $\bar{e} \defeq 
(a_{t_* + 1}, b_{t_* + 1})$ coordinate, let $\hat{R} \defeq R \setminus \{\bar{e}\}$, and let $\hat{g} \defeq \WitRoute(\hat{L}, \hat{f}, \hat{R})$. Since the $(a_i,b_i)$ in a witness list are distinct, the simple iterative structure of $\WitRoute(\cdot)$ implies that $h^{(0)}_{\bar{e}} = h^{(K_*)}_{\bar{e}}$ and $h^{(0)}_{\dE_*(\hat{L})} = \hat{g}$. Consequently it suffices to reason about $h^{(K_*)}_{\bar{e}}$ and $\hat{g}$.

Moreover, since $\hat{L}$ is a witness list of length $K_*$, the inductive hypothesis applies to $\hat{g}$. For all $i \in [K_*] \cup \{0\}$ let $\hat{g}^{(i)}$ denote $g^{(i)}$ as defined in \eqref{eq:route-cases} of  \Cref{lem:routing} applied to  $\WitRoute(\hat{L}, \hat{f}, \hat{R})$ and let $h^{(i)} \in \R^{\dE_*(L)}$ have $h^{(i)}_{\bar{e}} = h^{(K_*)}_{\bar{e}}$ and $h^{(i)}|_{\dE_*(\hat{L)}} = \hat{g}^{(i)}$ (this is consistent with the earlier definition of $h^{(K_*)}$ and $h^{(0)}$). The inductive hypothesis (and  $g_{\bar{e}} = h^{(K_*)}_{\bar{e}}$ and $g_{\dE_*(\hat{L})} = \hat{g}$) imply that $g = h^{(0)}$.

We now characterize $h^{(k)} - g^{(k)}$ for all $k \in (0,K_*]$ by considering two cases. First, when $\bar{e} \in \dE(L)$ then $h^{(K_*)} = g^{(K^*)}$ as \eqref{eq:route-cases} performs same changes to both $h^{(k)}$ and $g^{(k)}$. Consequently, in this case each $h^{(i)} = g^{(i)}$ for all $i \in [0, K_*]$ and $h^{(0)} = g^{(0)}$ as desired. 

Second, consider the case when $\bar{e} \notin \dE(L)$. In this case
\begin{align*}
    h^{(K_*)} &=  f - f_{\bar{e}} (\vec{1}_{\bar{e}} - \vec{1}_{(a_{K_* + 1},w_{K_* + 1})} - \vec{1}_{(w_{K_* + 1},b_{K_* + 1})} )\text{  whereas }\\
    g^{(K_*)} &=  f - f_{\bar{e}} (\vec{1}_{\bar{e}} - \vec{1}_{\walk_{\hat{L}}(a_{K_* + 1},b_{K_* + 1})} )\,.
    \end{align*} 
    Let $i < j \in [K_*]$ satisfy $\{(a_i,b_i),(a_j,b_j)\} = \{(a_{K_* + 1},w_{K_* + 1}) , (w_{K_* + 1}, b_{K_* + 1})\}$ and define $e_i \defeq (a_i,b_i)$, $e_j \defeq (a_j,b_j)$, $W_i \defeq \walk_{\hat{L}}(a_i,b_i)$, and $W_j \defeq \walk_{\hat{L}}(a_j,b_j)$. Note that 
    \[
    \vec{1}_{W_{K_*}} 
        = \vec{1}_{\walk_{\hat{L}}(a_{K_* + 1},w_{K_*})} + \vec{1}_{\walk_{\hat{L}}(w_{K_*}, b_{K_* + 1})} 
        = \vec{1}_{W_i} + \vec{1}_{W_j}
    \]
    since $\walk_{\hat{L}}(a_{K_* + 1},b_{K_* + 1}) = \walk_{\hat{L}}(a_{K_* + 1}), \walk_{\hat{L}}(w_{K_*}, b_{K_* + 1})$. For all $k \in (i, K_*]$, $(a_k,b_k)$ is not in $W_i$ or $W_j$, an thereofre iteration of \Cref{lem:routing} performs the same change to $g^{(k)}$ and
\[
h^{(k)} - g^{(k)} = f_{\bar{e}}(\vec{1}_{e_i} - \vec{1}_{W_i} + \vec{1}_{a_j} - \vec{1}_{W_j})
\text{ for all } k \in (i, K_*]\,.
\]
Next if $(a_i,b_i) \in \dE$ then $W_i = e_i$ and, consequently, $h^{(j)} - g^{(j)} = \vec{1}_{(a_j,b_j)} - \vec{1}_{\walk_{\hat{L}}(a_j,b_j)}$. However, if $(a_i,b_i) \in \dE_*$ then we see that iteration $i$ of \Cref{lem:routing} performs the same change to $g^{(k)}$ and $h^{(k)}$ other than also subtracting $f_{\bar{e}} \vec{1}_{e_i}$ and adding $f_{\bar{e}} \vec{1}_{W_i}$ to $h^{(k)}$. In both cases, $h^{(i-1)} - g^{(i-1)} = f_{\bar{e}}(\vec{1}_{(a_j,b_j)} - \vec{1}_{W_j})$. Additionally, by similar reasoning, since all $k \in (i, K_*]$, $(a_k,b_k)$ is not in $\walk_{\hat{L}}(a_i,b_i)$ or $\walk_{\hat{L}}(a_j,w_j)$ and therefore the step in \Cref{lem:routing} performs the same operation to $g^{(k)}$ and $h^{(k)}$ we see that 
$
h^{(k)} - g^{(k)} = f_{\bar{e}}(\vec{1}_{a_j} - \vec{1}_{W_j})
$ for all $k \in (j, i)$. Repeating a similar argument to when $k = i$ shows that $h^{(k)} - g^{(k)} = \vec{0}$ for all $k \in [0, j)$ completing the proof by induction.
\end{proof}

\subsection{Correctness}
We start by showing that the output of Algorithm~\ref{alg:max-flow} is correct: if it terminates, it outputs a maximum flow in the original instance.
\begin{lemma}\label{lem:correct-main}
At termination, Algorithm~\ref{alg:max-flow} outputs a maximum flow in the input instance $(\gGi,\ui)$.
\end{lemma}
\begin{proof}
Let $(f,\nr,\varepsilon,\roots)=(f^{(T+1)},\nr^{(T+1)},\varepsilon^{(T+1)},\roots^{(T+1)})$ be the iterate at the end of the last while cycle; let $(G,u)$ with $G=(V,E)$ denote this final instance. Thus, $\varepsilon=0$, and because this is a valid iterate (Lemma~\ref{lem:valid-main}), it follows that $f$ is a maximum flow in the current instance $(G,u)$. This instance is an admissible extension of the original instance: $(G,u)=(\gGi,\ui)\oplus(\Fe,\Fs)$.
We then call $\RerouteS(\Fs)$ to reroute the flow from shortcut arcs in $\Fs\cup \rev{\Fs}$ to arcs in the abundant trees ${\cal T}_k$, as well as $\dsTcWit{}$ and $\WitRoute{}$ to reroute the flow from the extension arcs in $\Fe$ to original arcs using the transitive cover data structure; \Cref{lem:routing} shows its correctness. We need to show that this results in a feasible flow $f^\star$ in the input instance.

Both subroutines may only change the flow on arcs that became abundant at some iteration of the algorithm. Let $\bar e\in E$ be an original arc that became abundant at the end of iteration $\tau$. Thus, $\res{f^{(\tau+1)}}_{\bar e}\ge \ab \varepsilon^{(t+1)}$, and by Lemma~\ref{lem:valid-main}, $\res{f}_{\bar e}\ge (\ab-2) \varepsilon^{(t+1)}$.
We  show that  $\RerouteS(\cdot)$  may  increase $f_{\bar e}$ by at most $4n\varepsilon^{(t+1)}\log_2 n$ and   $\WitRoute(\cdot)$ by at most $6n\varepsilon^{(t+1)}$. This implies that $f^\star_{\bar e}\le u_e$ because $\ab=4n^2>4n\log_2 n +6n+2$; recall the  assumption $n\ge 4$.

Let us first start with the subroutine $\RerouteS(\Fs)$. If $\bar e$ is not contained in any of the final spanning subgraphs ${\cal T}_k$, $f_{\bar e}$ will not be changed. Assume that  for a final root $k\in\roots$,  $\bar e$ is an arc in ${\cal T}_k$. %
The arc $\bar e$ will be on the fundamental cycle of shortcut arcs that 
 were created in iteration $\tau$ or later. These carry  at most  $4n\varepsilon^{(t+1)}\log_2 n$ amount of flow in total, according to Lemma~\ref{lem:shortcut-flow-bound}. %

Let $U_1$ and $U_2$ be the node sets of the two components of ${\cal T}_k$ connected by $\bar e$. Thus, $\sum_{i\in U_1} d(i)=-\sum_{i\in U_2} d(i)$, and both these values are bounded as $4n\varepsilon^{(t+1)}\log_2 n$. Hence, $\RerouteS(\Fs)$ will change the flow on $\bar e$ by at most $4n\varepsilon^{(t+1)}\log_2 n$ as claimed.

\medskip

We now turn to  bounding the flow changes in $\WitRoute(L,f|_{\gEA},\rev{H})$.
In each iteration $k$, after the new value of $\varepsilon=\varepsilon^{(k+1)}$ is set, we identify all newly abundant arcs with capacity at least $\ab\varepsilon^{(k+1)}$, and call $\dsTcAdd$ for them. If an arc $e$ is added to $\Fe$ from the transitive cover, then at the end of the same or the next iteration, $\rev{e}$ will become abundant and added by $\dsTcAdd(\cdot)$ (Lemma~\ref{lem:strict-to-free}). Recall that the shortcut arcs are not added by $\dsTcAdd(\cdot)$ and hence $\WitRoute(\cdot)$ will not use them to route flow.

For an arc $e\in E\setminus \gEi$, let $\gamma(e)$ denote the iteration when $e$ was added to $E$. If $\gamma(e')<\gamma(e)$, then $e'$ precedes $e$ on the witness list $L$ returned by $\dsTcWit(\cdot)$.
For original input arcs $e\in\gEi$, let $\gamma(e)=0$. The witness list $L$ contains the arcs in $\Fe$ in increasing order of $\gamma(e)$, and 
$\WitRoute(\cdot)$ processes the arcs in the reverse order. For arcs with $\gamma(e)=k$, the total flow amount is at most $f_e\le 3\varepsilon^{(k+1)}$ (Lemma~\ref{lem:valid-main}  and Definition~\ref{def:successor}\eqref{it:arc-prox}). 

The flow from an arc $e=(i,j)$ is rerouted to a simple $i-j$ path $p$ of arcs $e'$ with $\gamma(e')<\gamma(e)$, and such that $e'$ became abundant before iteration $\gamma(e)$ (see Lemma~\ref{lem:routing})%
$\WitRoute(\cdot)$ roots the flow from every arc in $\gEA$ on a simple path, by the requirement that witness list $L$ is path-structured (\Cref{def:general_transitive_cover}). Therefore,
 if we track the flow from an arc $e$ with $\gamma(e)=k$ through the sequence of reroutings, it may change the flow by at most  $n^{k-h-1} f_e\le 3n^{k-h-1} \varepsilon^{(k)}$ for an arc $e'$ that became abundant in iteration $h\le k-1$.
Recall that for $k\ge h+1$, $\varepsilon^{(k)}\le \ab^{-3(k-h-1)}\varepsilon^{(h+1)}$ (by Definition~\ref{def:successor}\eqref{it:eps-decrease}), and that $|\Fe|\le 2n$ (Lemma~\ref{lem:extension-tot}).

Consider now the original arc $\bar e\in E$ that became abundant at the end of iteration $\tau$; we have $\gamma(e)=0$ as it is an original arc. Let $\rho_k$ denote the number of extension arcs added in iteration $k$; thus, $\sum_{k=1}^{T}\rho_k\le 2n$. By the above, the total flow change on $\bar e$ can be bounded as
\[
\sum_{k=\tau+1}^{T+1} \rho_k 3n^{k-\tau-1} \varepsilon^{(k)}\le 3\varepsilon^{(\tau+1)} \cdot \sum_{k=\tau+1}^{T+1} \rho_k n^{k-\tau-1} \ab^{-3(k-\tau-1)}\le  6n \varepsilon^{(\tau+1)}\, .
\]
using that $\ab=4n^2$.  This completes the proof. %
\end{proof}

\subsection{Bounding the number of essential arcs}

We next  show Theorem~\ref{thm:maxflow-analysis-main}\eqref{part:tot-essential}.

\begin{lemma}\label{lemma:meta-essential-tot} 
 The total number of essential arcs  throughout all iterations is $O(m_c)$.
\end{lemma}{}
\begin{proof}
Lemma~\ref{lem:extension-tot}, the total number of extension arcs is $O(n)$. If an arc $e\in H$ is added to the instance a certain iteration, then it will be a gap arc in the next iteration and by Lemma~\ref{lem:strict-to-free}, it will be free in all subsequent ones. Hence, every extension arc can be essential in at most one iteration. Shortcut arcs will never be essential arcs as they have $\infty$ capacity on both sides.

It remains to bound the number of essential arcs among the original arcs in $\gEi$. Let us call an  arc $e$ \emph{regular} at iteration $\tau$ if 
\[
\ab^{-6}\varepsilon\bi{\tau}\le u_e+u_{\rev{e}}<\ab\varepsilon\bi{\tau}\, .
\]
Since $\varepsilon\bi{\tau+1}\le \ab^{-3}\varepsilon\bi{\tau}$ (by Lemma~\ref{lem:valid-main}), every arc can be regular arc for at most three subsequent iterations. 
Thus, the total number of regular arcs is $O(m_c)$ throughout the algorithm. 

By Lemma~\ref{lem:improved-opt-1}, every essential arc is either a gap arc, a boundary arc, or a regular arc. By
Lemma~\ref{lem:valid-main} and Lemma~\ref{lem:strict-to-free}, every arc in $\gEi$ may be a gap arc in at most one iteration. Moreover, Lemma~\ref{lem:nailed-arcs} shows that if $u_e+u_{\rev{e}}=\infty$, then $e$ can only be gap arcs if it is  incident to $s$ or $t$.
 Hence, the total number of gap arcs in $\gEi$ throughout is $O(n+m_c)$.

It remains to bound the number of essential arcs that are boundary arcs. Let $e=(s,i)$ be an essential arc in iteration $\tau$ for some $i\in\roots$; the case $e=(i,t)$ follows analogously. Let $f\defeq f\bi{\tau}$, $\varepsilon\defeq\varepsilon\bi{\tau}$,  $f'\defeq f\bi{\tau+3}$, and $\varepsilon'\defeq\varepsilon\bi{\tau+3}$. Let $P_i\in \CP_{f,\varepsilon}$ be the component with root $i$ in iteration $\tau$. First, assume $P_i\notin \CP_{f',\varepsilon'}$, i.e., $P_i$ is merged with some other component by iteration $\tau+3$. The total number of all such components $P_i$ can be bounded as $O(n)$.

 For the rest, let us assume $P_i\in \CP_{f',\varepsilon'}$, i.e., the component has not been merged. Moreover, assume that $e$ is the only essential arc incident to $P_i$ in iteration $\tau$, and all arcs incident to $P_i$ are original arcs in $\gEi$; otherwise, we could charge $e$ to an extension arc or to a non-boundary essential arc, bounding the total number of such arcs as $O(n+m_c)$.

 Let $\alpha\defeq \res{f}_{e}\ge \essentialthreshold\ge\ab^4\varepsilon'$ since $\varepsilon'\le\ab^{-9}\varepsilon$. We must have $\res{f'}_{e}<\ab\varepsilon'\le\ab^{-3}\alpha<\alpha/2$, as otherwise $e$ would be abundant in iteration $\tau+3$, and thus $P_i$ would have been merged into $P_s$.
 By flow conservation, there must be an arc $e'\in\dout{E\bi{\tau+2}}{P_i}$ with $f'_{e'}-f'_{\rev{e'}}>f_{e'}-f_{\rev{e'}}+\alpha/(2n^2)>f_{e'}-f_{\rev{e'}}+\ab^{-1}\alpha$. Consequently, $\res{f'}_{\rev{e'}}>\ab^{-1}\alpha\ge\ab^3\varepsilon'$, and therefore $\rev{e'}\in\gEA(f',\varepsilon')$. 
By the assumption that $P_i$ is still a free component for $(f',\varepsilon')$, it follows that $e'\notin\gEA(f',\varepsilon')$.

We have $u_{e'}+u_{\rev{e'}}>\ab^{-1}\alpha\ge \ab^{-6}\varepsilon$.
The proof is complete by showing that $u_{e'}+u_{\rev{e'}}< \ab\varepsilon$. Thus, $e'$ is a regular arc incident to $P_i$, but we already assumed there are no such arcs.

For a contradiction, assume $u_{e'}+u_{\rev{e'}}\ge \ab\varepsilon$. 
By Definition~\ref{def:proper}\eqref{it:bounds} for $(f,\varepsilon)$, $e'$ must be a boundary arc, or  
$e'\in\gEStr(f,r,\varepsilon)$, or $\res{f}_{e'}=0$, or $\res{f}_{\rev{e'}}=0$. We already assumed there are no gap arcs incident ot $P_i$.
Note that  $e'$ cannot be a boundary arc, since the only outgoing boundary arcs are $(i,s)$ and $(i,t)$.
 Clearly, $e'\neq (i,s)=\rev{e}$, and we had $f_{it}=u_{it}$, $f_{ti}=0$ due to $\res{f}_{ti}=0$ asserted in Definition~\ref{def:proper}\eqref{it:oneway}.
We have $\res{f}_{{e'}}>0$ by the choice of $e'$.  Thus, the remaining case is $\res{f}_{\rev{e'}}=0$. 
If $u_e+u_{\rev{e}}\ge\ab\varepsilon$,
then ${e'}\in\gEA(f,\varepsilon)$, and consequently, ${e'}\in \gEA(f',\varepsilon')$ by Lemma~\ref{lem:strict-to-free}\eqref{it:free-free}. Together 
with $\rev{e'}\in \gEA(f',\varepsilon')$ shown above, we get ${e'}\in \gEF(f',\varepsilon')$, a contradiction again.
\end{proof}

\subsection{Bounding the calls to transitive cover}\label{sec:total-cover-call}
We now prove part \eqref{part:tot-data} of Theorem~\ref{thm:maxflow-analysis-main}.
\begin{lemma}\label{lem:tot-cover-call}
Throughout Algorithm~\ref{alg:max-flow}, the total number of arcs added in $\dsTcAdd$ is $O(m)$. %
\end{lemma}
\begin{proof}
Every original arc may be added in $\dsTcAdd$ at most once. Further, the reverse arcs of extension arcs will also be added once; by Lemma~\ref{lem:extension-tot}, their total number is $O(n)$. Thus, the total number of arcs added in  $\dsTcAdd$ is $O(m)$. When calling $\dsTcCover(\Vess^{(\tau)})$, we have $|\Vess^{(\tau)}|\le 4|\Eaux^{(\tau)}|$, since every $e\in \Eaux^{(\tau)}$ may contribute two endpoints and two corresponding roots. Thus, the total number of nodes queried in $\dsTcCover$ is at most 4 times the total number of essential arcs, that is $O(m_c)$ by Lemma~\ref{lemma:meta-essential-tot}. %
\end{proof}

\subsection{Maintaining the basic data structures}
\label{sec:update-abundant}
We now describe how essential arcs, components, and other quantities can be efficiently maintained throughout the algorithm, and bound the running time of these operations.

\paragraph{Arc types and residual capacities}
 We initialize the abundant arc set $\gEA$ as the set of arcs with $u_e=\infty$; this includes all  arcs $(i,s)$ and $(t,i)$. Initially, $\gEF=\emptyset$.
 Arcs with $u_{\rev{e}}=\infty$ that are not incident to $s$ or $t$ are set to $\res{f}_e=0$ at initialization and this property is maintained throughout (see Lemma~\ref{lem:nailed-arcs}). Hence, such an $e$ will never become abundant, free, or essential. Shortcut arcs $\Fs$ will be by definition free throughout.

We now consider the set of all other arcs $E'$. These include the $m_c$ arcs with $u_e+u_{\rev{e}}<\infty$, the $2n$  arcs of the form $(s,i)$ and $(i,t)$, and the $O(n)$ extension arcs $F_s$. Thus, $|E'|=O(m_c)$.

 For maintaining residual capacities, as well as the status of abundant, free, and essential arcs on $E'$, we maintain the $\res{f}_e$ values in a binary tree. In each iteration, we identify arcs with key values between $\ab^{-5}\varepsilon$ and  $\ab\varepsilon$ as essential arcs. Their endpoints and the corresponding roots will be the essential nodes. The algorithm will also need the value $\delta_1$ that is $n^2$ times the largest capacity of the largest arc below the threshold $\ab^{-5}\nu(\cG,\cu)$; this amounts to one additional search operation per iteration.

 Flow values may only change on essential arcs, as well as on new extension arcs added in the current iteration. We update the key values of essential arcs at the end of the iteration, and also add to the data structure the new extension arcs to the binary tree.  We identify the abundant arcs as those with key values at least $\ab\varepsilon$.

 The amortized cost of these search and insert operations on $E'$ is $O(\log m)$, that is $O(\log n)$. Since the total number of essential arcs is $O(m_c)$ (Lemma~\ref{lem:extension-tot}), the total running time of these operations can be bounded as $O(m_c \log n)$.

\paragraph{Free components, roots, and spanning subgraphs}
The components in $\CP_{f,\varepsilon}$ are initialized as the set of singletons, along with the root set $\roots=V$. We maintain the components in a heap structure, with the roots corresponding to $\roots$. The pointers $\Root(i)$ point to the root of the heap containing $i$.
The boundary arc set $\gEbd$ will be simply the arcs between $\roots$ and $\{s,t\}$.

Whenever an arc $e=(i,j)$ is added to $\gEF$ and  $\Root(i)=k\neq k'=\Root(j)$, then we merge the components $P_k$ and $P_{k'}$. The root of the new component will become $s$ or $t$ if either of them is in $\{k,k'\}$; otherwise, the root of the larger component (breaking ties arbitrarily). The other root is then removed from $\roots$.  
Similarly, if a new arc $e=(i,j)$ is added to $\gEA$ with $\Root(i)=s$, $\Root(j)=k'\neq s$, we merge $P_{k'}$ into $P_s$.  If a new arc $e=(i,j)$ is added to $\gEA$ with $\Root(i)=k'\neq t$, $\Root(j)=t$, we merge $P_{k'}$ into $P_t$.
The spanning subgraphs ${\cal T}_k$ can be easily maintained during these merging operations: the spanning subgraphs of the different components will be connected by the free or abundant arcs identified when joining these components.
There are at most $n-2$ merging operations; maintaining the free components, roots, and spanning subgraphs can be done in $O(n\log n)$.

\subsection{Running time estimate}\label{sec:meta-runtime}
We next show part \eqref{part:tot-running} in Theorem~\ref{thm:maxflow-analysis-main}. Recall the notation ${\rm TC}_{\mathrm{time}}$ and 
${\rm TC}_{\mathrm{arcs}}$ and corresponding assumptions from the beginning of Section~\ref{sec:analyzing}.
\begin{lemma}\label{lemma:meta-running}
Algorithm~\ref{alg:max-flow} can be implemented in time% 
\[
{\rm TC}_{\rm time}(n,O(m),O(m_c)) + 
O\left(\Tsolv\left(O(m_c)\right)+{\rm TC}_{\rm arcs}(n,O(m),O(m_c))\right)+\tilde O(m)\,.
\]
\end{lemma}
\begin{proof}
By Lemma~\ref{lem:tot-cover-call}, the total number of arc additions is $O(m)$ and the total number of nodes queried in $\dsTcCover$ is $O(m_c)$. The bound on the transitive cover operations follows.

Let us bound the total number $\tilde m$ of arcs in the solver calls.
In iteration $\tau$, the solver $\approxFlow$ is called on a graph 
with arc set $\Eaux^{(\tau)}$ and $H^{(\tau)}$, as well as at most two shortcut arcs for each arc in $\Eaux^{(\tau)}$. 
Thus, $g\le 3\sum_{\tau=1}^T |\Eaux^{(\tau)}|+|H^{(\tau)}|$. 
By Lemma~\ref{lemma:meta-essential-tot}, we have $\sum_{\tau=1}^T \Eaux^{(\tau)}=O(m_c)$ and by definition, $\sum_{\tau=1}^T |H^{(\tau)}|={\rm TC}_{\rm arcs}(n,O(m),O(m_c))$. Thus, $\tilde m=O(m_c)+{\rm TC}_{\rm arcs}(n,O(m),O(m_c))$, and 
we obtain the second term in the bound on the total solver runtime.

Let us now consider all arithmetic operations in the algorithm other than the solver calls and the  the transitive cover data structure. Maintaining arc types, residual capacities and components takes $O(m+m_c\log n)=\tilde O(m)$ as described in Section~\ref{sec:abundant}. 

In every iteration, the other main operations are the following. We need to construct the auxiliary instance, add the flow $y$ to the previous, add new arcs to the instance, $\PostProcess(\cdot)$, and $\RemoveRoot(\cdot)$. These all  can be easily bounded as $O(|\Eaux^{(\tau)}|+|H^{(\tau)}|)$, which can be bounded as $O(m_c)+{\rm TC}_{\mathrm{arcs}}(n,O(m),O(m_c))$. Using the assumption  $\Tsolv(\tilde m)=\Omega(\tilde m)$, these operations  will be dominated by the calls to the solver.
\end{proof}

\subsection{Strong polynomiality of the algorithm}
We now complete the proof of Theorem~\ref{thm:maxflow-analysis-main} by proving part \eqref{part:bit-complexity}.
\begin{lemma}The algorithm is strongly polynomial. For rational input, all values of $f$, $\varepsilon$, and $r$ remain polynomially bounded in the input encoding length.
\end{lemma}
\begin{proof}
The flow $y$ returned by $\approxFlow{(\bar G,\bar u)}$ is a basic acylic flow, and hence all flow values can be written in the form $y_e=\sum_{e'\in E(G)} \chi_{ee'} u_{e'}$, where $\chi_{ee'}\in \{0,\pm 1\}$ (Lemma~\ref{lem:basic-flow}). For the auxiliary instances $(\bar G,\bar u)$, each $\bar u$ is either a residual capacity with respect to an original capacity value, or obtained from a safe capacity. By induction, we show that all safe capacities and flow values arise in the form $\bar f_e=\sum_{e'\in E(G)} \alpha_{ee'} u_{e'}$ where $\alpha_{ee'}$ is an integer with polynomially bounded encoding length.
Safe capacities are either $\res{f}_e$, or $\delta$ or $3\delta$. Here, $\delta=\delta_1+\delta_2$, where $\delta_1$ is $|E(\cG)|$ times the residual capacity of a certain arc, and $\delta_2$ is obtained
as $n^2$ times the residual capacity  $\cu^y_{e^\star}$ of an arc $e^\star$ output by $\approxFlow$.  
Thus, in every iteration, the flows and safe capacity values arise as integer combinations of the original capacities, adding at most $O(\log n)$ bits to each coefficient in every round.
\end{proof}

\section{Data Structures}\label{sec:data-structures}

In this section, we provide the \DITC{}s  we use to obtain our maximum flow algorithms. More precisely, we give different instantiations of the \DITC{} data structure (\Cref{def:general_transitive_cover}) introduced earlier in \Cref{sec:trans_cover_first_introduce}. 
In \Cref{sec:dsclo}, \Cref{sec:tree_depth},  \Cref{sec:fmm}, and \Cref{sec:dsord},
we provide different implementations of \DITC{}s with different guarantees. The \DITC{}s in \Cref{sec:tree_depth} and \Cref{sec:fmm} leverage additional information on where the queries and updates will be. In particular, \Cref{sec:tree_depth} assumes that arc updates between ancestors and descendants of a tree with bounded depth and \Cref{sec:fmm} assumes that arcs updates occur consistently with certain ordering of the nodes. 

The main theorems and definitions in this section are written in a self-contained way, so that the results of this section can be used just by referencing them. Below is a description of the only definitions and theorems that are used in Section~\ref{sec:algorithms-overall} to obtain the overall running times.
\begin{itemize}
    \item \Cref{sec:dsclo}: \Cref{thm:closure} gives our guarantees on a general \DITC{} that makes no additional assumptions. This facilitates our $O(mn)$ time maximum flow algorithms in \Cref{thm:general-flow}.
    
    \item \Cref{sec:tree_depth}: \Cref{cor:transitive_tree_depth} gives our guarantees on an \DITC{} for updates that we call \emph{tree-respecting} as defined in \Cref{def:dstree:tree_respecting} in the case that the relevant tree is depth $D$.  \Cref{cor:transitive_tree_depth} facilitates our $O(m^{1+o(1)}D)$ time maximum flow algorithm for bounded tree-depth graphs in \Cref{thm:tree-depth-overall}. 

    \item \Cref{sec:dsord}: In \Cref{def:dsord:ordered_transitive_cover}, we define \ODITC{}s which specialize the definition of a \DITC{} to when node ordering information is available as defined in \Cref{def:dsord:node_order}. In \Cref{cor:dsord:inc_closure} we give a simple, easy to use statement regarding our guarantees for an \ODITC{} which facilitates our  $O(n^{\omega-1}m_c)$ time maximum flow algorithm in \Cref{thm:ordered-overall}.

\end{itemize}
\Cref{cor:transitive_tree_depth} and \Cref{cor:dsord:inc_closure} follow from more general theorems of potential broader interest. \Cref{cor:transitive_tree_depth} follows from a more general bound in \Cref{thm:transitive_tree_depth} which bounds the performance of the $\DITC{}$ in terms of a notion of size defined in \Cref{def:dstree:size}. \Cref{cor:dsord:inc_closure} follows from a more general \Cref{thm:dsord:inc_closure}.

To facilitate our implementation of $\dsTcWit()$, each of our \DITC{}s will maintain a witness list with the particular property that they are \emph{rooted}, as defined below. This invariant essentially implies that the witnesses form in- and out-arborescence to specific nodes.

\begin{definition}[Rooted]
\label{def:rooted}
In a witness list $\wlistInst \in \wlistSpace$, we call $(a,b) \in \dE_*(\wlistInst)$, \emph{$a$-out rooted} if $(a,b)\in \dE(\wlistInst)$ or $(w,b) \in \dE(\wlistInst)$, and $(a,w)$ is $a$-out rooted for $w = \wit_\wlistInst(a,b)$. Similarly, we say that $(a,b) \in \dE_*(\wlistInst)$ is \emph{$b$-in rooted} if either $(a,b) \in \dE(\wlistInst)$, or $(a,w) \in \dE(\wlistInst)$ and $(w,b)$ is $b$-in rooted for $w = \wit_\wlistInst(a,b)$. We call $\wlistInst$ \emph{out-rooted} if every $(a,b) \in \dE_*(\wlistInst)$ is $a$-out rooted and we call $\wlistInst$ \emph{rooted} if every $(a,b) \in \dE_*(\wlistInst)$ is either $a$-out rooted or $b$-in rooted.
\end{definition}

We prove that any rooted witness list is path structured.

\begin{lemma}[Rooted Implies Path Structured]\label{lem:simp_paths} If a witness list is rooted then it is path structured. 
\end{lemma}

\begin{proof}
Let $(a,b) \in \dE_*(\wlistInst)$ for a rooted witness list $\wlistInst \in \wlistSpace$. Suppose, $(a,b)$ is $a$-out rooted and let $S = \{s \in \dV | (a,s) \in \dE_*(\wlistInst) \text{ and } (a,s) \text{ is } a\text{-out rooted}\}$. For each $s \in S$, let us define $e_s \defeq (a,s)$ if $(a,s) \in \dE(\wlistInst)$ and $e_s \defeq (w_s,s)$ otherwise, where $w_s =\wit_\wlistInst(a,s)$. Note that $\walk_\wlistInst(a,s) = e_s$ if $(a,s) \in \dE(\wlistInst)$ and $\walk_\wlistInst(a,s) = \walk_\wlistInst(a,w_s),e_s$ otherwise. By \Cref{lem:witlist_walk}, $\walk_\wlistInst(a,b)$ is a $a$--$b$ walk. Additionally, the arcs of this walk are a subset of $\{e_s | s \in S\}$. However, each node in $e_s$ has a distinct head and therefore any walk using the $e_s$ is simple, i.e., a path. By symmetry, if $(a,b)$ is $b$-in rooted the walk is simple. (In the symmetric argument the walk uses arcs with distinct tails.) \end{proof}

\subsection{Transitive closure}
\label{sec:dsclo}
\label{sec:dsclosure}

Here, we briefly show how to adapt the approach in Italiano's  incremental transitive closure data structure \cite{Italiano1986} to implement an \DITC{}.
Though the ideas of the algorithm largely follow from \cite{Italiano1986}, we provide a complete description and analysis to fit our particular definition of an \DITC{} and to facilitate our \DITC{} development in \Cref{sec:tree_depth}. Additionally, we make two more substantial modifications of \cite{Italiano1986}. First, in $\dsTcCover(S)$ rather than just outputting the closure of $S$, we output the smaller of the closure and all the arcs. Second, we maintain a witness list (\Cref{def:wit_list}).

 \Cref{alg:closure} presents the data structure. Whenever an arc $e$ is added in $\dsTcAdd(\cdot)$ the algorithm checks for each $a \in \dV$ if it is possible to use $e$ to reach a node that it could not reach before through a procedure, $\dstreeAddOut(a,e)$. If it can, it adds the arc to the closure and recursively checks whether $a$ can now use the arcs leaving the head of $e$ to reach a new node. This algorithm essentially maintains a tree from each $a$ to what nodes it can reach (and therefore the transitive closure) in total time linear in the size of all the trees. \Cref{alg:closure} essentially follows this strategy adding the arcs of the closure to a witness list. From the tree structure of these arcs, the witness list is out rooted. To output the cover of $S$, the algorithm simply outputs either the closure restricted to $S$ and the entire graph depending on $|S|$ and $|\dE|$. The guarantees of the algorithm are provided in \Cref{thm:closure}.

\begin{algorithm}[t!]
\caption{Transitive Closure}
\label{alg:closure}
\SetCommentSty{algCommentFont}
\SetKwProg{function}{function}{:}{}
\SetKw{Return}{return}

\tcp{Global variables }
    $\dV$, 
    $\dE \subseteq \dV \times \dV$, $\dE_* \subseteq \dV \times \dV$
        \tcp*{nodes, arcs, and transitive arcs} 
    $\wlistInst \in \wlistSpace$ \tcp*{witness list (\Cref{def:wit_list})}
\BlankLine
\lfunction{$\dsTcInit(\dV)$}{
    $\dV\gets \dV$, $\dE\gets\emptyset$, $\dE_* \gets \emptyset$, and
    $\wlistInst \gets \emptyset$
    \label{line:dsclo:init}
}

\BlankLine
\function{$\dsTcAdd(F \subseteq \dV \times \dV)$}{
    \For{$e = (u,v) \in F$\label{line:dsclo:for_add_start}}{
        $\dE \gets \dE \cup \{e\}$, $\dE_*\gets \dE_*\cup \{e\}$, and 
            $\wlistInst \gets \wlistInst , (u, v, v)$
            \label{line:dsclo:add-edge}\;
        \lFor{$a \in \dV$}{
            $\dstreeAddOut(a, e)$
        }
        \label{line:dsclo:for_add_end}
    }
}

\BlankLine
\function{$\dstreeAddOut(a \in \dV, (u,v) \in \dE)$}{
    \If{$(a,u) \in \dE_*$ and $(a,v) \notin \dE_*$}{
        $\dE_*\gets \dE_*\cup \{(a,v)\}$ and 
        $\wlistInst \gets \wlistInst , (a, v, u)$
        \label{line:dsclo:out_transitive_add} \;
        \lFor{$(v,w) \in \partial_\dE^\out(v)$}{
            $\dstreeAddOut(a, (v,w))$
            \label{line:dsclo:out_recurse}
        }
    }
}

\BlankLine
\function{$\dsTcCover(S\subseteq \dV)$}{
    \leIf{$|S|^2 \leq |\dE|$}{\Return $(S,\dE_*[S])$}{\Return $(\dV,\dE)$}
}

\BlankLine
\lfunction{$\dsTcWit()$}{
    \textbf{return} $\wlistInst$
}

\end{algorithm}

\begin{theorem}[Transitive Closure] 
\label{thm:closure}
Algorithm~\ref{alg:dstree} is an \DITC{} and can be implemented with the following properties, where $\Efinal$ is $\dE$ after the last $\dsTcAdd(\cdot)$:
\begin{itemize}

\item $\dsTcInit(\cdot)$ and $\dsTcAdd(\cdot)$ run in total time 
$O((n + |\Efinal|)n)$.

\item $\dsTcCover(S\subseteq \dV)$ runs in $O(\min\{|S|^2,|\dE|\})$ time and returns an $O(\min\{|S|^2,|\dE|\})$-arc cover. 

\item $\dsTcWit()$ returns a witness list of length $O(n^2)$ in time $O(n^2)$.%
\end{itemize}
\end{theorem}

\begin{proof}
\emph{Properties of $\dE$, $\dE_*$, and $\wlistInst$.}
First we establish several properties of $\dE$, $\dE_*$, and $\wlistInst$ that hold after the operations. Note that $\dE$ is only updated on \Cref{def:wit_list,line:dsclo:add-edge}. Additionally, $L$ is only updated on \Cref{line:dsclo:init,line:dsclo:add-edge,line:dsclo:out_transitive_add}; it is initialized to $\emptyset$ on \Cref{line:dsclo:init} and then elements are added on \Cref{line:dsclo:add-edge,line:dsclo:out_transitive_add}. On \Cref{line:dsclo:add-edge},  $(u,v,v)$ is added to $L$ thereby adding $(u,v)$ to $\dE(L)$ and on \Cref{line:dsclo:out_transitive_add}, $(a,v,u)$ is added where $v \neq u$ (since $(a,u) \in \dE_*$ and $(a,v) \notin \dE_*$ in the preceding line) and therefore this line does not add to $\dE(L)$. Consequently, after each $\dsTcAdd(\cdot)$, $\dsTcCover(\cdot)$, and $\dsTcWit()$, $\dE$ is the union of  all arcs added in $\dsTcAdd(\cdot)$ and $\dE = \dE(L)$.

Next, note that $\dstreeAddOut(a,(u,v))$ is only called for $(u,v) \in \dE$. Consequently, arcs are only added to $\dE_*$ and $\dELcl(L)$ when they are added to $\dE$ or on \Cref{line:dsclo:out_transitive_add}. When an arc $e = (a,v)$ is added to $\dE_*$ and $\dELcl(L)$  on \Cref{line:dsclo:out_transitive_add} it is the case that $(a,u) \in \dE_*$ and $(u,v) \in \dE$. Consequently, $L$ remains an out-rooted witness list after each addition to it. Since $\dsTcCover(\cdot)$ only outputs subsets of $\dE_*$ we see that $\dsTcWit()$ returns a witness list with the necessary properties of an \DITC{}.

We now show that before and after each call to $\dsTcAdd(\cdot)$ it is the case that $\dE_* = \dE^\tr$. By the reasoning in the preceding paragraph, $\dE_* \subseteq \dE^\tr$, and therefore it suffices to show that $e \in \dE^\tr$ implies $e \in \dE_*$. Proceed by contradiction and suppose that this is not the case after a call to $\dsTcAdd(\cdot)$ (it suffices to consider this case as after $\dsTcInit$, $\dE_* = \emptyset = \dE^\tr$). Pick $(a,b) \in \dE^{\tr} \setminus \dE_*$ such that the length of the shortest path, $P$, from $a$ to $b$ using $\dE$ is minimized. Note that the length of this path must be at least 2 since $\dE \subseteq \dE_*$. Let $c$ be the last node before $b$ on this $a$--$b$ path. By the assumptions there is an $a$--$c$ path using $\dE$ shorter than $P$ and therefore $(a,c) \in \dE_*$. This arc $(a,c)$ was added to $\dE_*$ on \Cref{line:dsclo:out_transitive_add}. However, since $(a,b) \notin \dE_*$ and $(c,b) \in \dE$ we see that $(a,b)$ would have been added during the recursive call to $\dstreeAddOut(\cdot)$ on \Cref{line:dstree:out_recurse} if $(c,b) \in \dE$ was present at that time. Consequently, $(c,b)$ must have been added to $\dE$ after $(a,c)$ was added to $\dE_*$. However, in that case $(a,b)$ would have been added to $\dE_*$ through a call to $\dstreeAddOut(a,(c,b))$ on \Cref{line:dsclo:for_add_end}, contradicting $(a,b) \notin \dE_*$.

\emph{Correctness.}
Leveraging the above properties of $\dE$, $\dE_*$, and $L$, we show that \Cref{alg:closure} is indeed an \DITC{} with the properties specified in the statement. $\dsTcCover(S)$ returns either $\dE_*[S] = \dE^\tr[S]$ by the invariants or $(\dV,\dE)$ both are covers of $S$ and have the specified number of arcs. Furthermore, we have already argued that $\dsTcWit()$ has the necessary properties for an \DITC{}. Additionally, since $|\dE_*| \leq n^2$ we see that the length of the witness list output by $\dsTcWit()$ is at most $n^2$.
 
\emph{Runtime.} Finally, we analyze the running time for implementing the algorithm. We use standard adjacency matrix data structures to maintain $\dE$ and $\dE_*$ so that $\dsTcInit(\cdot)$ is implementable in $O(n^2)$ and afterwards each of the operations involving these sets, e.g., adding an element and checking membership, can be performed in $O(1)$. This implies that $\dsTcCover(S)$ is implementable in $O(\min\{|S|^2,|\dE|\})$ time by maintaining $|\dE|$ in $O(1)$ per change to $\dE$ and when $|S|^2 \leq |\dE|$, checking if $(u,v) \in \dE_*$  for each $u,v\in S$ when $|S|^2 \leq |\dE|$, and outputting $(\dV,\dE)$ when this is not the case.
Additionally, we can store the witness list as a list so adding to it can be implemented in $O(1)$ time. Furthermore, since the list has at most $n^2$ elements, implementing $\dsTcWit()$ then can be done in $O(n^2)$ time. 

Next, note that \Cref{line:dsclo:add-edge} is called once for every arc in $\Efinal$, which can be done $O(|\Efinal|)$ time. Additionally, \Cref{line:dsclo:for_add_end} is called $n$ times for each $e \in \Efinal$ which can be done in $O(n|\Efinal|)$ time (ignoring the recursive cost). Furthermore, \Cref{line:dsclo:out_recurse} is called for a particular $a \in \dV$ and $(v,w) \in \dE$ only when $(a,v)$ was just added to $\dE_*$. Consequently, for a particular $(v,w) \in \dE$ \Cref{line:dsclo:out_recurse} is called at most once for each possible $a \in \dV$. Therefore, this line is implementable (ignoring the recursive cost) in $O(n |\Efinal|)$ time. Putting this altogether yields that $\dsTcInit(\cdot)$ and $\dsTcAdd(\cdot)$ are implementable in total time $O(n^2 + n|\Efinal|)$.
\end{proof}

\subsection{Tree-respecting updates}
\label{sec:tree_depth}

In this section we provide an \DITC{}, \Cref{alg:dstree}, in the special case where added arcs are only between what we call \emph{related} pairs of nodes with respect to a known rooted spanning tree, that is pairs of nodes that are either ancestors or descendants of each other. More formally, for a rooted spanning tree $T$ on nodes $\dV$, we define the following:
\begin{itemize}
    \item $\ancestor_{T}(a)$ denotes the set of ancestors of a node $a$ in $T$ (including $a$)
    \item $\descendant_{T}(a)$ denotes the set descendants of a node $a$ (including $a$).
    \item $\related_{T}(a) \defeq \ancestor_T(a) \cup \descendant_T(a)$ denotes what we call the set of nodes \emph{related} to $a$ in $T$. 
\end{itemize}
We design an \DITC{} where at initialization we are given a rooted spanning tree $T$ on $\dV$ and then all $F$ input in $\dsTcAdd(\cdot)$ \emph{respect $T$} as defined below. 

\begin{definition}[Tree Respecting]
\label{def:dstree:tree_respecting}
For rooted spanning tree $T$ on nodes $\dV$ we say that \emph{$(a,b) \in \dV \times \dV$ respects $T$} if $a \in \related_T(b)$ (or equivalently, $b \in \related_T(a)$). We let $\Etree \defeq \{(a,b) \in \dV \times \dV | a \in \related_T(b) \}$ denote the set of related node pairs and we say that \emph{$F \subseteq \dV \times \dV$ respects $T$} if $F \subseteq \Etree$.
\end{definition}

Given a rooted spanning tree $T$ and arc additions that respect $T$,  for every node $a$, our \DITC{} (\Cref{alg:dstree} to be introduced later) maintains arcs to each node $b$ that it can reach using $\dE[\descendant_T(a)]$ and arcs from each node $b$ that can reach it using $\dE[\descendant_T(a)]$. Equivalently, the \DITC{} maintains
\begin{equation}
\label{eq:dstree:Estar}
\dE_* 
    = \left\{ (a,b) \in \R^{\dV \times \dV} ~ | ~ (a, b) \in \dE[\descendant_T(a)]^\tr \text{ or } (a,b) \in \dE[\descendant_T(b)]^\tr \right\}\,.
\end{equation}
 To answer $\dsTcCover(S)$ queries we show that it then suffices to output the arcs of $\dE_*$ restricted to $\cup_{a \in S} \ancestor_T(a)$.
We show how to efficiently maintain $\dE_*$ and using this, obtain runtimes that on properties of the tree and its relationship to the edges added. Below we formally define a measure of the set of arcs relative to a rooted tree and with this provide our main \Cref{thm:transitive_tree_depth} regarding our \DITC{} for tree respecting updates. 

\begin{definition}[Relative Tree Size] \label{def:dstree:size}
For rooted spanning tree $T$ on nodes $\dV$ and $F \subseteq \dV \times \dV$ we define the \emph{size of $F$ relative to $T$} as 
\[
\treeComplexity_T(F) \defeq \sum_{(a,b) \in F} |\ancestor_T(a,b)|\, ,\, 
\text{ where }
\ancestor_T(u,v) \defeq \ancestor_T(u) \cap \ancestor_T(v)
\text{ for all }
u,v \in \dV
\,.
\]
\end{definition}

\begin{theorem}[Tree-respecting \DITC{}] 
\label{thm:transitive_tree_depth}
Let $T$ be a rooted spanning tree on $n$-element node set $\dV$.
Assume that for all calls $\dsTcAdd(F)$,  $F \subseteq \dV \times \dV$ respects $T$ and let  $\Efinal$ denote $\dE$ after the last $\dsTcAdd(\cdot)$. Then, Algorithm~\ref{alg:dstree} is an \DITC{},  and can be implemented with the following properties:
\begin{itemize}
\item $\dsTcInit(\cdot)$ and $\dsTcAdd(\cdot)$ has total runtime bounded as $O(\treeComplexity_T(\Efinal) + \sum_{a \in \dV} |\ancestor_T(a)|)$.

\item $\dsTcCover(S\subseteq \dV)$ has runtime, number of nodes in the graph output, and number of arcs in the graph output all bounded by $O(\sum_{a \in S} |\ancestor_T(a)|)$ and all arcs output respect $T$.

\item $\dsTcWit()$ returns a witness list of length $O(\treeComplexity_T(\Efinal))$ in time $O(\treeComplexity_T(\Efinal))$.
\end{itemize}
\end{theorem}

Before describing \Cref{alg:dstree} and proving \Cref{thm:transitive_tree_depth}, we first provide helpful tools for applying and reasoning about \Cref{thm:transitive_tree_depth}. First we provide \Cref{lem:dstree:size_bound} which provides upper bounds on the size of an arcs subset relative to a tree. Then we present and prove \Cref{cor:transitive_tree_depth} using \Cref{thm:transitive_tree_depth} and \Cref{lem:dstree:size_bound}. \Cref{thm:transitive_tree_depth} is a specialization of \Cref{thm:transitive_tree_depth} that assumes that $T$ has depth at most $D$ and gives running time in terms of $D$; we use it to prove \Cref{thm:tree-depth-overall}

\begin{lemma}[Relative Tree Size Bound] 
\label{lem:dstree:size_bound}
$\treeComplexity_T(F) \leq \min\{D |F|, 2 \sum_{a \in \dV} |\ancestor_T(a)|^2\}$ for any depth $D$ rooted spanning tree $T$ on $n$ nodes $\dV$ and $F \subseteq \dV \times \dV$ that respects $T$. %
\end{lemma}

\begin{proof}
$\treeComplexity_T(F) \leq D |F|$ since $|\ancestor_T(a,b)| \leq D$. Additionally,
\[
\treeComplexity_T(F) = \left| S \right| 
\text{ where } S \defeq \left\{
(a,b,c) \in \dV 
~ | ~
(a,b) \in F \text{ and } c \in \ancestor_T(a,b)
\right\}
\]
By assumption,  $a \in \related_T(b)$ for every $(a,b) \in F$. Therefore $a \in \ancestor_T(b)$ or $b \in \ancestor_T(a)$ and 
\[
S \subseteq T_1 \cup T_2 
\text { where } T_1 = \{(a,b,c) \in \dV | \{b,c\} \in \ancestor_T(a)\}
\text { and } T_2 = \{(a,b,c) \in \dV | \{a,c\} \in \ancestor_T(b)\}\,.
\]
The result follows as $|T_1| = |T_2| = \sum_{a \in \dV} |\ancestor_T(a)|^2$.
\end{proof}

\begin{corollary}
\label{cor:transitive_tree_depth} 
Let $T$ be a rooted depth-$D$ spanning tree on $n$-element node set $\dV$.
Assume that for all calls $\dsTcAdd(F)$, $F \subseteq \dV \times \dV$ respects $T$ and let  $\Efinal$ denote $\dE$ after the last $\dsTcAdd(\cdot)$; in particular, $|\Efinal| = O(nD)$.
Then, Algorithm~\ref{alg:dstree} is an \DITC{} and can be implemented with the following properties:
\begin{itemize}
\item $\dsTcInit(\cdot)$ and $\dsTcAdd(\cdot)$ run in total time $O((n +  |\Efinal|)D)$.
\item $\dsTcCover(S\subseteq \dV)$ has runtime, number of nodes in the graph output, and number of arcs in the graph output all bounded by $O(|S|D)$ and all arcs output respect $T$.

\item $\dsTcWit()$ returns a witness list of length $O(D |\Efinal|)$ in time $O(D |\Efinal|)$.
\end{itemize}
\end{corollary}

\begin{proof}
The result follows from \Cref{thm:transitive_tree_depth} by using that, $|\ancestor_T(a)| \leq D$ for all $a \in \dV$ since $T$ has depth-$D$, using \Cref{lem:dstree:size_bound} to bound $\treeComplexity_T(\Efinal)$, and noting that $|\Efinal|=O(\sum_{a \in \dV} |\ancestor_T(a)|) = O(nD)$) since $\Efinal$ respects $T$.
\end{proof}

Though we only apply \Cref{cor:transitive_tree_depth} (rather than \Cref{thm:transitive_tree_depth}) in later sections, we provide the more general \Cref{thm:transitive_tree_depth} as it elucidates the structure of the problem and can be faster in certain settings. To see this, note that, depending on $T$, $\treeComplexity_T(F)$ can be much lower than the $O(D|\Efinal|)$ bound in \Cref{cor:transitive_tree_depth}. For example, suppose the root has $n - (D - 2)$ children and one of those children has a path of $D - 2$ children (i.e., each of its non-leaf descendants has exactly one child), and $F$ is all arcs with related endpoints, i.e., $|F|=O(n+D^2)$,  then $\treeComplexity_T(F) = O(n + D^3)$, while $O(|F|D)=O(nD +D^3)$.

\medskip

We now discuss \Cref{alg:dstree} and prove \Cref{thm:transitive_tree_depth} in multiple steps. The algorithm and analysis use the following additional notation on
common ancestor and common descendants:
\[
\canc_T(a,b) \defeq 
\begin{cases}
a & \text{if }a \in \ancestor_T(b) \\
b & \text{if }b \in \ancestor_T(a)
\end{cases}
\text{ and }
\cdesc_T(a,b) \defeq 
\begin{cases}
a & \text{if }a \in \descendant_T(b) \\
b & \text{if }b \in \descendant_T(a)
\end{cases}
\,.
\]

\begin{algorithm}[t!]
\caption{Bounded tree-depth}
\label{alg:dstree}
\SetCommentSty{algCommentFont}
\SetKwProg{function}{function}{:}{}
\SetKw{Return}{return}

\tcp{Global variables }
    $\dV$, $\dE \subseteq \dV \times \dV$, $\dE_* \subseteq \dV \times \dV$
        \tcp*{nodes, arcs, and transitive arcs} 
    depth $D$ rooted spanning $T$ of $\dV$ 
        \tcp*{invariant: $a\in\related_{T}(b)$ for all $(a,b)\in \dE \cup \dE_*$ }
    $\wlistInst \in \wlistSpace$ \tcp*{witness list (\Cref{def:wit_list})} 
    $\Edesc(a) \subseteq  \dV \times \dV$ for all $a \in \dV$
        \tcp*{invariant: $\Edesc(a) = E[\descendant_T(a)]$}
    $\Ecov(a) \subseteq  \dV \times \dV$   for all $a \in \dV$
        \tcp*{invariant: $\Ecov(a) = \{ (b,c) \in \dE_* | a = \cdesc_T(b,c) \}$}
\BlankLine
\function{$\dsTcInit(V)$}{
    $\dV \gets \dV$, $\dE\gets\emptyset$, $\dE_* \gets \emptyset$, 
    $\wlistInst \gets \emptyset$, and $\Edesc(a) \gets \emptyset$ for all $a \in \dV$
    \label{line:dstree:init}
    \;
}

\BlankLine
\function{$\dsTcAdd(F \subseteq \dV\times \dV)$}{
    \For{$e = (u,v) \in F$\label{line:dstree:for_add_start}}{
        $\dE\gets E\cup \{e\}$, $\dE_*\gets \dE_*\cup \{e\}$, and 
            $\wlistInst \gets \wlistInst , (u, v, v)$
            \label{line:dstree:add-edge}\;
        $\Ecov(\cdesc_T(u,v)) \gets \Ecov(\cdesc_T(u,v)) \cup \{(u,v)\}$
            \label{line:dstree:cov-update-add}
            \;
        \lFor{$a \in \ancestor_T(u,v)$}{
            $\dE_\descendant(a) \gets \dE_\descendant(a) \cup \{e\}$, 
            $\dstreeAddOut(a, e)$, and
            $\dstreeAddIn(a,e)$
            \label{line:dstree:add-inout}
            \label{line:dstree:for_add_end}
        }
    }
}

\BlankLine
\function{$\dstreeAddOut(a \in \dV, (u,v) \in \Edesc(a))$}{
    \If{$(a,u) \in \dE_*$ and $(a,v) \notin \dE_*$}{
        $\dE_*\gets \dE_*\cup \{e\}$, $\Ecov(v) \gets \Ecov(v) \cup \{e\}$, and 
        $\wlistInst \gets \wlistInst , (a, v, u)$
        \label{line:dstree:out_transitive_add} \;
        \lFor{$(v,w) \in \Edesc(a)$}{
            $\dstreeAddOut(a \in \dV, (v,w))$
            \label{line:dstree:out_recurse}
        }
    }
}

\BlankLine
\function{$\dstreeAddIn(a \in \dV, (u,v) \in \Edesc(a))$}{
    \If{$(v, a) \in \dE_*$ and $(u,a) \notin \dE_*$}{
        $\dE_*\gets \dE_*\cup \{e\}$, $\Ecov(u) \gets \Ecov(u) \cup \{e\}$, and 
        $\wlistInst \gets \wlistInst , (u, a, v)$
         \label{line:dstree:in_transitive_add} \;
        \lFor{$(w,u) \in \Edesc(a)$}{
            $\dstreeAddIn(a \in \dV, (w,u))$
            \label{line:dstree:in_recurse}
        }
    }
}

\BlankLine
\function{$\dsTcCover(S\subseteq \dV)$}{
    \Return $(S_\out , E_\out)$ where $S_\out = \cup_{a \in S} \ancestor_T(a)$ and $\dE_\out = \cup_{a \in S} \Ecov(a)$
}

\BlankLine
\lfunction{$\dsTcWit()$}{
    \textbf{return} $\wlistInst$
}

\end{algorithm}

\Cref{alg:dstree} is closely related to the transitive closure algorithm we presented in \Cref{sec:dsclosure} (which in turn is turn closely related to Italiano's  incremental transitive closure data structure \cite{Italiano1986}). However, rather than maintaining transitive arcs corresponding to an out-tree from every node, instead we maintain transitive arcs corresponding to in-trees and out-trees for restricted set of nodes. More precisely, as arcs are added, for every $a \in \dV$ we maintain $\Edesc(a) = \dE[\descendant_T(a)]$ and add to $\dE_*$ every $(a,b) \in \dE[\descendant_T(a)]^\tr$ and every $(b,a) \in \dE[\descendant_T(a)]^\tr$. Just as maintaining an out-tree from $a$ was implemented in $O(|\Efinal|)$ time in \Cref{sec:dsclosure}, similarly maintaining the in- and out-tree for a single $a \in \dV$ can be implemented in essentially, $O(|\Efinal[\descendant_T(a)]|)$ time.
This gives the desired time complexity for process $\dsTcAdd(\cdot)$ by the following simple lemma.

\begin{lemma}
\label{lem:dstree:complexity_relation}
If $T$ is a spanning tree of $\dV$ and $F \subseteq \dV \times \dV$ then $\sum_{a \in \dV} |F[\descendant_T(a)]| = \treeComplexity_T(F)$.
\end{lemma}

\begin{proof}
$
\sum_{a \in \dV} |F[\descendant_T(a)]|
= \sum_{a \in \dV} \left[ \sum_{(b,c) \in F | a \in \ancestor_T(b,c)} 1 \right]
= \sum_{(b,c) \in F} \left[ \sum_{a \in \ancestor_T(b,c)} 1 \right]
= \treeComplexity_T(F)$.
\end{proof}

Consequently, we can efficiently maintain $\dE_*$ as given in \eqref{eq:dstree:Estar}. However, it is perhaps not immediately clear why maintaining this $\dE_*$ yields an efficient \DITC{} data structure. Note that \Cref{lem:dstree:complexity_relation} and the algorithm for maintaining $\dE_*$  do not use the assumption that each arc of $\dE$ is tree-respecting. However, in the following \Cref{lem:dstree:path-through-ancestor} we show that when $\dE$ is tree-respecting, then if $u$ can reach $v$ using $\dE$ then there is a path that goes through some $a \in \ancestor_T(u,v)$ and only uses arcs in $\dE[\descendant_T(a)]$.

\begin{lemma}[Paths Through Common Ancestor]
\label{lem:dstree:path-through-ancestor}
If $T$ is a spanning tree of $\dV$ and $P$ is $u$--$v$ path using arcs in $\Etree$, then $P$ passes through some $a \in \ancestor_T(u,v)$ with $(u,a),(a,v) \in \dE[\descendant_T(a)]^\tr$.
\end{lemma}

\begin{proof}
Let $a$ be a node in $P$ of least depth (i.e., closest to the root). We claim that every node $b$ on $P$ is in $\descendant_T(a)$. Proceed by contradiction and assume there is a node $b\notin\descendant_T(a)$, and let $b$ be the one nearest to $a$, i.e., all nodes on $P$ between $a$ and $b$ are in $\descendant_T(a)$. Thus, there is an arc $(b,w)$ or $(w,b)$ in $P$ with $w\in\descendant_T(a)$. By the tree-respecting property, either  $b\in\descendant_T(w)$ or  $w\in\descendant_T(b)$. The first case gives a contradiction because $\descendant_T(w)\subseteq\descendant_T(a)$. In the second case, since both $b$ and $a$ are ancestors of $w$, either $b\in\descendant_T(a)$, a contradiction to the choice of $b$, or $a\in\descendant_T(b)$, a contradiction to the choice of $a$.
Consequently, the parts of $P$ between $u$ and $a$ and between $a$ and $v$ are both in $\dE[\descendant_T(a)]$, and thus the claim follows.
\end{proof}

By leveraging \Cref{lem:dstree:path-through-ancestor}, there is a simple approach to answer $\dsTcCover(S)$ queries. The lemma implies that we output every node that is an ancestor of $S$ every arc in $\dE_*$ between a node of $S$ and its ancestor (in either direction) then the result is a transitive cover. \Cref{alg:dstree} gives an efficient implementation of this strategy. It proceeds similarly to \Cref{alg:closure} with a few modifications. First, it carefully maintains $\Edesc(a)$ for all $a$ as described calls a generalization of $\dstreeAddOut(\cdot)$ and a variant of it for in-trees, $\dstreeAddIn(\cdot)$ within $\dsTcAdd(\cdot)$. Additionally, it maintains $\Ecov(a)$ corresponding to the arcs of $\dE_*$ between $a$ and an ancestor of $a$.

In the remainder of this section, we formally prove the main invariants that \Cref{alg:dstree} maintains in the setting of \Cref{thm:transitive_tree_depth} and then put everything together to prove \Cref{thm:transitive_tree_depth}.

\begin{lemma}[Invariants]
\label{lem:dstree:invariants}
In the setting of \Cref{thm:transitive_tree_depth}, after $\dsTcInit$ and each call to $\dsTcAdd(\cdot)$, $\dsTcCover(\cdot)$, and $\dsTcWit()$, the following invariants hold:
\begin{enumerate}
    \item $\dE = \dE(L)$ is the union of all arcs input to $\dsTcAdd(\cdot)$ and $\Edesc(a) = \dE[\descendant_T(a)]$ for all $a \in \dV$.\label{item:dstree:E}
    
    \item $\dE_* = \dELcl(L)$ and $\Ecov(a) = \{ (b,c) \in \dE_* | a \in \cdesc_T(b,c) \}$ for all $a \in \dV$.\label{item:dstree:witness} 
       
    \item In every addition of $(x,y,w)$ to $L$, the resulting $L$ is rooted, $(x,y) \in \dE_*$, and $(x,y)$ is out-rooted in $L$ if $x \in \canc_T(x,y)$ and in-rooted in $L$ if $y \in \canc_T(x,y)$.     \label{item:dstree:rooted}

    \item $\dE_*$ satisfies \eqref{eq:dstree:Estar}, i.e., $\dE_* = 
    \{ (a,b) \in \R^{V \times V} ~ | ~ b \in \dE[\descendant_T(a)]^\tr \text{ or } a \in  \dE[\descendant_T(b)]^\tr \}$ \label{item:dstree:Estar}
\end{enumerate}
\end{lemma}

\begin{proof}
After $\dsTcInit(\cdot)$, $\dE = \dE_* = \Edesc(a) = \Ecov(a) = L = \emptyset$ for all $a \in \dV$. Therefore the invariants hold trivially. Additionally, the only modifications to $\dE$, $\dE_*$, $\Edesc(a)$, $\Ecov(a)$, and $L$ occur through calls to  $\dsTcAdd(\cdot)$. Consequently, to prove the claim it suffices to show that if items \ref{item:dstree:E}-\ref{item:dstree:Estar} hold before $\dsTcAdd(\cdot)$ then they hold after.

Next, note that \Cref{line:dstree:add-edge} is the only line where $\dE$ and $\Edesc(\cdot)$ are modified and an element $(x,y,w)$ with $y = w$ is added to $L$ (since $F$ contains no self-loops). Additionally, \Cref{line:dstree:add-inout} is the only line in which $\Edesc(\cdot)$ is updated. Invariant \ref{item:dstree:E} follows straightforwardly from the nature of these lines and \Cref{line:dstree:for_add_start}. Additionally, since such modifications happen one at a time (\Cref{line:dstree:for_add_start} and \Cref{line:dstree:add-edge}), we see that it suffices to prove that invariants~\ref{item:dstree:witness}-\ref{item:dstree:Estar} hold after an iteration of the for loop (\Cref{line:dstree:for_add_start} to \Cref{line:dstree:for_add_end}) provided that they hold before.

Now we show that invariant~\ref{item:dstree:witness} holds. After $\dsTcInit(\cdot)$, $\dE$ and $\dE_*$ are only updated on Line~ \ref{line:dstree:add-edge}, \ref{line:dstree:out_transitive_add}, and \ref{line:dstree:in_transitive_add} which are also the only lines in which $L$ is modified. By the nature of these lines $\dE_* = \dELcl(L)$. Additionally, after $\dsTcInit(\cdot)$, each $(x,y)$ added to $\dE_*$ is added to a single $\Ecov(a)$ in one of three cases:
\begin{itemize}
    \item In \Cref{line:dstree:cov-update-add} it is directly added to $\Ecov(\cdesc(x,y))$.
    \item In \Cref{line:dstree:out_transitive_add} it is added to $\Ecov(v)$ where $v \in \Edesc(a)$ and therefore added to $\Ecov(\cdesc(x,y))$.
    \item In \Cref{line:dstree:in_transitive_add} it is added to $\Edesc(\cdesc(x,y))$. (Similarly to the preceding case.)
\end{itemize}
Consequently, each $(x,y) \in \dE_*$ is in $\Ecov(\cdesc(x,y))$ and $\Ecov(a) = \{(b,c)\in \dE_* | a \in \cdesc_T(b,c)\}$. 

Next, we show that invariants~\ref{item:dstree:rooted} and
\ref{item:dstree:Estar} hold. %
  Note that arcs $(x,y)$ are added to $\dE_*$ only on \Cref{line:dstree:add-edge} (in which case they are in $\dE$ and in $\dE[\descendant_T(\canc_T(x,y)]$) or added along with the update $L \gets L , (x,y, w)$ with either $x = a$, $y = v$, and $w = u$ on \Cref{line:dstree:out_transitive_add} or $x = u$, $y = a$, and $w = v$ on \Cref{line:dstree:in_transitive_add}. In the former case, $(u,v) \in \Edesc(a)$ so $u,v \in \descendant_T(a)$ and $x = \canc_T(x,y)$ and the input witness $v$ satisfies $(a,v) \in \dE_*$ and $v \in \descendant_T(a)$ (so $a = \canc_T(a,v)$). Similarly, in the latter case, $(u,v) \in \Edesc(a)$ so $y \in \canc_T(x,y)$ and the input witness $v$ satisfies $(v,a) \in \dE_*$ and $v \in \descendant_T(a)$ (so $v = \canc_T(a,v)$). Consequently, the result follows by induction.

Finally, it remains to prove that if $(x,y) \in \dE[\descendant_T(x)]^\tr$ or $(x,y) \in \dE[\descendant_T(y)]^\tr$ then $(x,y) \in \dE_*$. We prove the case when $(x,y) \in \dE[\descendant_T(x)]^\tr$ as the other case follows similarly. Our proof is similar to the proof that $\dE_* = \dE^\tr$ in the proof of \Cref{thm:closure}. Proceed by contradiction and suppose that not the case. Among all such $(x,y) \in \dE[\descendant_T(x)]^\tr$ at that time let $(a,b)$ be the one for which the length of the shortest path, $P$, from $a$ to $b$ using $\dE[\descendant_T(a)]$ is minimized. Note that the length of this path must be at least 2 since $\dE \subseteq \dE_*$ by \Cref{line:dstree:add-edge}. Let $c$ be the last node before $b$ on this $a$--$b$ path. By the assumptions there is an $a$--$c$ path using $\dE[\descendant_T(a)]$ shorter than $P$ and therefore $(a,c) \in \dE_*$. However, since $c \in \descendant_T(a)$ we see that $(a,c)$ was added to $\dE_*$ on \Cref{line:dstree:out_transitive_add}. Since $(a,b) \notin \dE_*$ and $(b,c) \in \dE[\descendant_T(a)]$ we see that $(a,c)$ would have been added during the recursive call to $\dstreeAddOut(\cdot)$ on \Cref{line:dstree:out_recurse} if $(c,b)$ had been present at that time. Consequently, $(c,b)$ must have instead been added on \Cref{line:dstree:add-edge} after $(a,c)$ was added to $\dE_*$ in which case $(a,b)$ would been added during the call to $\dstreeAddOut(\cdot)$ on \Cref{line:dstree:add-inout}, contradicting $(a,b) \notin \dE_*$.
\end{proof}

\begin{proof}[Proof of \Cref{thm:transitive_tree_depth}]
First, we prove that \Cref{alg:dstree} is indeed an \DITC{}, i.e., that the output of $\dsTcCover(\cdot)$ and $\dsTcWit()$ follow the specification of \Cref{def:general_transitive_cover}.

\paragraph{$\dsTcCover(S)$ Correctness:} Each $\Edesc(a) \subseteq \dE_* \subseteq \dE^\tr$ by \Cref{lem:dstree:invariants} and therefore $\dE_\out \subseteq \dE^\tr$. Additionally, clearly $S \subseteq S_{\out}$. If $(u,v) \in \dE^\tr$ for $u,v \in S$ then by \Cref{lem:dstree:path-through-ancestor} there is $a \in \ancestor_T(u,v)$ with $(u,a),(a,v) \in \dE[\descendant_T(a)]^\tr$. Consequently, $(u,a) \in \Ecov(u) \subseteq E_\out$ and $(a,v) \in \Ecov(v) \subseteq E_\out$ and there is a $u$--$v$ path in $(S_\out, E_\out)$ as desired.

\paragraph{$\dsTcWit()$ Correctness:} By \Cref{lem:dstree:invariants} we know that after each $\dsTcAdd(\cdot)$, $\dE = \dE(L)$ and $\dELcl = \dELcl(L)$ and that $L$ is a rooted witness list. Additionally, note that $L$ is only modified in $\dsTcAdd(\cdot)$ and when it is modified elements are added. Furthermore, the output of $\dsTcCover()$ is always a subset of $\dE_*$ and $\dE$ is the union of all arcs input to $\dsTcAdd(\cdot)$. Putting these facts together implies that $\dsTcWit()$ meets the criteria for an \DITC{}. %

\paragraph{}
With the correctness established we now analyze each operation, describing its implementation, and proving its properties:

\paragraph{$\dsTcInit(\cdot)$ Properties.}
 We use $O(n)$ time to process $T$ so that for any $a \in \dV$ its parent and its depth (distance along $T$ to the root) can be computed in $O(1)$. Additionally, for each $a \in \dV$ we use an array of size $O(|\ancestor_T(a)|)$ to store which $(x,a)$ and $(a,x)$ for $x \in \ancestor_T(a)$ are in $\dE_*$. This takes $O(\sum_{a \in \ancestor_{T}(a)})$ time and allows adding to $\dE_*$ and checking membership in $\dE_*$ to be done in $O(1)$ time (by checking if they are related and checking the array of the descendent, since $\dE_*$ only includes arcs between related nodes by \Cref{lem:dstree:invariants}). 
 Finally, we store $L$ as a linked list so adding an element can be added in $O(1)$ and the list can be returned in time linear in its length. Consequently, $\dsTcInit(\cdot)$ can be implemented in $O(n)$ time. 

\paragraph{$\dsTcAdd(\cdot)$ Properties.}
Ignoring the cost of the calls to $\dstreeAddOut(\cdot)$ and $\dstreeAddIn(\cdot)$ each iteration of the for-loop (\Cref{line:dstree:for_add_start} to \Cref{line:dstree:for_add_end}) can be implemented in $O(\ancestor_T(a,b))$ time. This can be done by determining $\canc_T(u,v)$ in $O(1)$ time (by choosing the node with smaller depth) and then repeatedly following parents to compute the nodes on the path from $\canc_T(u,v)$ to the root to compute $\ancestor_T(u,v)$. Additionally, $F \setminus E$ can be computed in $O(|F|)$ time by checking membership in $\dE$ for each element of $F$. Since the for-loop on \Cref{line:dstree:for_add_start} to \Cref{line:dstree:for_add_end} considers each $e \in \dE$ once, the total cost of $\dsTcAdd(\cdot)$ operation (ignoring the cost of $\dstreeAddOut(\cdot)$ and $\dstreeAddIn(\cdot)$ calls) is $O(n + \treeComplexity_T(\Efinal))$.

Next, consider $\dstreeAddOut(\cdot)$. This operation is implementable in $O(1)$, unless $(a,u) \in \dE_*$ and $(a,v) \notin \dE_*$. If $(a,u) \in \dE_*$ then  a new $(a,v)$ arc is added to $\dE_*$ and the cost of this operation is $O(1 + |(v,w) \in \Edesc(a)|)$ (by storing $\Edesc(a)$ as an adjacency listed). 
Since $\Edesc(a) = E[\descendant_T(a)]$ by \Cref{lem:dstree:invariants}
The total time need for this operation, beyond the $O(1)$ for initial calls accounted for in the preceding paragraph, is $O(\cup_{v \in \Efinal[\descendant_T(a)]} |(v,w) \in \Efinal[\descendant_T(a)]|) = O(|\Efinal[\descendant_T(a)]|)$. Similarly, the total additional time needed to implement $\dstreeAddOut(\cdot)$ is also $O(|\Efinal[\descendant_T(a)]|)$. Consequently, the total extra runtime for $\dstreeAddOut(\cdot)$ and $\dstreeAddIn(\cdot)$ is $O(\sum_{a \in \dV} |\Efinal[\descendant_T(a)]|) = O(\treeComplexity_T(\Efinal))$ by \Cref{lem:dstree:complexity_relation}. 

\paragraph{$\dsTcCover(\cdot)$ Properties.} We can compute $S_{\out}$ simply by computing the path from each $a \in S$ to the root in $\sum_{a \in S} |\ancestor_T(a)|$ time and consequently $S_{\out}$ has $\sum_{a \in S} |\ancestor_T(a)|$ nodes. Additionally, by \Cref{lem:dstree:invariants} each $\Ecov(a) = \{ (b,c) \in \dE_* | a \in \cdesc_T(b,c) \}$ and therefore $|\Ecov(a)| \leq \ancestor_T(a)$ yielding the desired runtime and bound on the number of arcs output for this operation.

\paragraph{$\dsTcWit()$ Properties.} 
This operation can be implemented in time linear in the length of $L$ which is $O(|\dE_*|)$ since $|\dE_*| = |\dELcl(L)|$ and $L$ is a witness list by \Cref{lem:dstree:invariants}.However, by the preceding $\dsTcAdd(\cdot)$ analysis we see that $|\dE_*| \leq \treeComplexity_T(\Efinal)$ and therefore the length of $L$ and the running time for implementing this operation are $O(\treeComplexity_T(\Efinal))$.\footnote{Note that improved complexity bounds are obtainable in certain cases, for instance $|\dE_*| \leq \sum_{a \in \dV} |\ancestor_T(a)|$. However, given the runtime bounds for $\dsTcAdd(\cdot)$ and that $\dsTcWit()$ is called once in our applications, such improvements don't seem to immediately yield faster algorithms for maximum flow.}
\end{proof}

\subsection{Node orderings}
\label{sec:dsord}
\label{sec:ds:fmm}
\label{sec:fmm}

Here we provide a \DITC{} (\Cref{def:general_transitive_cover}) in the special case where the data structure is given access to an \emph{ordering} of the nodes. We call such a data structure an \emph{ordered incremental transitive cover data structure (\ODITC{})} and its operations are the same as in the case without the ordering except that (1) at initialization the data structure is given access to initial \emph{ordering values} $x \in \R^V$, where  higher values of $x_a$ correspond to earlier positions for $a \in \dV$ and (2) the ordering values can be modified via a new procedure, $\dsTcReorder(\cdot)$. We provide a data structure (\Cref{alg:dsord}) that uses fast matrix multiplication (FMM) to obtain runtimes that improve as arc additions, cover queries, and re-orderings occur to nodes earlier in the ordering. 

In this section, we provide our node ordering notation (\Cref{def:dsord:node_order}), define \ODITC{}s (\Cref{def:dsord:ordered_transitive_cover}), and state our main theorem about designing efficient \ODITC{}s (\Cref{thm:dsord:inc_closure}) and  \Cref{cor:dsord:inc_closure} that will be used to prove \Cref{thm:ordered-overall}. We then explain \Cref{alg:dsord}, which underlies \Cref{thm:dsord:inc_closure}, and prove this theorem over several steps.

\begin{definition}[Node Orderings]
\label{def:dsord:node_order}
 For $n$-element finite set $\dV$ and $x\in\R^{\dV}$, we let $v^{x}_{1},\ldots,v^{x}_{n}\in \dV$
denote the ordering of $\dV$ with $x_{v^x_{1}}\geq \cdots\geq x_{v^x_n}$ and ties broken consistently.\footnote{More formally, we assume that there is total ordering on the nodes and if $x_{v_i^x} = x_{v_j^i}$ with $v_i^x$ appearing before $v_j^x$ in that node ordering then then $i < j$.} We define,
\[
\dV_x(j)\defeq\{v^x_1,\ldots,v^x_j\} 
\text{ for all $j \in [n]$ and }
\dV_x(S) \defeq \{i\in \dV\, :\, x_i\ge \min_{j\in S} x_j\}
\text{ for all $S \subseteq \dV$}
\,.
\]
Equivalently, $\dV_x(S) =  \dV_x(i)$ for the smallest $i$ such that $S \subseteq \dV_x(i)$.
\end{definition}

\begin{restatable}[Ordered Incremental Transitive Cover]{definition}{defOrderedTransitiveCover}
\label{def:dsord:ordered_transitive_cover} 
A data structure is an \emph{ordered incremental transitive cover data structure (\ODITC{})} if it supports the same operations as an \DITC{} (\Cref{def:trans-cover}) with the following changes:
    \begin{itemize}
        \item $\dsTcInit(\dV,x\in\R^{\dV})$ replaces $\dsTcInit(\dV)$ in \Cref{def:trans-cover} and initializes the data structure on a finite set of $n$ nodes $\dV$ with initial ordering values $x\in\R^{\dV}$ and sets $\dE \gets \emptyset$.

        \item $\dsTcReorder(S \subseteq \dV, y \in \R^S)$: a new operation that sets $x_a \gets y_a$ for all $a \in S$ for $|S| > 1$.
    \end{itemize}
\end{restatable}

\begin{theorem}[FMM for \ODITC{}]
\label{thm:dsord:inc_closure}
For any $n$-element node set $\dV$, Algorithm~\ref{alg:fmm} is an ordered transitive cover data structure (\Cref{def:dsord:ordered_transitive_cover}) and can be implemented with the following properties:
\begin{itemize}

\item $\dsTcInit(\cdot)$ runs in amortized $O(n^2)$ time.

\item$\dsTcAdd(F)$ runs in amortized $\tilde{O}(|\dV_x(\dV(F))|^\omega)$ time.

\item $\dsTcCover(S\subseteq \dV)$ runs in amortized $O(|\dV_x(S)|^2)$ time and outputs a $O(|\dV_x(S)|)$-node graph.\footnote{Note that from this output the transitive closure of $S$ can be computed in an additional $O(|\dV_x(S)|^\omega)$ time.}

\item $\dsTcReorder(S \subseteq \dV, y \in \R^S)$ runs in amortized $\tilde{O}(|\dV_{x}(S)| n^{\omega -1})$ time.

\item $\dsTcWit(\cdot)$ runs in $O(n^2)$ time.
\end{itemize}
\end{theorem}

\begin{corollary}[Total Runtime for FMM \ODITC{}]
\label{cor:dsord:inc_closure}
Algorithm~\ref{alg:fmm} is an \ODITC{} (\Cref{def:dsord:ordered_transitive_cover}) and can be implemented so that $\dsTcInit(\cdot)$ followed by a sequence of $\dsTcAdd(\cdot)$, $\dsTcCover(\cdot)$, and $\dsTcReorder(\cdot)$ operations, and then a single call to  $\dsTcWit()$ run in $O(Y n^{\omega - 1} + n^2)$ time and ensure that the sum of the number of nodes in the sets output in $\dsTcCover(\cdot)$ is $O(Y)$, where $Y$ is the sum of values of $|\dV_x(\dV(F))|$ and $|\dV_x(S)|$ for each call to $\dsTcCover(\cdot)$, $\dsTcAdd(\cdot)$, and $\dsTcReorder(\cdot)$.
\end{corollary}

\begin{proof}
This follows from the guarantees of \Cref{thm:dsord:inc_closure}.
\end{proof}

Before proving \Cref{thm:dsord:inc_closure} we first give the following overview of the associated data structure (\Cref{alg:dsord}) and our analysis of it.

\paragraph{Data Structure Invariants.} 
For input $n$-element $\dV$, the data structure maintains all arcs added, $\dE$, as well as $\dE_*$ containing $\dE$ and some transitive arcs, i.e., $\dE \subseteq \dE_* \subseteq \dE^{\tr}$. Additionally, the data structure maintains node subsets $\dV_1 \subseteq \dV_2 \subseteq \cdots \subseteq \dV_K = \dV$ for $K \defeq \lceil \log_2 n \rceil$ of exponentially growing size; each $|\dV_k| = v_k \defeq \min\{2^k,n\}$. The $\dV_k$ and $\dE_*$ are periodically updated during $\dsTcAdd(\cdot)$ and $\dsTcReorder(\cdot)$ through a procedure, $\dsordRebuild(\cdot)$, that we will describe. 
Each time $\dV_k$ is updated, we set $\dV_k = \dV_x(v_k)$ and ensure that $\dE_*$ contains all transitive arcs between $\dV_k$ by adding (at least) $\dE^\tr[\dV_k]$ to $\dE_*$. $\dV_k$ is updated whenever an arc is added in $\dsTcAdd(\cdot)$ outside of $\dV_k$. Additionally, we keep a counter, $t_k$, of the sum of the values of $|\dV_x(S)|$ that have been the input to $\dsTcReorder(\cdot)$ since the last time $\dV_k$ was updated and update so that after each operation, $t_k \leq v_k / 2$. Consequently, $\dE^{\tr}[\dV_k] = \dE_*[\dV_k]^\tr$ and $\dV_x(\lfloor v_k/2 \rfloor) \subseteq \dV_x(v_k - t_k) \subseteq \dV_k$ after each operation. These key invariants will be proved in \Cref{lem:dsord:invariant}.

\paragraph{Add and Cover.} Maintaining these invariants facilitates a straightforward implementation of $\dsTcAdd(\cdot)$ and $\dsTcCover(\cdot)$. 
For $\dsTcAdd(F)$ it suffices to find the largest value $k_*$ for which $\dV(F)$ is not contained in $\dV_{k_*}$ and update $\dV_k$ for all $k \in [k_*]$ (since $\dV_k \subseteq \dV_{k + 1}$ for all $k \in [K - 1]$). For $\dsTcCover(\cdot)$, it suffices to find the smallest value of $k_*$ for which $S \subseteq \dV_{k_*}$ and then output $(\dV_{k_*}, \dE_*[\dV_k])$ (since $\dE^{\tr}[\dV_k] = \dE_*[\dV_k]^\tr$ by the invariants). In both cases, finding $k_*$ is straightforward to implement in $O(v_{k_*}^2)$ time by checking the $\dV_k$ in sequence. Additionally, since the $v_{k}$ increase geometrically and we update so that  $t_k \leq v_k / 2$, we show that $v_{k_*} = O(|\dV_x(\dV(F))|$ in the case of $\dsTcAdd(\cdot)$ and $v_{k_*} = O(|\dV_x(S)|)$ in the case of $\dsTcCover(\cdot)$ (see \Cref{lem:dsord:size} and \Cref{cor:dsord:set}). This implies the desired runtime for implementing $\dsTcCover(\cdot)$. To implement $\dsTcAdd(\cdot)$, the $\dV_i$ can be updated in $O(v_{k_*}^\omega)$ by using fast matrix multiplication to compute the transitive closure on $(\dV_{\ell_*}, \dE_*[\dV_{\ell_*}])$ for $\ell_* = \max\{k_* + 1, K\}$ (since $\dV_{x}(k_*)$ is contained in $\dV_{\ell_*}$ by the invariants).  

\paragraph{Reordering.} To implement $\dsTcReorder(S,y)$ we update the values of $x$, compute $\dV_{x}(S)$, and update the $t_k$ values. If the invariant that $t_k \leq v_{k} / 2$ no longer holds for some $k$, we update all $\dV_k$ up to the largest $k_*$ for which $t_{k_*} > v_{k_*} / 2$. Similarly to the analysis of $\dsTcAdd(\cdot)$, this updating can be done in $O(v_{k_*}^\omega) = O(|\dV_x(S)|^\omega)$ time. In both cases, procedure $\dsordRebuild(k_*)$ in \Cref{alg:dsord} is called to perform the updates. To analyze the total runtime of re-ordering, let $X = \sum_{i} |\dV_{x^{(i)}}(S^{(i)})|$ where at the start of the $i$-th call to $\dsTcReorder(\cdot)$, $S^{(i)}$ is the input $S$ and $x^{(i)}$ is the value of $x$. We show that the number of times that $k_* = k$ is at most $O(X/k)$ and consequently, these updates can be implemented in $O(\sum_{k\in [K]} (X/v_k) \cdot v_k^\omega) = O(X \sum_{k\in [K]} 2^{k(\omega - 1})) = O(X n^{\omega - 1})$ time in total.

\paragraph{Witness List.}
Following the above approach yields the desired runtimes for all operations except for $\dsTcWit()$. \Cref{alg:dsord}, departs from this description and performs more work to in order to implement $\dsTcWit()$ in $O(n^2)$ time. \Cref{alg:dsord} maintains an out-rooted witness list $L$ where $\dE(L)$ is every arc inserted in $\dsTcAdd(\cdot)$ and $\dELcl(L)$ is every transitive arc computed. This suffices to implement $\dsTcWit()$ by \Cref{lem:simp_paths} but we perform further modifications to maintain these invariants which incurs extra polylogarithmic factors in the runtime of $\dsTcAdd(\cdot)$ and $\dsTcReorder(\cdot)$. 

To maintain such a witness list $L$, rather than just updating the $\dV_k$ by computing a transitive closure, we use that it is also possible to compute \emph{reachability trees} from every node (what we call a \emph{reachability forest}) in an $n$-node graph in $\otilde(n^\omega)$ time (this is the source of the aforementioned polylogarithmic factors).

\begin{definition}[Reachability Trees]
\label{def:dsord:reachability}
We call $T \subseteq \dE$ a \emph{reachability tree from $s \in \dV$ in graph $\dG = (\dV,\dE)$} if $(\dV,T)$ is an arborescence rooted at $s$ and $s$ can reach $t \in \dV$ in $(\dV,T)$ if and only if $s$ can reach $t$ in $\dG$. We call $\{T_a\}_{a \in \dV}$ a \emph{reachability forest in $\dG$} if each $T_a$ is a reachability tree from $a$. 
\end{definition}
 
\begin{algorithm}[t!]
\caption{Transitive Cover with Node Ordering}
\label{alg:dsord}
\label{alg:fmm}

\SetCommentSty{algCommentFont}

\SetKwProg{fninit}{init(}{):}{}
\SetKwProg{fnupdate}{update(}{):}{}
\SetKwProg{fnclosure}{closure(}{):}{}
\SetKwProg{fnpath}{path(}{):}{}

\SetKwProg{function}{function}{:}{}

\tcp{Global variables}
$\dG=(\dV,\dE)$, $\dE_* \subseteq \dV \times \dV$, $x\in\R^{\dV}$
    \tcp*{current graph, transitive arcs, ordering}
$K = \lceil \log_2 n \rceil$ and $v_k = \min\{2^k,n\}$ for all $k \in [K]$
    \tcp*{subgraph sizes ($n = |\dV|$)}
$\dV_{1} \subseteq \dV_2 \subseteq \cdots \subseteq \dV_K = \dV$
    \tcp*{node sets with $|\dV_k| = v_k$}
$t_{1},\ldots,t_{K}$ 
    \tcp*{number of order changes since computed $\dG_k$}
$\wlistInst \in \wlistSpace$ \tcp*{witness list (\Cref{def:wit_list})}

\BlankLine
\function{$\dsTcInit(\dV,x\in\R^{\dV})$}{
    $\dV\gets V$, $\dE \gets \emptyset$, $\dE_* \gets \emptyset$, set global $x$ to $x$, $n = |\dV|$, and $\wlistInst \gets \emptyset$\;
    $K = \lceil \log_2 n \rceil$ and $v_k = \min\{2^k,n\}$ for all $k \in [K]$ and then $\dsordRebuild(K)$\;
}

\BlankLine
\function{$\dsTcAdd(F\subseteq \dV\times \dV)$}{
    \lFor{$e = (a,b)\in F \setminus E$}{$\wlistInst \gets \wlistInst , (a, b, b)$, $\dE \gets \dE \cup \{e\}$, and $\dE_* \gets \dE_* \cup \{e\}$ \label{line:dsord:add1}}
    $\dsordRebuild(k_*)$ for the largest $k_* \in [K]$ such that $\dV(F) \not \subseteq \dV_{k_*}$ if there is one\label{line:dsord:add2}\;
}

\BlankLine
\function{$\dsTcCover(S \subseteq \dV)$}{
    \textbf{return} $(\dV_{k_*}, \dE_*[\dV_{k_*}])$ for smallest
    $k_* \in [K]$ with $S \subseteq \dV_{k_*}$\;
}

\BlankLine
\lfunction{$\dsTcWit()$}{
    \textbf{return} $\wlistInst$
}

\BlankLine
\function{$\dsTcReorder(S \subseteq \dV, y \in \R^S)$}{
    $t_{k} \gets t_{k} + |\dV_x(S)|$ for all $k \in [K]$ and
    $x_a \gets y_a$ for all $a \in S$ \label{line:dsord:reorder1}\;
   $\dsordRebuild(k_*)$ for largest $k_* \in [K]$ with $t_{k_*} > v_{k_*} / 2$ if there is one \label{line:dsord:reorder2}\;
}

\BlankLine
\function{$\dsordRebuild(k_* \in [K])$}{
      
    $\dV_k \gets \dV_x(v_k)$ and $t_k \gets 0$ for all $k\in[k_*]$ 
    and   $\ell_* = \min\{k_* + 1, K\}$
    \;
    Compute reachability forest $\{T_a\}_{a \in \dV_{\ell_*}}$ in $\dG_{\loc} = (\dV_{\ell_*}, \dE_*[\dV_{\ell_*}])$ \tcp*{see \Cref{thm:dsord:reachability_compute}} 
	\For{$a \in \dV$ and then each $(w,b)$ in a traversal (see \Cref{foot:dsord:traversal}) of $T_a$ from $a$ \label{line:dsord:traverse}}{
            \lIf{$(a,b) \notin \dE_*$}{$\dsordNewTransitive(a,(w,b))$ \label{line:dsord:newtransitcall}}
    }
}

\BlankLine
\function(\tcp*[f]{invariant: $(a,w) \in \dE_*, (a,b) \notin \dE_*$}){$\dsordNewTransitive(a \in \dV,(w,b) \in \dE_*)$}{
    \lIf{$(w,b) \in \dE$}{
        $\wlistInst \gets \wlistInst , (a, b, w)$ and $\dE_* \gets \dE_* \cup (a,b)$
    \label{line:dsord:addimmediate}
    }
    \Else{
        $v = \wit_L(w,b)$ 
        \label{line:dsord:wit}
        \tcp*{invariant: 
        $(v,b) \in \dE$, $(w,b)$ is $w$-out rooted in $L$}
        \lIf{$(a,v) \notin \dE_*$}{$\dsordNewTransitive(a \in \dV,(w,v) \in \dE_*)$}
        $\wlistInst \gets \wlistInst , (a, b, v)$
        \label{line:dsord:addedge}
        \;
    }
}

\end{algorithm}

\begin{theorem}[Efficient Reachability Trees]
\label{thm:dsord:reachability_compute}
There is a deterministic algorithm that computes a reachability forest of any input $n$-node graph $\dG = (\dV,\dE)$ in $\tilde{O}(n^\omega)$ time.
\end{theorem}

\begin{proof}
In $O(n^{\omega})$ time it possible to compute $A \in \{0,1\}^{\dV \times \dV}$ such that $A_{a,b} = 1$ if and only if $a$ can reach $b$ in $\dG$ \cite{FM71}. Additionally, by \cite{AlonGMN92}, in $\tilde{O}(n^\omega)$ time it is possible to compute $w_{a,b} \in \dV$ for all $a,b \in \dV$ with $w_{a,b} = 1$ such that $\mathrm{path}(a,b)$ defined recursively as $(a,w_{a,b})$ concatenated with $\mathrm{path}(w_{a,b},b)$\footnote{The paper \cite{AlonGMN92} refers to the $w_{a,b}$ as witnesses, similar to our usage of the term in witness maintainers. Also note that \cite{AlonGMN92} showed this for $\mathrm{path}(a,b) = \mathrm{path}  (a,w_{a,b}) (w_{a,b},b)$. However, the claimed results can be obtained by applying \cite{AlonGMN92} to the graph with the arcs reversed and appropriately reversing the result.} for distinct $a,b \in \dV$ with $w_{a,b} = 1$ and $\mathrm{path}(a,a) = \emptyset$ is an $a$--$b$ path. Note that $T_{a} = \{(w_{a,b},b) | b \in \dV , w_{a,b} = 1\}$ contains exactly one arc to every node that $a$ can reach in $\dG$ and, by the properties of $\mathrm{path}$, $a$ can reach every node it can reach in $\dG$ using the arcs in $T_a$. Consequently, $T_a$ is a reachability tree for $a$ in $\dG$ and we can obtain an algorithm with the desired properties by simply computing $A$ and the $w$ and then outputting the $T_a$ (in an additional $O(n^2)$ time).
\end{proof}

Reachability trees give a natural way to find witnesses. Perform a traversal\footnote{\label{foot:dsord:traversal}We use the term ``traversal'' to refer to any procedure for exploring the arcs of the tree that from a single starting node only follows arcs from the starting node or endpoints of arcs that have been followed. Traversals include breadth first search, depth first search, and more.} of a reachability tree $T$ from $a$ starting from $a$, every time we traverse $(w,b) \in \dE$ and $(a,b)\notin \dE_*$ add $w$ as the witness for a transitive $(a,b)$ arc, i.e., add $(a,b,w)$ to $L$. If all additions to $L$ occur this way, $(a,b)$ will be $a$-out rooted in $L$ (\Cref{def:rooted}) since $(w,b) \in \dE$ and $(a,w)$ was added before traversing $(w,b)$.

In \Cref{alg:dsord} we apply a similar approach. Instead of just computing the transitive closure of $(\dV_{\ell_*}, \dE_*[\dV_{\ell_*}])$ during $\dsordRebuild$ (as described earlier) we compute a reachability forest, $\{T_a\}_{a \in \dV_{\ell_*}}$. We then perform an traversal of each $T_a$ from $a$ to compute all transitive arcs from $a$ along with witnesses. However, it is possible that the  $(w,b)$ arc traversed is actually a transitive arc in $\dE_*\setminus E$ 
(since we compute the reachability forest on $(\dV_{\ell_*}, \dE_*[\dV_{\ell_*}])$ and $\dE_{*}$ may contain transitive arcs). Whenever this occurs, by repeatedly computing preceding witnesses, we can find a path $b_j-b_{j-1}-\ldots-b_1$ in $\dE$ with $b_1=b$, $(b_{i+1},b_{i})\in \dE$ and $(a,b_i)\notin \dE_*$ for all $i\in [j-1]$, and $(a,b_j)\in \dE_*$.
In this case, we can then add the entire path and the associated transitive arcs to $\dE_*$, and preserve that $L$ is an out-rooted witness list. We provide a procedure $\dsordNewTransitive(\cdot)$ which does this. The overall cost of these operations can be bounded as $O(n^2)$, since there can be at most $O(n^2)$ transitive arcs and the time to implement is $O(1)$ per transitive arc.

\medskip

In the remainder of this section we prove \Cref{thm:dsord:inc_closure}. First, in \Cref{lem:dsord:newtransitive} we give a lemma about $\dsordNewTransitive(\cdot)$ which shows that it efficiently maintains $\dE_* = \dELcl(L)$, $\dE = \dE(L)$, and that $L$ is an out-rooted witness list. We use \Cref{lem:dsord:newtransitive} to prove \Cref{lem:dsord:invariant} which gives the invariants maintained by \Cref{alg:dsord}. We then give, technical \Cref{lem:dsord:size}, and its \Cref{cor:dsord:set}; we use these to analyze the runtime of \Cref{alg:dsord} and conclude with a proof of \Cref{thm:dsord:inc_closure}.

\begin{lemma}[$\dsordNewTransitive(\cdot)$ Invariants]
\label{lem:dsord:newtransitive}
Consider a call to $\dsordNewTransitive(a \in \dV, (w,b) \in \dE_*)$ 
on \Cref{line:dsord:newtransitcall} where  $(a,w) \in \dE_*$ and $(a,b) \notin \dE_*$. If at the start of the call $\dE_* = \dELcl(L)$, $\dE = \dE(L)$, and $L$ is out-rooted witness list then, these are still the case immediately after the call, and a set $F$ containing $(a,b)$ is added to $\dELcl(L)$ through exactly $|F| - 1$ recursive calls to $\dsordNewTransitive(\cdot)$.
\end{lemma}

\begin{proof}
If $(w,b) \in \dE$ the claim is trivial. Otherwise, the procedure computes $v \in \dV$ such that $(v,b) \in \dE$ (since $(a,b)$ is $a$-out rooted) and $(w,v) \in \dE_*$, recursively calls $\dsordNewTransitive(a,(w,v))$ if $(a,v) \notin \dE_*$, and then adds $(a,b)$ with witness $v$. This ultimately computes a path $b_j-b_{j-1}-\ldots-b_1$ in $\dE$ with $b_1=b$, $(b_{i+1},b_{i})\in \dE$ and $(a,b_i)\notin \dE_*$ for all $i\in [j-1]$, and $(a,b_j)\in \dE_*$. It then calls $\wInstance.\dsordNewTransitive((a,b_{i+1}),b_i)$ from $i = j - 1$ to $1$ over $j - 1$ recursive calls. Since each $(b_{i+1},b_i) \in \dE$, we see that $L$ remains an out-rooted witness list.
\end{proof}

\begin{lemma}[\ODITC{} Invariants]
\label{lem:dsord:invariant}
In \Cref{alg:dsord} after $\dsTcInit(\cdot)$ and each subsequent $\dsTcAdd(\cdot)$, $\dsTcCover(\cdot)$, $\dsTcWit()$, and $\dsTcReorder(\cdot)$, the following invariants hold:
\begin{itemize}
    \item \textbf{Graph:} $\dE$ is the union of all $F$ input to $\dsTcAdd(\cdot)$ and $\dV_1 \subseteq \dV_2 \subseteq \cdots \subseteq \dV_K = \dV$ where $|\dV_k| = v_k \defeq \min\{2^k,n\}$ for all $k \in K$ where $K = \lceil\log_2 n\rceil$. 
    \item \textbf{Closure:} $\dE_* \subseteq \dE^\tr$, $\dE^\tr[\dV_k] = \dE_*[\dV_k]^\tr$, $\dV_{x}(v_k - t_k) \subseteq \dV_k$, and $t_k \leq v_k / 2$ for all $k \in [K]$.
    \item \textbf{Witness:} $E = \dE(L)$, $\dE_* = \dELcl(L)$, $L$ is an out-rooted witness list, and $\dsordNewTransitive(a,(w,b))$ was called at most once for every $(a,b) \in \dV \times \dV$. 
\end{itemize}
\end{lemma}

\begin{proof}
We prove by induction. Note that the statement trivially holds after $\dsTcInit(\cdot)$ as after this operation $\dE = \emptyset$, $\dV_k = \dV_{x}(v_k)$, $\dE_*[\dV_k] = \emptyset$, $t_k = 0$, and the other variables are set accordingly. Now suppose the invariant holds after $\dsTcInit$ and before one of the $\dsTcAdd(\cdot)$, $\dsTcCover(\cdot)$, $\dsTcWit()$, and $\dsTcReorder(\cdot)$ operations. We prove that it holds after. Since $\dsTcCover(\cdot)$ and $\dsTcWit()$ do not change the relevant variables, it is suffices to prove this for $\dsTcAdd(\cdot)$ and $\dsTcReorder(\cdot)$.

To analyze, $\dsTcAdd(\cdot)$ and $\dsTcReorder(\cdot)$, note that each operation updates the variables to preserve the invariants and then call $\dsordRebuild(k_*)$ for the largest $k_*$ for which the invariants do not hold for $\dV_{k_*}$. More precisely, $\dsTcAdd(F)$ adds the arcs of $F$ to $\wInstance$, $\dE$, and $\dE_*$ and then calls $\dsordRebuild(k_*)$ for the largest $k_*$ where $F$ is not contained in $\dV_{k_*}$. For all $k \in [K]$ with $k > k_*$ all added arcs are within $\dV_k$ so $\dE^\tr[\dV_k] = \dE_*[\dV_k]^\tr$ is preserved. Similarly, $\dsTcReorder(\cdot)$ updates $t_k$ and $x$ in accordance with the input and then calls $\dsordRebuild(k_*)$ where for all $k \in [K]$ with $k > k_*$ it is the case that $t_k \leq v_k/2$. 

Consequently, it suffices to show that if all invariants hold, except for those involving $\dV_k$ with $k < k_*$, then all invariants hold after calling $\dsordRebuild(k_*)$. Consider such a call to $\dsordRebuild(k_*)$ and let $\ell_* = \max\{k_* +1, K\}$ and $\dG_\loc = (\dV_{\ell_*}, \dE_*[\dV_{\ell_*}])$. If $\ell_* = K$ then by the invariants ($\dE_*[\dV_k] = \dE_* \subseteq \dE^{\tr}$) the transitive closure of $\dG_\loc$ is the transitive closure of $\dG$. On the other hand, if $\ell_* < K$, then by the invariants $v_{\ell_*} - t_{\ell_*} \geq v_{\ell_*} / 2 = v_{k_*}$, $\dV_x(v_{k_*}) \subseteq \dV_x(v_{\ell_*} - t_{\ell_*}) \subseteq \dV_{\ell_*}$ and $\dE^{\tr}[\dV_{\ell_*}] = \dE_*[\dV_{\ell_*}]^\tr$. Consequently, in both cases, the transitive closure of $\dG_{\loc}$ contains the transitive closure of $\dG$ restricted to a superset of $\dV_k$ for each $k \in [k_*]$ and $\dV_k = \dV_x(v_k) \subseteq \dV_{\ell_*}$.

Due to how the $\dV_1,\ldots,V_{k_*}$ are set, it only remains to show that after the call to $\dsordRebuild(k_*)$, if $a$ can reach $b$ in $\dG_{\loc}$ then $(a,b) \in \dE_*$, $\dE_* \subseteq \dE^{\tr}$, and the witness invariants are maintained. However, $\dsordNewTransitive(a,(w,b))$ preserves invariants if $(a,b)\notin \dE_*$ and $(a,w),(w,b) \in \dE_*$  by \Cref{lem:dsord:newtransitive}. Additionally, when an arc is $(w,b)$ is traversed on \Cref{line:dsord:traverse}, provided $(a,w) \in \dE_*$ then afterwards $(a,b) \in \dE_*$ (either by \Cref{lem:dsord:newtransitive} and $\dsordNewTransitive(a,(w,b))$ or since it was before the traversal). Since the arcs are explored in a traversal, $(a,w) \in \dE_*$ always holds by induction. Finally, since if $a$ can reach $b$ in $\dG_{\loc}$ some $(w,b)$ arc will be traversed the result follows.
\end{proof}

\begin{lemma}[Set Sizes]
\label{lem:dsord:size}
In the setting \Cref{thm:dsord:inc_closure} after $\dsTcInit(\cdot)$ and each subsequent $\dsTcAdd(\cdot)$, $\dsTcCover(\cdot)$, $\dsTcWit()$ for any $i \in [n]$, the smallest $k_*$ with $\dV_{x}(i) \subseteq \dV_{k_*}$ satisfies $|\dV_{k_*}| \leq 4 i$. 
\end{lemma}

\begin{proof}
Let $\ell_* \defeq \lceil\log_2(i) + 1\rceil$ and note that $2^{\ell_*} \leq 2^{\log_2(i) + 2} = 4 i$. Consider two cases:
\begin{itemize}
    \item Case 1: $\ell_* > K$: In this case $n < 2^{\ell_*} \leq 4 i$. Consequently, $i \geq n/4$ and $|\dV_{k_*}| \leq n \leq 4i$. 
    \item Case 2: $\ell_* \leq K$: In this case  $\ell_* \in [K]$ and  $v_{\ell_*}/2 \geq 2^{\log_2(i) + 1}/2 \geq i$. Since 
    $v_{\ell_*} - t_{\ell_*} \geq v_{\ell_*} / 2$ and $\dV_x(v_{\ell_*} - t_{\ell_*}) \subseteq \dV_{\ell_*}$ by \Cref{lem:dsord:invariant}, $\dV_{x}(i) \subseteq \dV_x(\ell_*)$ and therefore $|\dV_{k_*}| \leq  |\dV_{\ell_*}| \leq v_{\ell_*} \leq 2^{\ell_*} \leq 4i$.
\end{itemize}
\end{proof}

\begin{corollary}[Set Compute]
\label{cor:dsord:set}
In the setting \Cref{thm:dsord:inc_closure}, given access to the $v_i^x$ in order, in a linked list, and input $S \subseteq \dV$ there is an algorithm that computes the minimum $k_* \in [K]$ with $S \subseteq \dV_{k_*}$ in $O(v_{k_*})$ time where $v_{k_*} = O(|\dV_x(S)|$. 
\end{corollary}

\begin{proof}
With $O(|S|)$ pre-processing, so that checking membership in $S$ can be done in $O(1)$ time, for any $k \in [K]$, in $O(v_k)$ time we can check if $S \subseteq \dV_k$. Consequently, we can compute $k_*$ in $O(|S| + \sum_{k \in [k_*]} v_k) = O(|S| + v_{k_*})$ time by checking if $S \subseteq \dV_k$ for increasing $k$. Since $|S| \leq |\dV_x(S)|$ and $v_{k_*} = O(\dV_{\ell_*})$ for the minimum $\ell_* \in [K]$ with $\dV_x(S) \subseteq \dV_{\ell_*}$, \Cref{lem:dsord:size} yields the result.
\end{proof}

We put everything together and show how to implement \Cref{alg:dsord} to prove \Cref{thm:dsord:inc_closure}.

\begin{proof}[Proof of \Cref{thm:dsord:inc_closure}]
By the invariants of \Cref{lem:dsord:invariant} we see that the output of $\dsTcCover(\cdot)$ is as desired and the size of the subgraph output follows from \Cref{cor:dsord:set}. Additionally, standard data structures allow us to add to $L$ in $O(1)$, compute $\wit_L(w,b)$ in \Cref{line:dsord:wit} in $O(1)$, and output $L$ in time linear in its length. Since after each operation $L$ is an out-rooted witness list with $\dE = \dE(L)$ and $\dE_* = \dELcl(L)$ by \Cref{lem:dsord:invariant} (and $\dE$ is the union of all arcs added in $\dsTcAdd(\cdot)$ and only arcs of $\dE_*$ are output in $\dsTcCover$), \Cref{lem:simp_paths} implies that the output $\dsTcWit(.)$ is as desired. It only remains to bound the runtime of implementing \Cref{alg:dsord}.

To efficiently implement \Cref{alg:dsord}, we store $\dE$ and $\dE_*$ as an adjacency matrix. We use a balanced binary search tree to ensure that each change to $x_v$ can be implemented in $O(\log n)$, it is possible to compute $i$ such $v^x_{i} = v$ in $O(\log n)$ for any $v \in \dV$, and so that the $v_{1}^x, \ldots, v_{n}^x$ can be traversed through a doubly linked list (computing next and previous nodes in $O(1)$). With these data structures, we show how to implement each operation in the desired running time.

\paragraph{$\dsTcInit(\cdot)$:} At the end of this operation each $\dE_i = \emptyset$ and each $\dV_i = \dV_x(i)$. This can all be computed in $O(\sum_{k \in [K]} v_k^2) = O(\sum_{k \in [K]} 2^{2k}) = O(2^{2K}) = O(n^2)$ time. The runtime of this operation is claimed as amortized to cover the cost of future calls to $\dsordNewTransitive(\cdot)$.

\paragraph{$\dsTcAdd(\cdot)$:} \Cref{line:dsord:add1} is straightforward to implement $O(|F|) = O(|\dV_x(\dV(F))|)$. To implement \Cref{line:dsord:add2}, first in $O(|F|)$ time we compute $S = \cup_{(a,b) \in F} \{a,b\}$. Since, $F = F[\dV_k]$ if and only if $S \subseteq \dV_k$, we can find $k_*$ by finding the minimum $k_*$ with $S \subseteq \dV_{k_*}$. Therefore, using by \Cref{cor:dsord:set}, $k_*$ can be computed in $O(v_{k_*})$ time where $v_{k_*} = O(|\dV_x(S)|) = O(|\dV_x(\dV(F))|)$. Below, we show that $\dsordRebuild(k_*)$ can be implemented in amortized $\otilde(v_{k_*}^\omega)$ time. Since, $O(v_{k_*}) = O(|\dV_x(\dV(F))|)$, $\dsTcAdd(\cdot)$, can be implemented in $\otilde(|\dV_x(\dV(F))|^\omega)$ time.

\paragraph{$\dsTcCover(\cdot)$:} By \Cref{cor:dsord:set}, computing $k_*$ can be done in $O(v_{k_*}) = O(|\dV_x(S)|)$ time. Since $\dE$ is stored as an adjacency matrix, computing $\dE[\dV_{k_*}]$ can be done in $O(v_{k_*}^2) = O(|\dV_x(S)|^2)$ time and therefore $\dsTcCover(\cdot)$ can be implemented in $O(|\dV_x(S)|^2$ time

\paragraph{$\dsTcWit()$:} Since $\dE_* = \dELcl(L)$ by \Cref{lem:dsord:invariant} and $|\dELcl(L)| \leq n^2$, $\dsTcWit()$ is implementable in $O(n^2)$. 

\paragraph{$\dsTcReorder(\cdot)$:}  Let $X = \sum_{i} |\dV_{x^{(i)}}(S^{(i)})|$ where at the start of the $i$-th call to $\dsTcReorder(\cdot)$, $S^{(i)}$ is the input $S$ and $x^{(i)}$ is the value of $x$. In $O((|\dV_{x^{(i)}}(S^{(i)})| + 1) \log n)$ time, each step of the $i$-th call of $\dsTcReorder(\cdot)$, except for, $\dsordRebuild(k_*)$, can be implemented. Therefore all $\dsTcReorder(\cdot)$ calls, ignoring $\dsordRebuild(\cdot)$ costs, can be implemented in $O(X \log n)$ total time.\footnote{Logarithmic factor improvements may be possible. For example, rather than explicitly keeping $K = O(\log n)$ different $t_k$, a single counter could be used to more quickly ensure that each $\dsordRebuild(k_*)$ is called at least every $v_{k_*} / 2$ steps. However, in our applications, the logarithmic factors are subsumed by the those induced by \Cref{thm:dsord:reachability_compute}.} Additionally, as we show below, $\dsordRebuild(k_*)$ can be implemented in $\otilde(v_{k_*}^\omega)$ time.  However, the number of times that $k_* = k$ is $O(X/v_k)$ for all $k \in [K]$ and therefore all $\dsordRebuild(k_*)$ calls can be implemented in
\[
\otilde\left(\sum_{k \in [K]} \frac{X}{2^k} \cdot (2^k)^\omega\right) = \otilde\left(X \sum_{k \in [K]} 2^{k(\omega-1)}\right) = \otilde\left(X 2^{K(\omega - 1)}\right)=\otilde\left(X n^{\omega - 1}\right)
\]
time in total. This is equivalent to the claimed amortized runtime.

\paragraph{$\dsordRebuild(\cdot)$:} Using \Cref{thm:dsord:reachability_compute} and that $\dE$ is stored as an adjacency matrix, $\dG_{\loc}$ and the ${T_a}_{a \in \dV_{\loc}}$ can be computed in $\otilde(v_{\ell_*}^\omega) = \otilde(v_{k_*}^\omega)$. The breadth-first-searches can then be implemented in $O(v_{\ell_*}^2) = O(v_{k_*}^2)$ and below we show that the total cost of all calls to $\dsordNewTransitive(\cdot)$ is $n^2$. Finally, by storing $\dG_{\loc,*}$ as an adjacency matrix computing $\dG_1,\ldots,\dG_{k_*}$ can be done in $O(\sum_{k \in [k_*]} v_{k}^2) = O(v_{k_*}^2)$ time. Consequently, $\dsordRebuild(\cdot)$ is implementable in amortized $\otilde(v_{k_*}^\omega)$ time. 

\paragraph{$\dsordNewTransitive(\cdot)$.} By \Cref{lem:dsord:invariant}, $\dsordNewTransitive(\cdot)$ is called $O(n^2)$ times and each call can be implemented in $O(1)$ (not counting the recursive call). Consequently, the total runtime of $\dsordNewTransitive(\cdot)$ is $O(n^2)$ and covered by the amortized runtime of $\dsTcInit(\cdot)$.

\end{proof}

\newcommand{\Etotal}{\mathcal{E}_{\rm total}}
\newcommand{\Tflow}{\mathcal{T}_{\rm flow}}
\newcommand{\Tfrwk}{\mathcal{T}_{\rm frwk}}
\newcommand{\Afmm}{\mathcal{A}_{\rm fmm}}
\newcommand{\Ait}{\mathcal{A}_{\rm it}}
\newcommand{\Atd}{\mathcal{A}_{\rm td}}

\section{Maximum Flow Algorithms}\label{sec:algorithms-overall}

In this section we give the running times we achieve for strongly polynomial maximum flow framework on a large variety of flow instances.
We obtain these by combining Theorem~\ref{thm:maxflow-analysis-main} with the bounds in Section~\ref{sec:data-structures} on the  incremental transitive cover data structures. Recall that  $\Tsolv(\tilde m)=\tilde m^{1+o(1)}$ by
Theorem~\ref{thm:apx-flow-solver-run}. %

In all cases below, we can make assumption \eqref{eq:i-s-t}, i.e., all nodes $i\in V\setminus\{s,t\}$ have $\infty$ capacity $(t,i)$ and $(i,s)$ arcs. We simply add new such arcs if they are not present, increasing $m$ by $n$. Consider the case when $G$ has tree-depth $D$, and is given by a representing depth-$D$ tree $T$. Then, we obtain a representing tree $T'$ of the new graph by moving $s$ and $t$ to the top of the tree. This representation has tree-depth at most $D+2$.

Finally, note that the maximum flow $f$ output by our algorithm in the extended graph uses some of the added $(t,i)$ and $(i,s)$ arcs. We can post-process $f$ using 
the link-cut trees data structure by Sleator and Tarjan \cite{Sleator83}. This enables reducing $f$ to an acyclic flow $f'$ by removing flow around  directed cycles from the support in $O(m\log n)$ time. This $f'$ will only use the arcs of the original input instance.

\subsection{General graphs}
In this section, we show how our framework yields an $O(nm)$ algorithm for arbitrary instance $(G,u)$, assuming the number of capacitated arcs is $m_c=O(n^{2-\varepsilon})$ for some $\varepsilon>0$.  Recall that for $m=\Omega(n\log n)$, King, Rao, and Tarjan~\cite{King1994} already achieved $O(nm)$, while 
Orlin's algorithm~\cite{Orlin2013} is $O(nm)$ for $m=O(n^{\frac{16}{15}-\varepsilon})$. Our algorithm achieves this running time for all but very dense graphs. More precisely, we could even have $m=O(n^2)$, assuming the number of capacitated arcs is less by a factor $n^\varepsilon$.

Further, we note that one could similarly obtain a randomized $O(nm)$ algorithm up to  $m_c=O(n^{2}/\log^c n)$ for some $c\ge 1$ using the randomized $\tilde O((m+n^{1.5})\log U)$ solver by van den Brand et al.~\cite{Brand2021}.%

\genflowthm*
\begin{proof}
We use Algorithm~\ref{alg:max-flow} with the transitive data structure as in Algorithm~\ref{alg:closure}.
From Theorem~\ref{thm:maxflow-analysis-main}\eqref{part:tot-data}, we get 
$\madd=O(m)$ and $\Efinal=O(m)$, and 
the total size of the node sets $S$ in the calls to $\dsTcCover(S)$ is $O(m_c)$. According to Theorem~\ref{thm:closure},  the total time of $\dsTcInit(\cdot)$, $\dsTcAdd(\cdot)$  is $O(nm)$, and $\dsTcWit(\cdot)$ creates a witness list of size $O(n^2)$ in time $O(n^2)$. With a running time bound $O(\min\{|S|^2,m\})$ in each call, we get an upper bound $O(m\cdot (m_c/\sqrt{m}))=O(m_c m^{1/2})$ on the calls to $\dsTcCover(.)$. By \Cref{lem:routing}, $\WitRoute(\cdot)$ takes $O(n^2)$ time. The total running time of the data structure operations is therefore  ${\rm TC}_{\mathrm{time}}(n,O(m),O(m_c))=O(nm)$.
We also get ${\rm TC}_{\mathrm{arcs}}(n,O(m),O(m_c))=O(m_c m^{1/2})$, leading to the claimed overall running time by Theorem~\ref{thm:maxflow-analysis-main}.
\end{proof}

\subsection{Bounded tree-depth}
\label{sec:tree-depth-proof}
Let us now consider the bounded tree-depth case, i.e., when the input graph $G$ has tree-depth $D$, given with a representing rooted depth-$D$ tree $T$. We use Algorithm~\ref{alg:dstree}, with the bounds in
Corollary~\ref{cor:transitive_tree_depth}.

\threedepththm*
\begin{proof}
We combine Corollary~\ref{cor:transitive_tree_depth} with the bounds in Theorem~\ref{thm:maxflow-analysis-main}\eqref{part:tot-data}. All arcs  in $\dsTcAdd(\cdot)$ respect the tree $T$. The total size of the node sets $S$ in the calls to $\dsTcCover(S)$ is $O(m)$, and $\dsTcWit(\cdot)$ returns a witness list of size $O(mD)$ in time $O(mD)$. By \Cref{lem:routing}, $\WitRoute(\cdot)$ takes $O(mD)$ time. The bounds in Corollary~\ref{cor:transitive_tree_depth} therefore gives ${\rm TC}_{\mathrm{time}}(n,O(m),O(m))=O(mD)$ for the running time of the transitive cover data structure. We also get that the total size of the outputs to $\dsTcCover(S)$ is ${\rm TC}_{\mathrm{arcs}}(n,O(m),O(m))=O(mD)$. Theorem~\ref{thm:maxflow-analysis-main} yields the desired bounds.
\end{proof}

\subsubsection{Application to layered graphs} \label{sec:layered-graphs}  A natural application of the bounded tree-depth setting (\Cref{thm:tree-depth-overall}) appears in the context of \emph{layered graphs}. Assume the node set of the graph $G=(V,E)$ is partitioned as $V=V_1\cup V_2\cup\ldots \cup V_\ell$, with $s\in V_1$ and $t\in V_\ell$, and such that for every arc $(i,j)\in E$ with $i\in V_p$ and $j\in V_q$, $|p-q|\le K$ for some $k\ge 1$. Assume further that $|V_p|\le \bar n$ for every $p\in [\ell]$. It is then easy to verify that the tree-depth is bounded as $O(K\bar n \log n)$.

A common example of layered graphs are \emph{time-expanded graphs} that appear in time-varying network optimization. Assume that the arcs of a graph also have transit times; moreover, the capacities, transit times (and possibly costs) may change in each time period. An early example of such a model is  by Gale in 1959 \cite{gale1959transient}; see \cite{sha2007time} for details and references. Let $G=(V,E)$ be the underlying $n$-node $m$-arc graph. The maximum-flow problem over a time horizon $T$ can be solved by creating $T+1$ layers, each with a copy of $V^{(r)}$ of $V$, and for each arc $e=(i,j)$ that has transit time $\tau(e,r)$ at time period $r\in T$, create an arc between the copy $i^{(r)}\in V^{(r)}$ and $j^{(r+\tau(e,r))}\in V^{(r+\tau(e,r))}$ if $r+\tau(e,r)\le T$. We also add arcs between subsequent copies of the same node. This network has $Tn$ nodes and $O(Tm)$ arcs, hence, the simple bound for a strongly-polynomial max flow algorithm is $O(T^2 nm)$. However, if all transit times are bounded as $\tau(e,r)\le K$, then the time expanded graph has tree-depth $O(Kn\log n)$, and thus \Cref{thm:tree-depth-overall} gives $\tilde O\left(K (Tm)^{1+o(1)}n \right)$. 

\subsection{Ordered transitive closure}
We now use Algorithm~\ref{alg:dsord}, i.e., transitive closure with node ordering. This maintains ordering values $x\in \R^V$.
Recall from Definition~\ref{def:dsord:node_order} that for $S\subseteq V$, $V_x(S)$ is the set of nodes $u$ such that $x_u\ge \min_{v\in S} x_v$.

 Recall that for $v\in\roots$, $P_v$ denotes the free component with root $i$. We use the ordering values $x\in\R^V$ defined as follows:
\begin{equation}\label{eq:def-node-ord}
x_v=\begin{cases}
\max\left\{\res{f}_e\,:\, e\in \dtot{}{P_v}\, ,\, \res{f}_e<\ab\varepsilon\right\}\, , & \mbox{ if }v\in\roots\, ,\\
0\, ,&\mbox{ otherwise}.
\end{cases}
\end{equation}
We can maintain the residual capacities 
\[\left\{\res{f}_e\,:\,  e\in \dtot{}{P_v}\, ,\,  \res{f}_e<\ab\varepsilon\right\}
\]
 for each $v\in\roots$ in a heap data structure, for example, using binomial heaps. Thus, each arc will be present in two or zero copies. In each iteration, these values may change for essential arcs, and we may add new entries for new arcs in $H$; arcs that become abundant are deleted from the heap. 
 Note that all these operations can be carried out in $O(\log n)$ amortized time.

 The value of $\res{f}_e$ may change at most $O(1)$ times for a given arc $e$. This is because the value may change while $e$ is essential, which is only for $O(1)$ iterations, or in the iteration when $e$ is added as a new extension arc.
 These updates take $O(m\log n)$ time altogether.

When merging free components $P_i$ and $P_j$, we need to merge 
 the corresponding heaps and delete the two copies of the arcs between these components. We can identify these as follows. For each component $P_i$, the corresponding heap contains all arcs incident to $P_i$. Recall that the root of the merged component will be the root of the larger component, or $s$ or $t$ if they are one of $i$ and $j$. Accordingly, when merging $i$ and $j$, assume $j$ ceases to be a root (that is, $|P_i|\ge |P_j|$ or $i\in\{s,t\}$). We scan through all arcs in the heap of $P_j$. For those that are between $P_i$ and $P_j$, we remove the other  copy from the heap of $P_i$; and all others are added to the heap of $P_i$. Thus, we need to process every arc at most $O(\log n)$ times. The overall running time of the heap operations will be $O(m\log^2 n)$.

In the algorithm, the set of essential roots $\Vess$ is defined as the roots of components with incident essential arcs. Equivalently,
\[
\Vess=\{v\in\roots\, :\, x_v\ge \ab^{-5}\varepsilon\}\, .
\]
The algorithm calls $\dsTcCover(\Vess)$, and modifies the flow on essential arcs as well as on new arcs added between nodes in $\Vess$.
 Consequently, $x_v$ is not changed for any $v\notin\Vess$. Therefore, to maintain the $x_v$ values in the data structure, it suffices to recompute their values on $\Vess$, and call $\dsTcReorder(\Vess,y)$, where $y$ is the set of new values, obtained from the  heap.
By definition of $\Vess$, $V_x(\Vess)=\Vess$.

\orderedthm*
\begin{proof}
We use Algorithm~\ref{alg:max-flow} with the transitive data structure as in Algorithm~\ref{alg:dsord}. We add a call to $\dsTcReorder(\Vess)$ at the end of the while cycle.
Thus, both $\dsTcCover(\cdot)$ and  $\dsTcReorder(\cdot)$ are called on  the essential root set $\Vess$ in every iteration, and $V_x(\Vess)=\Vess$.
By Theorem~\ref{thm:maxflow-analysis-main}\eqref{part:tot-essential}, the total number of essential roots throughout is $O(m_c)$.

As noted above, the time for the heap operations is $O(m\log^2 n)$.
In each iteration, $\dsTcAdd$ is called for an arc set on $\Vess$.
By Theorem~\ref{thm:dsord:inc_closure}, the total time of all operations $\dsTcInit$, $\dsTcCover$, $\dsTcReorder$, $\dsTcAdd$ (after the first call) can be bounded as $\tilde{O}(n^{\omega-1}m_c +m)=\tilde{O}(n^{\omega-1+o(1)}m_c)$ since $m_c\ge n$. 
We also get that $\dsTcWit(\cdot)$ returns a witness list of size $O(n^2)$  in $O(n^2)$ time. Since $\omega \geq 2$ and $m_c = \Omega(n)$, we also have $O(n^2) = O(n^{\omega-1}m_c)$. By \Cref{lem:routing}, $\WitRoute(\cdot)$ also takes time $O(n^2)$. Thus, we get 
${\rm TC}_{\mathrm{time}}(n,O(m),O(m_c))=\tilde{O}(n^{\omega-1}m_c)$.

To derive the overall bound in Theorem~\ref{thm:maxflow-analysis-main}\eqref{part:tot-running}, we will show that ${\rm TC}_{\mathrm{arcs}}(n,O(m),O(m_c)) = O(n m_c)$.  Let $S_1,\dots,S_k$ denote the sequence of sets on which the algorithm makes transitive cover calls. Note that each $S_i$ corresponds to the essential roots at some iteration, and hence is always a prefix in the order induced by $x$ at that iteration. Therefore, the call to $\dsTcCover(S_i)$ guarantees that the cover has at most $O(|S_i|^2)$ edges. Since the total number of essential roots $\sum_{i=1}^k |S_i| = m_c$ , we have that the total number of edges is $\sum_{i=1}^k |S_i|^2 \leq n \sum_{i=1}^k |S_i| = O(n m_c)$. In particular, the time spent on approximate flow calls is $(n m_c + m)^{1+o(1)} =  n m_c^{1+o(1)}$ since $m_c = \Omega(n)$ and $m \leq n^2$. The total running time is therefore
\[
\tilde{O}(n^{\omega-1}m_c + n m_c^{1+o(1)}).
\]
Note that if $\omega > 2$, then $\tilde{O}(n^{\omega-1} m_c)$ is the dominant term in the running time. 

We now quickly sketch how to reduce the running time to $\tilde{O}(n^{\omega-1}m_c)$ by instead using randomized $\tilde{O}(m + n^{1.5})$-time algorithm of~\cite{Brand2021} as our the approximate flow solver. Firstly, we note that the $\dsTcCover(S_i)$, $i \in [k]$, for $S_i$ as above, always return sparser option between the full graph and a subgraph of the transitive closure on $O(S_i)$-vertices. For simplicity, let us modify $\dsTcCover(S_i)$ to always return the latter option. In this case, the number nodes in the approximate flow call is $O(S_i)$ and the number of arcs is $O(|S_i|^2)$. In particular, the algorithm of~\cite{Brand2021} solves these instances in randomized $\tilde{O}(|S_i|^{1.5} + |S_i|^2) = \tilde{O}(|S_i|^2)$-time. Therefore, using the same calculation as for the total number of arcs, the time spent on approximate flow calls is $\tilde{O}(n^{1.5} + m + \sum_{i=1}^k |S_i|^2) = \tilde{O}(n^2+ n m_c) = \tilde{O}(n m_c) = \tilde{O}(n^{\omega-1}m_c)$. This completes the proof.
\end{proof}

\section*{Acknowledgements}

Thank you to anonymous reviewers for their helpful feedback and suggestions. Thank you to Jan van den Brand for helpful feedback, suggestions, and discussions in particular,  thank you for help with the literature on algebraic methods for dynamic transitive closure and related problems. Part of this work was completed while authors were participating in the 
2021 HIM program Discrete Optimization and while visiting the Simons Institute for the Theory of Computing.  Daniel Dadush was supported by European Research Council grant 805241-QIP. 
Aaron Sidford was supported in part by a Microsoft Research Faculty Fellowship, NSF CAREER Award CCF-1844855, NSF Grant CCF1955039, a PayPal research award, and a Sloan Research Fellowship.  L\'aszl\'o A. V\'egh's was affiliated with the London School of Economics and Political Science for part of this work, where he was supported by the European Research Council grant 757481-ScaleOpt.

\bibliographystyle{abbrv}
\bibliography{fast-flow}

\appendix

\section{The Approximate Flow Subroutine}\label{sec:approx_flow}

In this section, we derive Lemma~\ref{lem:basic-flow} and Theorem~\ref{thm:apx-flow-solver-run}.

We will use the fundamental flow decomposition theorem, which states that
any flow can be decomposed into path and cycle flows. A flow $f \in \R^E_+$ is an
\emph{$s$--$t$ path flow} if the support of $f$ is an $s$--$t$ path $P \subseteq E(G)$ and $f_e = f_{e'} > 0$ for all $e,e' \in P$. Similarly, $f$ is a \emph{cycle flow} if it is supported on
cycle $C \subseteq E(G)$ and $f_e = f_{e'} > 0$ for all $e,e' \in C$.  The following decomposition is well-known and easy to verify using a greedy construction.

\begin{theorem}
Let $(G,u)$ be a maximum flow instance with $m$ arcs,
and let $f\in \flowP(G,u)$.
Then we can write $f = \sum_{i \in [k_p]} f^{P_i} + \sum_{i \in [k_c]}
f^{C_i}$, such that $k_c + k_p \leq m$, $f^{P_i} \in \R^E_+$ is an $s$--$t$ path flow supported on an $s$--$t$ path $P_i \subseteq E(G)$ for all $ i \in [k_p]$, and $f^{C_i} \in \R^E_+$ is a cycle flow supported on a cycle $C_i
\subseteq E(G)$ for all $ i \in [k_c]$.
\label{thm:flow-decomp}
\end{theorem}

\basicflow*
\begin{proof}
Using the link-cut trees data structure by Sleator and Tarjan \cite{Sleator83}, every flow can be reduced to an acylic flow of the same value by removing flow around  directed cycles from the support in $O(m\log n)$ time. Using a slight modification of this data structure, see also  \cite{Tarjan91}
, there is an  $O(m\log n)$ time algorithm that turns a  flow $f$  into a basic feasible flow of the same or better value  by setting some arcs with $0<f_e<u_e$ to either of the two capacity bounds. Running these algorithms one after the other results in an acyclic basic flow, by noting that the second algorithm cannot create any new directed cycles in the support of the flow.

Properties {\em (i)} and {\em (ii)} follow because $\bar f$ is acyclic. 
 Thus, a flow decomposition (Theorem~\ref{thm:flow-decomp}) may only contain paths, and the total value carried along the paths is at most $\val{\bar f}$.
 For {\em (ii)}, notice that any arc entering $s$ or leaving $t$ with positive flow would be contained in a cycle in the flow decomposition. Property {\em (iii)} holds for any basic flow: if $C$ is a cycle of arcs strictly between the upper and lower capacity, then we can write $f$ as the convex combination of two flows $f- g$ and $f+ g$, where $g$ is a flow sending some positive amount of flow around such a cycle. Similarly, if $s$ and $t$ were in the same component of the forest, then we could set $g$ as positive flow supported on an $s$--$t$ path. Property {\em (iv)} is also true for any basic flow and can be easily derived from property {\em (iii)}, or by the total unimodularity of the incidence matrix of directed graphs.%
\end{proof}

\apxflowsolver*

The stated algorithm $\approxFlow(G,u,M)$ is shown on Algorithm~\ref{alg:approx-flow-solver}. It relies on the solver $\FastMaxFlow(G,u)$ guaranteed in 
 Theorem~\ref{thm:approx-flow-near-linear}. For an instance $(G,u)$ with $m$ arcs and integer capacities $u_e$ at most $U$, this subroutine returns an integer maximum flow in running time $m^{1+o(1)}\log U$. We also use the subroutine $\MaxCap(G,u)$ that computes the maximum capacity of an $s$--$t$ path in the instance $(G,u)$, and returns an arc $e^\star\in E(G)$ with the same capacity. This algorithm can be implemented in time $O(m+n \log n)$ in an instance with $n$ nodes and $m$ arcs, using Fibonacci heaps \cite{fredman1987fibonacci}. Finally, let $\ConvertBasic{(G,u,f)}$ denote the subroutine described in Lemma~\ref{lem:basic-flow} that converts a flow to an acyclic basic one.
 We will use the following simple lemma.
 \begin{lemma}\label{lem:maxcap}
 For a maximum flow instance $(G,u)$ with $m$ arcs, the maximum capacity of an $s$--$t$ path lies between $\nu(G,u)/m$ and $\nu(G,u)$.
\end{lemma}

\begin{algorithm}[htb!]
    \caption{$\approxFlow$}\label{alg:approx-flow-solver}
    \KwData{An instance $(G,u)$ with  $\nu(G,u)<\infty$,  and $M\in \Rnn$.}
    \KwResult{A flow $f$  in  $(G,u)$ and $e^\star\in E(G)$ as in Definition~\ref{def:approx-solver}.}
     $u_{\bar e}\gets \MaxCap{(G,u)}$ \;
      $\delta\gets u_{\bar e}/(m^2 M)$ \;
      \lFor{$e\in E(G)$}{
         $w_e\gets\min\left\{\left\lfloor\frac{u_e}{\delta}\right\rfloor, m^3 M\right\}$}
       $g\gets \delta \cdot \FastMaxFlow{(G,w)}$ \;       
     $f\gets\ConvertBasic{(G,u,g)}$ \;
    $e^\star\gets \MaxCap{(G,\res{f})}$ \;
     \Return{$(f,e^\star)$}  \;
\end{algorithm}

The algorithm computes the bottleneck arc $\bar e$ by calling $\MaxCap{(G,u)}$. It then sets $\delta=u_{\bar e}/(m^2 M)$ and rounds all capacities as $w_e=\min\{\lfloor u_e/\delta\rfloor,m^3 M\}$.  We  call $\FastMaxFlow(G,w)$ to obtain an integer maximum flow, and set $g$ as $\delta$ times this flow. Finally, we convert $g$ into an acylic basic flow $f$, and obtain the bottleneck arc $e^\star$.

Rounding the capacities to integer multiples of $\delta$ may decrease the instance value by at most $m\delta$.
Since $\nu(G,u)\le m u_{\bar e}$ (Lemma~\ref{lem:maxcap}), this decreases the instance value by at most $\nu(G,u)/(m M)$. Further, truncating capacities at 
$m^3 M \delta=m u_{\bar e}\ge \nu(G,u)$ does not change the maximum flow value of the instance. Consequently, $\val{g}\ge \nu(G,u)(1-1/(m M))$, and $\nu(G,\res{f})\le \nu(G,\res{g})\le \nu(G,u)/(m M)$. Again by Lemma~\ref{lem:maxcap}, for the bottleneck arc $e^\star$,
$\res{f}_{e^\star}\le m\val{f}\le \nu(G,u)/M$. Thus, $(f,e^\star)$ satisfy the requirements in Definition~\ref{def:approx-solver}.

\paragraph{Arithmetic operations} As described above, Algorithm~\ref{alg:approx-flow-solver} uses additions, subtractions, multiplications, integer divisions, and comparisons  (beyond the operations in the other subroutines). 
We show that all these can be replaced by using additions, subtractions, and comparisons only, while keeping the same overall running time bound. Note that given $\alpha$, expressions of the form $t\alpha$ for integer $t$ can be computed by $O(\log t)$ additions.
We do not need to explicitly compute $\delta$ to compute the capacities $w_e$. This is because $w_e=k$ for the integer value $k\in\{0,1,2,\ldots,Mm^3\}$ for the value $k$ where $ku_{\bar e}\le M m^2 u_e  < (k+1) u_{\bar e}$. Hence, these values can be computed using binary search in $O(\log (m+M))$ additions and comparisons.

Let $g'$ denote the output of $\FastMaxFlow{(G,w)}$; this is an integer flow. The algorithm would set $g=\delta g'$, and then convert this to a basic flow $f$. However, since $\delta=u_{\bar e}/(M m^2)$, this involves division by $M m^2$. To avoid this, we can set $g''\defeq u_{\bar e} g'=M m^2 g$; the value of this flow on every arc is an integer multiple of the original capacity $u_{\bar e}$.  Let $u''\defeq M m^2 u$; thus, $g''$ is a feasible flow in $(G,u'')$.
We can use $\ConvertBasic{(G,u'',g'')}$ to convert $g''$ to an acyclic basic flow $f''$. 
This solution is uniquely defined by two arc sets $E_0,E_1\subseteq E(G)$ and the conditions $f''_e=0$ for $e\in E_0$ and $f''_e=u''_e$ for $e\in E_1$; the rest of the arcs forms a forest, with $s$ and $t$ being in different components.
Using these same two arc sets and the forest, we can compute the basic solution $f$ in $(G,u)$  such that $f=f''/(M m^2)$, without using division.

\section{Incremental Transitive Cover Lower Bounds via OMv}
\label{sec:omv}

In this section, we optimize a known lower bound for incremental transitive closure and adapt it to the precise setting needed in \Cref{thm:fundamental} for incremental transitive cover.

\paragraph{\bf The OMv Conjecture.} As mentioned in the introduction, the lower bound is derived from online Boolean matrix multiplication conjecture (OMv) of Henzinger, Krinninger, Nanongkai and Saranurak~\cite{henzinger2015unifying}. The $\gamma$-OMv problem, for fixed $\gamma > 0$ is as follows. We are given a Boolean matrix $A$ of dimension $n_1\times n_2$, where $n_1 = \lfloor (n_2)^\gamma \rfloor$ which we have $\poly(n_1,n_2)$-time to preprocess. From here, a Boolean matrix $B = (b_1,\dots,b_{n_3})$ of dimension $n_2 \times n_3$ arrives online one column at a time. After column $b_k$ arrives, $k \geq 1$, the algorithm must output the Boolean product $Ab_k$ before any remaining vector arrives. The OMv conjecture states that for $\gamma$-OMv, no online algorithm as above can compute the online matrix product in sub $O(n_1 n_2 n_3)$-time, which corresponds to the trivial bound for computing the Boolean products one by one. By sub $O(n_1n_2n_3)$-time, the precise meaning is $O(n_1^{1-\eps}n_2n_3 + n_1n_2^{1-\epsilon}n_3 + n_1n_2n_3^{1-\epsilon})$-time, for any fixed $\eps > 0$.

\paragraph{\bf OMv based Incremental Transitive Closure Lower Bound.} In~\cite[Corollary 4.8]{henzinger1997faster}, Henzinger, Krinninger, Nanongkai and Saranurak prove that any algorithm for the incremental transitive closure problem that processes $O(m)$ updates and $O(n^2)$ queries in sub $O(mn)$-time would break OMv for $n_1=n_3=n$ and $n_2 = m/n$.

The reduction from OMv is as follows. Given $A$, $B$ of dimensions $n_1 \times n_2$ and $n_2 \times n_3$ respectively, one constructs a three layer graph as follows. The nodes are in layers $L_1,L_2,L_3$, where $|L_r|=n_r$ for $r \in [3]$. The node labelled $i$ in $L_1$ and $j$ in $L_2$ are connected by a directed arc $(i,j)$ iff $A_{ij} = 1$. The nodes are in $L_3$ are initially isolated. We begin by inserting the arcs associated with $A$ into the incremental transitive cover data structure (these can technically be counted as OMv preprocessing).

When the vector $b_k$ arrives, we add the arc $(j,k)$ between the node labelled $j$ in $L_2$ and the node labelled $k$ in $L_3$, for all $j$ such that $(b_k)_j = 1$. These corresponding arcs are inserted into the data structure.

To compute the Boolean matrix vector product $A b_k$, we note that $(A b_k)_i = 1$, $i \in [n_1]$, if and only if there is a directed path of length $2$ from the node labelled $i$ in $L_1$ to the node labelled $k$ in $L_3$ in the current network. This can be determined by asking the transitive closure data structure if $k$ is reachable from $i$ in the current network.

In the above reduction, the number of queries to the transitive closure data structure is precisely $n_1 \times n_3 = n^2$ (the dimensions of $A B$), and the number of insertions is precisely the number of non-zeros and $A$ and $B$, which are at most $n_1 \times n_2 + n_2 \times n_3 = 2 (m/n) (n) = 2m$. Furthermore, the number of nodes in the graph is $n_1+n_2+n_3 = 2n + m/n = O(n)$ (noting that $m \leq n^2$).  The OMv conjecture therefore posits that the total query and update time cannot be sub-$O(n_1n_2n_3) = O(nm)$.

\paragraph{\bf Optimizing the Lower Bound for Incremental Transitive Cover.} As mentioned in the introduction, one can run the above reduction using an incremental transitive cover data structure in place of incremental transitive closure. Namely, for each $(i,k)$ transitive query, one simply queries for a transitive cover on the set $\{i,k\}$, and runs a BFS from $i$ on this cover to see if it can reach $k$. Since the BFS-time is linear in the number of arcs of the outputted cover, and the transitive closure processing-time dominates the size of the outputted covers, the lower bound above extends to transitive cover under the condition that the total size of query sets is $O(n^2)$.

In \Cref{cor:conditional}, the number of arc insertions and the total query size are only $O(m)$, noting that the number of queries required above is $\Omega(n^2)$. We now show that optimizing the above OMv lower bound under these new conditions yields that the total processing time of incremental transitive cover cannot be sub-$O(m^{3/2})$. 

We recall that the upper bound on the number of inserted arcs $n_1n_2 + n_2n_3$ must be at most $O(m)$ and the number of nodes $n_1+n_2+n_3$ must be $O(n)$. Furthermore, the number of queries $n_1 n_3$ must also be $O(m)$. Under these conditions, optimizing the lower bound simply means maximizing the product $n_1 n_2 n_3$. By the AM-GM inequality, we have that
\begin{align*}
n_1 n_2 n_3 &= \left((n_1 n_2)^{\frac{3}{2}})\right)^{\frac{1}{3}} \left((n_2 n_3)^{\frac{3}{2}})\right)^{\frac{1}{3}}  \left((n_1 n_3)^{\frac{3}{2}})\right)^{\frac{1}{3}} \\ 
&\leq \frac{1}{3} \left( (n_1 n_2)^{3/2} + (n_2 n_3)^{3/2} + (n_1 n_2)^{3/2} \right) = O(m^{3/2}),
\end{align*}
using our assumption that each of $n_1n_2,n_2n_3,n_1n_3$ is $O(m)$. Furthermore, a feasible solution achieving this bound is simply $n_1=n_2=n_3=\sqrt{m} \leq n$.

\end{document}